\documentclass[12pt]{article}

\usepackage{amsmath,amsthm,amsfonts,amscd,latexsym,amssymb,amsbsy,dsfont,mathrsfs,color}
\usepackage[small, centerlast, it]{caption}
\usepackage{epsfig}
\usepackage[titles]{tocloft}
\usepackage[latin1]{inputenc}
\usepackage{verbatim} 
\usepackage[active]{srcltx} 
\usepackage{fancyhdr}
\usepackage{caption}
\usepackage{enumerate}
\usepackage{color}
\usepackage{hyperref}
\hypersetup{
    hidelinks
}

\definecolor{NiColor}{RGB}{77,77,255}
\definecolor{NiColoRed}{RGB}{255,77,77}
\definecolor{NiCitation}{RGB}{77,255,77}

\def\sp{\hskip -5pt}
\def\spa{\hskip -3pt}

\newsymbol\rest 1316    
\def\emptyset{\varnothing} 
\def\b1{{1\!\!1}}
\def\pperp{{\perp\!\!\!\!\perp}}

\def\cA{{\ca A}}
\def\cB{\mathscr{B}}

\def\cD{\mathscr{D}}
\def\cF{{\ca F}}

\def\cH{{\ca H}}

\def\cS{\mathscr{S}}

\def\cT{{\ca T}}

\def\cW{{\ca W}}
\def\cK{{\ca K}}

\def\sA{{\mathsf A}}

\def\sK{{\mathsf K}}

\def\sH{{\mathsf H}}

\def\sS{{\mathsf S}}

\def\sV{\mathsf{V}}
\def\sT{{\mathsf T}}

\def\bC{{\mathbb C}}           

\def\bN{{\mathbb N}}
\def\bM{{\mathbb M}}

\def\bR{{\mathbb R}}

\def\gA{{\mathfrak A}}       
\def\gB{{\mathfrak B}}

\def\gD{{\mathfrak D}}

\def\gF{{\mathfrak F}}
\def\gG{{\mathfrak G}}

\def\gM{{\mathfrak M}}

\def\gR{{\mathfrak R}}
\def\gS{{\mathfrak S}}

\def\gW{{\mathfrak W}}

\def\beq{\begin{eqnarray}}
\def\eeq{\end{eqnarray}}

\newcommand{\ca}[1]{{\cal #1}}         

\usepackage{sectsty}
\sectionfont{\normalsize}
\subsectionfont{\normalsize\normalfont\itshape}
\newtheoremstyle{TheoremStyle}
{3pt}
{3pt}
{\slshape}
{}
{\bf}
{:}
{.5em}
{}

\theoremstyle{TheoremStyle}
\newtheorem{theorem}{Theorem}[section]
\newtheorem{corollary}[theorem]{Corollary}
\newtheorem{proposition}[theorem]{Proposition}
\newtheorem{lemma}[theorem]{Lemma}
\newtheorem{definition}[theorem]{Definition}
\newtheorem{remark}[theorem]{Remark}
\newtheorem{notation}[theorem]{Notation}

\begin{document}


\hfill{\sl June 2026} 
\par 
\bigskip 
\par 
\rm


\par
\bigskip
\large
\noindent
{\bf  Spatial Localization of Relativistic Quantum Systems: The Com-
mutativity Requirement and the Locality Principle. \\Part II: A Model from Local QFT}
\bigskip
\par
\rm
\normalsize 


\noindent  {\bf Valter Moretti$^{a}$ }\\
\par

\noindent 
  Department of  Mathematics, University of Trento, and INFN-TIFPA \\
 via Sommarive 14, I-38123  Povo (Trento), Italy.\\
 $^a$valter.moretti@unitn.it, corresponding author\\

 \normalsize

\par

\rm\normalsize

\rm\normalsize


\par

\begin{abstract}
This paper is the second and final part of a two-part study. We construct a class of positive-energy relativistic spatial localization observables in Minkowski spacetime within the standard framework of quantum field theory, based on the stress--energy--momentum tensor smeared with suitable test functions. For each fixed timelike direction, the construction yields a family of positive operator-valued measures (POVMs) defined on spacelike hypersurfaces, which are well defined on each $n$-particle sector and satisfy a natural relativistic causality condition ruling out superluminal propagation of detection probabilities.
The proposed localization observables arise from local or quasi-local quantum-field-theoretic quantities, thereby providing a rigorous realization of previously heuristic constructions. In the one-particle sector, the scheme reduces to the observable introduced by the author in previous literature, and its first moment reproduces the Newton--Wigner position operator under suitable normalization and centering conditions.
Since the Reeh--Schlieder theorem implies that the normally ordered stress--energy--momentum tensor need not be positive on the full Fock space, we analyze the role of quantum energy inequalities and establish lower bounds that allow us to control deviations from positivity. This leads to the introduction of regularized families of operators, bounded from below, that approximate the localization effects.
We further construct conditional localization observables associated with finite laboratories by means of suitably modified local energy operators. In particular, by Haag duality, the resulting conditional POVMs are shown to belong to local von Neumann algebras and hence to commute when associated with causally separated regions, in agreement with the Araki--Haag--Kastler framework.
Our results provide a quantum-field-theoretic implementation of the idea that commutativity of localization observables should be recovered at the level of conditional measurements in spacetime regions of finite extent.\end{abstract}
\tableofcontents

\section{Introduction}
 \subsection{Localization, Causality, Commutativity, and QFT}
In \cite{Mor126}, we addressed the issue of commutativity of effects for POVMs describing the spatial localization of a quantum system in the rest space of a reference frame in Minkowski spacetime. The problem stems from the analysis of spatial localization by Halvorson and Clifton \cite{HC}, who showed that natural assumptions such as positive energy, additivity, etc., underlying any notion of localization are incompatible with the commutativity of the effects $A(\Delta)$, $A(\Delta')$ associated with spacelike separated regions $\Delta$, $\Delta'$, thus appearing to be in tension with the description of microcausality in the Araki--Haag--Kastler (AHK) framework of local quantum physics: $[A(\Delta), A(\Delta')] \neq 0$ in general.
Without entering into the details of specific localization observables, and drawing on an analysis of relativistic causality largely inspired by Busch and collaborators \cite{Buschloc1,Buschloc2,B07}, we showed in \cite{Mor126} that the requirement of commutativity is in fact not justified for observables describing the position of a relativistic quantum {\em particle}. Here, a particle is understood as a quantum system endowed with the propensity to assume a definite position when subjected to a complete detection procedure in which the entire rest space $\Sigma$ of an observer---more generally, a Cauchy surface---is filled with detectors. These detectors are represented by a POVM associated with the Borel sets $\Delta \in \cB(\Sigma)$ of the space: $\cB(\Sigma)\ni \Delta \mapsto A(\Delta)$. In this setting, there is no compelling reason to require commutativity of the effects associated with causally separated spatial regions because standard arguments based on no-signaling or relativistic consistency are not triggered. 
In that analysis, we worked within a framework {\em more basic} than the AHK one, without assuming that physically localized sets of observables carry any a priori algebraic structure. From the AHK perspective, the effects of the POVMs under consideration are therefore not elements of local operator algebras. When the AHK formalism is assumed, this conclusion can also be justified {\em a posteriori}, in particular in light of certain consequences of the Reeh--Schlieder theorem when one imposes that the localization probability of the vacuum state is zero.

On the other hand, in \cite{Mor126} we suggested that more realistic experimental situations can be considered, in which detectors are switched on only within a finite-size laboratory and one considers {\em conditional localization POVMs} $\cB(\Delta_0)\ni \Delta \mapsto B_{\Delta_0}(\Delta)$. {\em That is, these POVMs measure the probability of detecting a particle in a spatial subregion $\Delta$, given that it is detected somewhere in the finite rest space of the laboratory $\Delta_0 \supset \Delta$}. In these cases \cite{Mor126}, commutativity could in principle be recovered in the following sense: two effects $B_{\Delta_0}(\Delta)$ and $B_{\Delta'_0}(\Delta')$ associated with regions in causally separated laboratories with respective rest spaces $\Delta_0$, $\Delta'_0$ may commute:
\beq [B_{\Delta_0}(\Delta), B_{\Delta'_0}(\Delta')]=0\:.\label{Xero} \eeq
More precisely, they could be represented by elements of local operator algebras $B_{\Delta_0}(\Delta) \in \gW({\cal O})$, $B_{\Delta'_0}(\Delta') \in \gW({\cal O}')$ with $\Delta \subset {\cal O}$, $\Delta' \subset {\cal O}'$, in accordance with the AHK framework.

 Assume that a localization observable is given in terms of effects $A(\Delta)$ with $\Delta$ any Borel region of any complete rest space $\Sigma$, and the bounded $\Delta_0$ defines the spatially finite rest space of a laboratory.  Taking advantage of the so-called  {\em gentle measurement lemma}, it was suggested in \cite{Mor126} that the POVM normalized on $\Delta_0$  
\beq B_{\Delta_0}(\Delta) = V \frac{1}{\sqrt{A(\Delta_0)}} A(\Delta) \frac{1}{\sqrt{A(\Delta_0)}} V^\dagger \:, \quad \Delta \subset \Delta_0 \label{oneG}\eeq
where $V$ is a given unitary operator, has the properties of a conditional localization POVM  for states whose probability of finding the system in $\Delta_0$ (measured with the effect $VA(\Delta_0)V^\dagger$) is close to $1$.
This interpretation makes sense in the general case where the operators $A(\Delta)$ are positive and define a positive-operator-valued measure on $\Sigma$, but are not necessarily effects, i.e., bounded by $I$.

 The present paper has two main aims. First, we show that positive-energy  localization observables can be constructed within the standard formalism of QFT, in particular by exploiting the stress--energy--momentum tensor operator $:\spa \hat{T}_{\mu\nu}\spa :[f] = \int_{\bM} :\spa \hat{T}_{\mu\nu}\spa :(x) f(x) d^4x$ smeared with test functions $f \in \cD(\bM)$. 
This idea is consistent with a previous result in which a positive-energy localization observable was constructed {\em on the one-particle space} of a real scalar quantum field by restricting the stress--energy tensor to that space \cite{M23}, thereby placing an idea originally proposed by Terno on rigorous mathematical grounds and extending it \cite{Terno}. That construction was also shown \cite{M23,DM24} to satisfy a basic causality requirement due to Castrigiano \cite{Castrigiano2} who extended and generalized the causality constraints studied by Hegerfeldt in his celebrated works \cite{Hegerfeldt,Hegerfeldt2}.

For future reference, observe that, 
if $\bR^4 \ni x \mapsto U_x$ is the unitary representation of the translation group of $\bM^4$, then  
$U_x :\spa \hat{T}_{\mu\nu}\spa: [f] U_x^\dagger = :\spa \hat{T}_{\mu\nu}\spa: [f_x]$ where  
$f_x(y) := f(y-x)$. 
Concerning the generalization developed in this work, we consider an $n$-particle state $\Psi\in \cH^{(n)}$ with $n>0$. 
For every Borel set $\Delta\subset \Sigma$, where $\Sigma$ is a spacelike $3$-plane of Minkowski spacetime $\bM$, the effects $A_f^u(\Delta)$ of the notion of spatial localization we construct satisfy, where $u_\Sigma$ is the future-pointing unit vector normal to $\Sigma$,
$$ \langle \Psi| A_f^u(\Delta)\Psi\rangle =\lim_{\epsilon\to 0^+} \left\langle \Psi \left| \frac{1}{\sqrt{H^u+ \epsilon I}} \int_{\Delta} :\spa \hat{T}_{\mu\nu}\spa: [f_x^2] \: u^\mu u^\nu_\Sigma\: d\Sigma(x) \frac{1}{\sqrt{H^u+ \epsilon I}}\right.\Psi \right\rangle$$ 
for a choice of a future-directed unit timelike vector $u$ defining the frame of the detectors, and $f\in \cD_\bR(\bM)$ such that $\int_\bM f^2 d^4x =1$. Here, $H^u$ is the Hamiltonian operator in the $u$-direction, and the $\epsilon$-regularization is necessary since the Minkowski vacuum $\Omega$ lies in the kernel of this operator. In the identity above, the effect $A_f^u(\Delta)$ on the left-hand side is defined independently of the integral operator on the right-hand side, and the identity is valid when considering {\em expectation values}: it is generally false for generic off-diagonal matrix elements.

The family $A^u_f(\Delta)$ is a well-behaved positive-energy localization observable in each space $\cH^{(n)}$, in particular for one-particle states. Notably, the integral for $\Delta= \Sigma$ is normalized independently of the choice of $f\in \cD_\bR(\bM)$, provided $\int_\bM f^2 d^4x =1$. 
The causality requirement denoted by CC and already established in \cite{M23,DM24} also holds for the localization observables constructed here. We further analyze the physical meaning of the notion of localization introduced above, showing that it is associated with the center of $u$-energy of the field and that, for one-particle states, it reduces to the observable already introduced in \cite{M23}. In the one-particle case, one can say even more: if one takes $f(x^0,\vec{x}) = h'(x^0)h(\vec{x})$ with $f^2$ centered as specified later, then, in the non-relativistic large-mass limit, $h$ determines the precision of a von Neumann measurement scheme for detecting a single particle, whereas $h'$ can be chosen independently within the centered class (for instance even functions around the chosen time origin). 
 
A problem with the family of operators $ \frac{1}{\sqrt{H^u+ \epsilon I}} \int_{\Delta} :\spa \hat{T}_{\mu\nu}\spa: [f_x^2] \: u^\mu u^\nu_\Sigma\: d\Sigma(x) \frac{1}{\sqrt{H^u+ \epsilon I}}$ is that they do not define effects on the whole Fock space $\gF_s(\cH^{(1)})$, since the positivity condition 
$:\spa \hat{T}_{\mu\nu}\spa: [f] u_\mu v_\nu \geq 0$, for $u,v$ timelike and future-directed, 
fails as a consequence of the Reeh-Schlieder theorem (though they are bounded from below). This issue will be carefully analyzed from the perspective of quantum energy inequalities. Relying on known general results by Fewster and collaborators \cite{F2,FS08,F12}, we prove in particular that, given $u\in \sT_+$, for every $\eta>0$, it is possible to modify the temporal part $h'$ by enlarging its support in the smearing function $f$ in such a way that there is a constant $c^u_{\eta, f} >0$ such that $\left(:\spa \hat{T}_{\mu\nu}\spa :[f_x^2] - \eta g_{\mu\nu} I \right) u^\mu v^\nu \geq c^u_{\eta, f} |u\cdot v| I$ for every $v\in \sT_+$ and every $x\in \bM$. Also note that $J_\mu(x) := -\left(:\spa \hat{T}_{\mu\nu}\spa :[f_x^2]- \eta g_{\mu\nu} I \right) u^\mu$ is a conserved, causal, and future-directed (where it does not vanish) current. 
(The sign $-$ also in front of $\eta$ is due to the fact that we are adopting the signature $-,+,+,+$.)

Based on this result we pass to the second main goal of this work. The paper aims to construct conditional localization observables using a slightly modified (in order to remove negative energies) local energy operator, as indicated above.
Focusing on local-energy symmetric operators
$$\sH^u_{f, \eta}(\Delta) := \int_{\Delta} ( :\spa \hat{T}_{\mu\nu}\spa: [f_x^2] - \eta g_{\mu\nu} I) u^\mu u^\nu_\Sigma \:d\Sigma(x)$$
we use their closures (local Hamiltonians) $\overline{{\sH}}^u_{f,\eta}(\Delta):=\overline{{\sH}^u_{f,\eta}(\Delta)}$; these are the unique selfadjoint extensions as a consequence of a known result by Sanders \cite{Ko13} about essential selfadjointness of the renormalized stress--energy tensor in static spacetimes. Functional calculus for these selfadjoint closures produces a physically meaningful definition of local conditional POVMs of this type
$$B_{\Delta_0}(\Delta) = \overline{\frac{1}{\sqrt{\overline{{\sH}}^u_{f, \eta}(\Delta_0)}} \sH^u_{f, \eta}(\Delta) \frac{1}{\sqrt{\overline{{\sH}}^u_{f, \eta}(\Delta_0)}}}\:, \quad \Delta \subset \Delta_0\:. $$ At this point, 
{\em Haag duality}, together with a careful analysis of affiliation of the above  local Hamiltonians  with von Neumann algebras, eventually yields the desired commutativity result (\ref{Xero}) when $\Delta_0$ and $\Delta_0'$ are strongly causally separated.
More precisely, we prove that the operators $B_{\Delta_0}(\Delta)$ belong to local von Neumann algebras $\gW({\cal O})$ obtained from local Weyl algebras, in agreement with the AHK perspective.

We also prove a further result.
Defining positive operators
$$ \sA_{f,\epsilon,\eta}^u(\Delta) :=\frac{1}{\sqrt{H^u+ \epsilon I}} \int_{\Delta}  (:\spa \hat{T}_{\mu\nu}\spa: [f_x^2] - \eta g_{\mu\nu} I ) u^\mu u^\nu_\Sigma \: d\Sigma(x) \frac{1}{\sqrt{H^u+ \epsilon I}} $$
they allow one to approximate the localization effects $A^u_{f}(\Delta)$ with arbitrary precision in a laboratory based on the finite rest space $\Delta_0$, since
$$ |tr(\rho A^u_{f}(\Delta)) - \lim_{\epsilon\to 0^+} tr(\rho \sA^u_{f, \epsilon,\eta}(\Delta)) | \leq \eta \frac{|u\cdot u_\Sigma | |\Delta_0|}{m}$$
for every $n$-particle state $\rho$ with $n>0$, where $m$ is the mass of the particle, and $\Delta  \subset \Delta_0$ measurable.
We establish the following relation, in agreement with (\ref{oneG}),
$$B_{\Delta_0}(\Delta)^u_{f,\eta} = V^u_{f,\epsilon,\eta, \Delta_0} \frac{1}{\sqrt{\sA^u_{f, \epsilon,\eta}(\Delta_0)}} \sA^u_{f, \epsilon,\eta}(\Delta) \frac{1}{\sqrt{\sA^u_{f, \epsilon,\eta}(\Delta_0)}} V^{u\dagger}_{f,\epsilon,\eta,\Delta_0}\:,$$
valid for some unitaries $ V^u_{f,\epsilon,\eta,\Delta_0}$.

In the recent literature, there have been other analyses aimed at reconciling the notion of locality in QFT with the concept of spatial localization for relativistic quantum systems. We mention in particular two works written in a more physics-oriented style. One of them is the comparative study \cite{FC24}, where various aspects of the apparent violation of locality and causality are analyzed. The other is \cite{Tu}, where $(1+1)$-dimensional QFT is considered for both bosons and fermions, making use of the reduced density matrix formalism. In particular, it is shown there that scalar particles cannot be localized within any compact region.

The style of this work is deliberately elementary, pedagogical, and, as far as possible, self-contained with respect to free QFT in Minkowski spacetime. The relevant notions, including, in particular, some basic foundational aspects of the Araki--Haag--Kastler approach, such as the {\em Reeh-Schlieder property} and {\em Haag duality}, are introduced progressively. It is nevertheless assumed that the reader is familiar with the mathematical notions of $*$-algebra, $C^*$-algebra, von Neumann algebra, and general spectral theory. The broader goal is to bring the community working on quantum measurement theory closer to the community working on local quantum field theory.

\subsection{Structure of this work}

After a subsection devoted to listing the fundamental notions and notation used in this work, Section \ref{SEZloc}
provides a brief review of the notions introduced in \cite{Mor126} regarding localization observables in terms of POVMs and conditional localization observables.
Section \ref{QFTsec} introduces the basic notions of QFT in Minkowski spacetime for a free real massive scalar quantum field. Section \ref{SEC2agg} presents several crucial concepts from mathematical physics, such as the smeared normally ordered stress-energy tensor operator and its fundamental properties concerning locality and positivity.
Section \ref{SEC3agg} is devoted to the construction of a relativistic localization observable from the stress-energy operator and to the analysis of its fundamental mathematical and physical properties.
The final section \ref{SECFagg}, before the conclusions stated in Section \ref{CONC}, focuses on conditional localization POVMs arising from a local notion of the energy operator. In particular, we prove that these POVMs belong to local von Neumann algebras, as expected in the AHK framework, and in particular that they commute when associated with causally separated laboratories. Furthermore, we establish a relation between these POVMs and the relativistic localization observables defined on the whole spacetime introduced in the previous section.
The appendix contains several technical proofs of intermediate statements appearing in the main text.
\subsection{Notions and notations} \label{secdefgen} 
The reader may initially skip this section and come back to it later when necessary.

We assume $c=1$ and $\hbar=1$ throughout the rest of this paper, and the notation $A\subset B$ allows the case $A=B$. The Hilbert spaces we consider are complex, and symmetric operators are densely defined by definition.\\ 
 
\noindent {\bf A. Minkowski spacetime}. 
A four-dimensional real affine space whose space of translations $\sV$ is equipped with a bilinear, non-degenerate, symmetric form $g$ with signature $(-,+,+,+)$ is the {\bf Minkowski spacetime} $\bM$. The points of $\bM$ are called {\bf events} and $g$ is called the {\bf Minkowski metric}. We shall make use of the notation $u\cdot v :=g(u,v)$ if $u,v \in \sV$. We also use the dot to indicate the standard (positive) scalar product of $3$-vectors $\vec{u}\cdot \vec{v}$, viewing them as spacelike vectors (see below). 
Upon choosing an origin $o\in \bM$, the points $p\in \bM$ are in one-to-one correspondence with vectors of $\sV$ through the map $\bM \ni p \mapsto p-o \in \sV$. We shall take advantage of this identification several times in the rest of this paper. If $p\in \bM$ and $v\in \sV$, $q=p+v$ means that $q-p=v$.  
 
A vector $v\in \sV$ 
is {\bf spacelike} if $g(v,v)>0$ or $v=0$. 
 It is {\bf causal} if $g(v,v)\leq 0$ and $v\neq 0$. A causal vector $v$ is {\bf timelike} if $g(v,v) <0$, 
or {\bf lightlike} if $g(v,v) =0$. 
Smooth curves are classified analogously according to their tangent vectors if their type is constant along the curve. 
  
A set $\Lambda\subset \bM$ is {\bf achronal} if $p-q$ cannot be timelike for $p,q\in \Lambda$. A {\bf maximal achronal set} is an achronal set that is not a proper subset of another achronal set. $\Lambda \subset \bM$ is {\bf spacelike} if $p-q$ is spacelike 
 for $p,q\in \Lambda$. 
 
The set of timelike vectors is an open cone made up of two disjoint open connected halves. A choice of one of them $\sV_+$ defines a {\bf time orientation} of 
$\bM$. In what follows, $\bM$ is assumed to be {\bf time oriented}: $\sV_+\subset \sV$ is the open cone of {\bf future-directed timelike vectors}. $\overline{\sV_+}\setminus\{0\}$ is the cone of {\bf future-directed causal vectors}. Notice that if $v \in {\sV_+}$ and $u$ is causal, then $u \in \overline{\sV_+}$ if and only if $g(u,v) < 0$. 
We finally define the set of unit timelike future-directed vectors $\sT_+ := \{u\in \sV_+ \:|\: g(u,u)=-1\}$. 
 
If $A\subset \bM$, the {\bf causal future} $J^+(A):= \{p\in \bM \:|\: p-q \in \overline{\sV_+} \quad \mbox{for some $q\in A$}\}$ represents the events of $\bM$ in the future of $A$ that can be physically influenced by $A$. 
The {\bf causal past} $J^-(A):= \{p\in \bM \:|\: q-p \in \overline{\sV_+} \quad \mbox{for some $q\in A$}\}$ is defined symmetrically. 
 
$A,B \subset \bM$ are said to be {\bf causally separated} if $(J^+(A)\cup J^-(A)) \cap B= \emptyset$ (which is equivalent to $(J^+(B)\cup J^-(B)) \cap A= \emptyset$). 
 
If $R\subset \bM$, its {\bf causal complement} and {\bf causal completion} are, respectively,  \beq R^\pperp := \bM \setminus (J^+(R)\cup J^-(R)) \quad \mbox{and}\quad (R^{\pperp})^{\pperp}\:.\label{CComp}\eeq    It is easy to prove that $R\subset (R^{\pperp})^{\pperp}$. Furthermore, if $R$ is a subset of a flat spacelike 3-plane $\Sigma$, then $(R^{\pperp})^{\pperp}$ turns out to consist of the points $p\in \bM$ such that every causal {\em straight line} passing through $p$ meets $R$ somewhere\footnote{  
This is equivalent, in $\bM$, to the set of points $p$ such that every inextensible causal curve passing through $p$ also meets $R$ somewhere. In a generic spacetime $M$, this latter set is also called the {\em domain of dependence} of $R$ (usually defined when $R$ is achronal).}.

A {\bf  Minkowskian reference frame} -- physically representing an {\bf inertial reference frame} or an {\bf observer} -- is a unit timelike vector $n\in \sT_+$. 
A Cartesian coordinate system $\psi:  \bM \ni p \mapsto (x^0, x^1,x^2,x^3)\equiv (x^0, \vec{x})\in \bR \times \bR^3$ with origin $o\in \bM$ and axes $e_0,e_1,e_2,e_3 \in \sV$, 
is a {\bf Minkowskian coordinate system} if the basis $\{e_0,e_1,e_2,e_3\}= \{\partial_{x^0},\partial_{x^1},\partial_{x^2},\partial_{x^3}\}$ is $g$-{\bf orthonormal}: 
$g(e_a,e_b)= g(\partial_{x^a},\partial_{x^b})= g_{ab}$, where $[g_{ab}] = \eta := diag(-1,1,1,1) $ and  
$e_0=\partial_{x^0}\in \sT_+$. 
The Minkowskian coordinate system $\psi$ is {\bf adapted} to or {\bf comoving} with $u\in \sT_+$ if $\partial_{x^0}=u$. 
 
Given Minkowskian coordinates, vectors are decomposed as $\sV \ni v \equiv (v^0, \vec{v})$ where, according to the {\em Einstein summation convention} we adopt henceforth, $v=  
v^\mu e_\mu= v^\mu \partial_{x^\mu}$ so that 
$g(u,v)=u\cdot v = -u^0v^0 + \vec{u} \cdot \vec{v}\:.$ 
Here $v^0$ is called the {\bf temporal component} of $v$ and $\vec{v}$ are called the {\bf spatial components} of $v$ with respect to the chosen Minkowskian reference frame (or Minkowskian coordinate system). 
 
We shall take advantage of tensorial notation and of the raising- and lowering-index procedure, so that, for instance, in Minkowskian coordinates, $p_\mu = g_{\mu\nu}p^\nu$, ${T^\mu}_\nu = T^{\mu\alpha} g_{\alpha \nu}$. 
 
A {\bf rest space} of a Minkowski reference frame $u\in \sT_+$ is an affine $3$-plane $\Sigma\subset \bM$ $g$-normal to $u$, written $u\perp \Sigma$. 
If $o\in \bM$ is a given origin, the family of rest spaces of $u\in \sT_+$ is labeled by the time at which they occur in the chosen reference frame,  
$t_{o,\Sigma} =-(p-o)\cdot u$, which does not depend on $p\in \Sigma$.  
 Every spacelike affine $3$-plane is the rest space of a Minkowskian reference frame, indicated by $u_\Sigma$ and said to be {\bf adapted} to $\Sigma$, at some time. $u_\Sigma$ is the future-directed unit normal vector to $\Sigma$, which is necessarily timelike.  

A rest space meets exactly once every {\bf straight line} $p(r)= p_0 + ru$, $r\in \bR$, parallel to any given causal vector $u \in \overline{\sV_+}\setminus\{0\}$ and passing through any given $p_0\in \bM$. That is because a rest space is a {\em spacelike smooth Cauchy surface} \cite{Oneill} of $\bM$: more generally, it meets exactly once every {\em inextendible smooth causal curve}. 

$\cB(\Sigma)$ denotes the family of Borel subsets of a spacelike $3$-plane $\Sigma$ and $\cB_b(\Sigma)\subset \cB(\Sigma)$ the subfamily of bounded elements, boundedness being equivalently referred to any of the coordinate systems just described. 
 $d\Sigma$ denotes the natural translationally invariant Borel measure on the rest space $\Sigma$ of $u_\Sigma$, which coincides with   
the Lebesgue measure on the $\bR^3$ space of the spatial coordinates of any Minkowskian coordinate system $x^0,x^1,x^2,x^3$ comoving with $u_\Sigma$. It is easy to prove that $d\Sigma$ does not depend on the choice of such a Minkowskian coordinate system. 
We use the notation $|\Delta| = \int_\Sigma \chi_\Delta(x) d\Sigma(x)= \int_{\bR^3} \chi_\Delta(x) d^3x$ for the Lebesgue measure of  $\Delta \in \cB(\Sigma)$.  
 
The Lie group of metric-preserving affine maps $h: \bM \to \bM$ is known as the {\bf Poincar\'e group} $IO(1,3)$. The Lie subgroup of affine maps that also preserve the time orientation is called the {\bf orthochronous Poincar\'e group} $IO(1,3)_+$.  
The subgroup of $IO(1,3)_+$ that leaves fixed an arbitrarily chosen origin\footnote{Different choices of $o$ give rise to isomorphic definitions and the component $\Lambda$ of an element of $IO(1,3)_+$ does not depend on the choice of $o$.} $o\in \bM$ is the {\bf orthochronous Lorentz group} $O(1,3)_+$. 
 Elements $h\in IO(1,3)_+$ are in one-to-one correspondence with the pairs $(\Lambda, v)$ where $v\in \sV$ and $\Lambda \in O(1,3)_+$, and the action of $h$ on a point $q\in \bM$ is 
$ (\Lambda, v) q = o + v+ \Lambda (q-o)$.
In a given Minkowskian coordinate system centered at $o\in \bM$, the transformations $\Lambda \in O(1,3)_+$ are in one-to-one correspondence with the matrices, denoted by the same symbol,  
$\Lambda  
\in GL(4,\bR)$ 
such that $\Lambda^t \eta \Lambda =\eta$ and 
${\Lambda^0}_0 > 0$, where $\eta = diag(-1,1,1,1)$ as above.
 
The subgroup of $IO(1,3)_+$ known as the {\bf proper orthochronous Poincar\'e group} $ISO(1,3)_+$ is obtained by replacing $O(1,3)_+$ with the {\bf proper orthochronous Lorentz group} $SO(1,3)_+$. 
The latter, representing $O(1,3)_+$ in a Minkowskian coordinate system as above, is constructed by restricting to the Lorentz matrices with $\det \Lambda >0$. \\

\noindent {\bf B. Quantum observables in Hilbert space}. For every set ${\cal S}\subset \cH$, the latter being a Hilbert space, the  
{\bf span} of ${\cal S}$, denoted by $span \: {\cal S}\subset \cH$, is the subspace consisting of all complex {\em finite} linear combinations of elements of ${\cal S}$. 
 
Operators in Hilbert space have their own domains, $A: D(A)\to \cH$, where $D(A)\subset \cH$ is a subspace. The operations $A+B$, $AB$, $aA$ ($a\in \bC$) are defined on their standard domains: 
$D(A+B):= D(A)\cap D(B)$, $D(AB):= \{x\in D(B) \:|\: Bx \in D(A)\}$, $D(aA):=D(A)$ unless $a=0$, in which case $D(0A):=D(0) := \cH$. 
$A\geq 0$ means that $\langle \psi|A\psi \rangle \geq 0$ if $\psi \in D(A)$. In that case we say that the operator $A$ is {\bf positive}.   
 
$\gB(\cH)$ denotes the $C^*$-algebra of bounded operators $A:\cH\to \cH$. If $A,B \in \gB(\cH)$, $A\geq B$ (equivalently $B\leq A$) means that $A-B\geq 0$.

The (generalized) notion of observable that we shall use throughout is that of a {\bf Positive Operator-Valued Measure} (POVM) on a Hilbert space $\cH$. It is a map  
$\Sigma(X) \ni B \mapsto E(B)$, where $\Sigma(X)$ is a $\sigma$-algebra over the set $X$, each $E(B)\in \gB(\cH)$ is an {\bf effect}, i.e., satisfies $0\leq E(B)\leq I$, together with the {\bf normalization condition} $E(X)=I$, and, finally, the requirement that, 
for every $\psi \in \cH$, the associated map $\Sigma(X) \ni B \mapsto \langle \psi|E(B)\psi\rangle$ is $\sigma$-additive and therefore is a positive measure on $X$, which is a probability measure if $||\psi||=1$. By the positivity of the operators involved, this condition is equivalent to the strong $\sigma$-additivity of the map $\Sigma(X) \ni B \mapsto E(B)$, which, obviously, is also additive. 
$X$ is interpreted as the set of outcomes of the observable defined by the POVM $\Sigma(X) \ni B \mapsto E(B)$.  
 
Generally, mixed states $\rho$ are trace-class operators $\rho \in \gB_1(\cH)$, positive ($\rho \ge 0$), and normalized ($\mathrm{tr} \rho = 1$). The convex body of states will be denoted by $\sS(\cH)$. 
A special class of states is given by one-dimensional projectors $\rho = |\psi\rangle\langle \psi|$ for unit vectors, $\psi \in \cH$.  
These are pure states, i.e., extremal elements in the space of states if the von Neumann algebra of observables is the whole $\gB(\cH)$.
 
For a state $\rho\in \sS(\cH)$, $tr(E(B)\rho)= tr(\rho E(B))$ is interpreted as the probability of obtaining an outcome in $B$ when the system is in the state $\rho$. 
 
If $\Sigma(\bR)$ is the Borel $\sigma$-algebra $\cB(\bR)$ and all the effects $E(B)$ are orthogonal projectors, we have a standard {\bf Projector-Valued Measure} (PVM). As is well known, every PVM is in one-to-one correspondence with a selfadjoint operator $\hat{E}= \int_{\bR} xE(dx)$ through the spectral theorem. In this sense, a POVM is a generalized observable. 
 
In the special case where $\Sigma(X)$ is the power set of a countable set $X$, a POVM on $X$ is completely determined by the special effects $E_x:= E(\{x\})$ with $x\in X$. 
As a general textbook reference on this mathematical technology applied to physics, we suggest \cite{Buschbook}. References on general spectral theory as applied to physics that we shall use are \cite{Moretti1,Moretti2}. We assume the reader is familiar with basic properties of von Neumann algebras \cite{Takesaki,L} \\

\noindent {\bf C. Basic elements of the AHK approach}.
In the {\em Araki--Haag--Kastler (AHK) formalism} for local quantum theories in Minkowski spacetime $\bM$ \cite{Haag,Araki}, the von Neumann algebra $\gA$ of physically relevant operators of a quantum system described on the Hilbert space $\cH$ is generated by {\em local von Neumann algebras} $\gA({\cal O})$. 
There is a local von Neumann algebra $\gA({\cal O})$ for every open bounded set ${\cal O} \subset \bM$. Each such algebra contains operations and observables that are physically associated with ${\cal O}$: the corresponding physical operations and measurements are performed there.
More precisely, local observables are represented by selfadjoint operators that belong to these local algebras in the bounded case, or are {\em affiliated} with these local algebras in the unbounded case.
 The identity operator $I$ is, of course, common to all local algebras.
If ${\cal O}$ is open but not bounded, $\gA({\cal O})$ is the von Neumann algebra generated by the family of bounded open subsets of ${\cal O}$. {\em Isotony} holds:
$\gA({\cal O})\subset \gA({\cal O}_1)$ if ${\cal O}\subset {\cal O}_1$.

One of the fundamental assumptions is {\em relativistic locality}: operators belonging to algebras associated with causally separated regions must commute,
$$[A_1,A_2]=0\quad \mbox{if $A_1\in \gA({\cal O}_1)$,  $A_2\in \gA({\cal O}_2)$ and ${\cal O}_1 \cap (J^+({\cal O}_2) \cup J^-({\cal O}_2)) = \emptyset$.}$$ 
The net of algebras is assumed to admit a strongly continuous unitary representation of the Abelian translation group $\sV$ of $\bM$ satisfying the {\em spectral condition}: the joint spectrum of the self-adjoint generators must lie in $\overline{\sV_+}$. This representation is assumed to extend to a full (strongly continuous) unitary representation of $IO(1,3)_+$.

Finally, the Hilbert space contains a preferred Poincar\'e-invariant state, represented by a unit vector $\Omega$, the {\em vacuum vector state}, which has the property of being {\em cyclic}: the subspace spanned by the vectors $A\Omega$ with $A\in \gA$ is dense in $\cH$.

{Most features of this approach can be generalized to the case in which $\gA$ and every $\gA(\cal O)$ are {\em unital $*$-algebras}, in particular algebras of operators on a given Hilbert space with a common invariant domain. Some further notions and results of the AHK approach will be briefly presented in Sec.\ref{secHD}.}

\section{Localization observables and conditional localization}\label{SEZloc}

\subsection{Localization observables, causal conditions, non-commutativity} 
The notion of localization used in \cite{Mor126} and in this paper is encapsulated in the following definition of {\em relativistic spatial localization observable}. This notion, in a slightly simplified form, was first introduced and 
analyzed in depth by Castrigiano (see \cite{Castrigiano2} and references therein) under the 
name of {\em Poincar\'e covariant POL}. \\
 
\begin{definition}[Relativistic Spatial Localization Observable]\label{REMM0}  
A {\bf relativistic spatial localization observable} is a quadruple $(\cH, {\cal R}, A, U)$ where   
\begin{itemize}  
\item[(a)]  $\cH$ is a complex Hilbert space;
 \item[(b)] ${\cal R}= \cup \{\cB(\Sigma)\:|\: \Sigma \subset \bM \: \mbox{spacelike  $3$-plane}\}$ where $\cB(\Sigma)$  
is the Borel $\sigma$-algebra on $\Sigma$;  
    \item[(c)]  $A= \{A^s\}_{s\in \sS}$ is a family of  maps $A^s: {\cal R} \to \gB(\cH)$ -- where $\sS$ is a set of tensors\footnote{The dependence on the tensorial index $s\in \sS$ could be trivial, as happens for various relativistic spatial localization observables constructed in the literature, in particular for the fermionic  POLs in \cite{Castrigiano2} the bosonic ones in \cite{C23} where no such dependence exists. Conversely, it appears in the localization observable constructed from the stress--energy tensor  \cite{M23,DM24} and some of \cite{CDRM}.} 
  of definite order that is invariant under $O(1,3)_+$ --  
 such that every restriction  
    $A^s\spa \rest_{\cB(\Sigma)} : \cB(\Sigma) \to \gB(\cH)$ is a (normalized) POVM;  
\item[(d)]  $U: IO(1,3)_+ \to \gB(\cH)$  is a  strongly continuous unitary representation of the  orthochronous Poincar\'e group;  
\item[(e)] $A$ is $U$-covariant, i.e.,  $U_g A^s(\Delta) U_g^{-1} = A^{gs}(g\Delta)$ for every $g\in {IO(1,3)_+}$ and $\Delta \in {\cal R}$,  
 where   
$\sS \ni s\mapsto gs\in \sS$ denotes the action of  ${IO(1,3)_+}$ on the tensors in $\sS$.  \end{itemize}  
$(\cH, {\cal R}, A, U)$ is of {\bf positive-energy} type if the selfadjoint generator of spacetime translations $\bR \ni r \mapsto U_{(I,rv)}$ is positive for $v\in \sT_+$. \hfill $\blacksquare$\\
\end{definition}

 For future reference, observe that, if $V\in \gB(\cH)$ is a unitary operator and $(\cH, {\cal R}, A, U)$ is a relativistic spatial localization observable, then 
 $(\cH, {\cal R}, A_V, U_V)$, where $A^s_V(\Delta):= V A^s(\Delta) V^\dagger$ and $(U_{V})_g:= V U_g V^\dagger$, is another relativistic spatial localization observable. It also satisfies the causality condition below  and is of positive-energy type if $(\cH, {\cal R}, A, U)$ is.
 
Let us now turn to the {\em causality condition} imposed on relativistic localization observables in order to comply with locality constraints at the level of detection probabilities, which cannot evolve superluminally. The condition stated below, due to Castrigiano, is stronger than the original one formulated by Hegerfeldt, which, however, ruled out all positive-energy  relativistic spatial localization observables described by PVMs, such as the one associated with the triple of Newton-Wigner operators. A discussion of the relevance of this type of condition appears in  \cite{Castrigiano2,M23,C24,CDRM}. \\
 
\begin{definition}\label{REMM0444}  
{\em A relativistic spatial localization observable $(\cH, {\cal R}, A, U)$ is {\bf causal} if it satisfies the following. For every $\Delta \in {\cal R}$, every spacelike $3$-plane $\Sigma$, and every $s\in \sS$,  
\begin{itemize}  
    \item[{\bf (CC)}] $\quad$ $A^s(\Delta) \leq A^s(\Delta_\Sigma) \quad \mbox{where $\Delta_\Sigma:= \Sigma \cap (J^+(\Delta)\cup J^{-}(\Delta))$ provided $\Delta_\Sigma \in {\cal R}$.}$   \hfill $\blacksquare$ 
  \end{itemize} } 
\end{definition}  
The explicit examples of relativistic spatial localization observables, especially of positive-energy type,  presented in  
\cite{Castrigiano2, M23, DM24, C23, C24, CDRM} for various types of particles satisfy CC. They therefore show that Hegerfeldt's causality issue can in fact be regarded as harmless when one considers certain unsharp notions of localization, whereas sharp ones are ruled out.

In \cite{Mor126}, after a detailed analysis of relativistic locality from Busch's perspective \cite{Buschloc1,Buschloc2,B07}, we have asserted that, when dealing with relativistic spatial localization observables (also satisfying CC), there is no compelling reason to require that $[A^s(\Delta), A^{s'}(\Delta')]=0$ when $\Delta$ and $\Delta'$ are sharply causally separated, i.e., they are included in respective open regions which are causally separated.  There is no compelling reason, in particular,  if referring to  systems like {\em particles} which have the propensity to localize at a unique position of a rest space when the rest space is ideally filled with detectors. 
This result agrees with the celebrated theorem of Halvorson and Clifton \cite{HC} which proves that the above  commutativity is actually forbidden if $(\cH, {\cal R}, A, U)$ is of positive-energy type\footnote{This is even true with a weaker notion of localization observable than that of Definition \ref{REMM0444}, see the review in \cite{Mor126}.}.
Failure of commutativity has an important consequence.  If $\Delta$ and $\Delta'$ are bounded regions contained in causally separated open sets ${\cal O}$ and ${\cal O}'$ of $\bM$, one  concludes that $A^{s}(\Delta)$ and $A^{s}(\Delta')$ cannot  belong to the corresponding local algebras of observables in the AHK approach. Indeed, if this were the case, they would necessarily commute. This shows that the effects $A^{s}(\Delta)$ of a relativistic spatial localization observable are not local observables in the AHK sense.

\subsection{Conditional localization and commutativity}
As discussed in \cite{Mor126}, realistic localization experiments in $\bM$ are performed in {\em laboratories}, which are spatially finite regions that are also {\em causally complete} in their spacetime description: everything that may happen in such a spacetime region should be determined by physical actions performed within it, at least at a macroscopic level.

If $\Sigma$ is a flat $3$-dimensional plane, a {\bf laboratory} $L(\Delta_0)\subset \bM$, with {\bf space} given by a bounded set $\Delta_0\in \cB(\Sigma)$, is defined as the causal completion $ L(\Delta_0) := (\Delta_0^\pperp)^\pperp$. If $\Delta_0$ is open in the relative topology, the resulting open set $L(\Delta_0)$ is a globally hyperbolic spacetime in its own right, with $\Delta_0$ as a possible smooth spacelike Cauchy surface, as well as other {\em curved} smooth spacelike Cauchy surfaces.

As discussed in \cite{Mor126}, if we are given a relativistic spatial localization observable $A$ on the Hilbert space $\cH$ (we omit the tensorial index $s$ for simplicity), we can define a POVM $ B^V_{\Delta_0}$ in each laboratory $L(\Delta_0)$
\beq B^V_{\Delta_0}(\Delta)=V \frac{1}{\sqrt{A(\Delta_0)}}A(\Delta)\frac{1}{\sqrt{A(\Delta_0)}}V^\dagger \:, \quad \Delta \in \cB(\Delta_0)\label{BD222}
\eeq
where $V\in \gB(\cH)$ is a given unitary operator. We assumed above that $A(\Delta_0)$ is strictly positive; a discussion of this technical point can be found in \cite{Mor126}.

As established in \cite{Mor126}, for $0\leq \delta <1$ and all $\rho \in \sS(\cH)$, the so-called {\em gentle measurement lemma} entails
\beq\label{G}
tr(\rho A_V(\Delta_0)) \geq 1-\delta \quad \mbox{implies}\quad\left|tr\left({\rho} B^V_{\Delta_0}(\Delta) \right) -\frac{tr(\rho A_V(\Delta))}{tr(\rho A_V(\Delta_0))}\right| \leq 2\sqrt{\delta} + \delta\:,\eeq
where $A_V(\Delta) := VA(\Delta)V^{-1}$.
In other words, if we start with a state $\rho \in \sS(\cH)$ whose probability of finding the system in $\Delta_0$ is close to $1$, $B^V_{\Delta_0}(\Delta)$ measures the fraction of detections in $\Delta$ over the total number of detections {\em observed in the laboratory} $\Delta_0$, when an ensemble of identical systems is prepared in the initial state ${\rho}$ and the relativistic spatial localization observable with effects ${\cal R}\ni \Delta \mapsto A_V(\Delta) := V A(\Delta)V^\dagger$ is used to define the spatial localization of the system.
In this sense, $B^V_{\Delta_0}$ represents, at least in the above limit, a {\bf conditional localization observable}: $B^V_{\Delta_0}(\Delta)$ gives the probability of finding the system in $\Delta\subset \Delta_0$, provided that it is known to be found in the bounded measurable set $\Delta_0$. This interpretation becomes increasingly accurate as the initial probability of finding the system in $\Delta_0$ approaches unity. However, we can assume that this interpretation is valid also in general situations.

As discussed in \cite{Mor126}, in principle Definition (\ref{BD222}) can be extended to the case where the map $\cB(\Sigma) \ni \Delta \mapsto A(\Delta)$ used to construct $B^V_{\Delta_0}(\Delta)$, with $\Delta_0$ a bounded measurable set of the rest space $\Sigma$, is merely a non-normalized positive operator-valued measure. In this case, to apply the gentle measurement lemma, the hypothesis in (\ref{G}) must be replaced by $tr(\rho A'_V(\Delta_0)) \geq 1-\delta$, where we have defined the {\em effect} $A'_V(\Delta):= ||A_V(\Delta_0)||^{-1} A_V(\Delta)$. The definition of $B^V_{\Delta_0}$ (\ref{BD222}) and the fraction in the second equation in (\ref{G}) are invariant if we replace $A_V$ with $A'_V$. 

In \cite{Mor126} we asserted that, in contrast to what happens for the effects of relativistic spatial localization observables, the effects $B^V_{\Delta_0}(\Delta)$ and $B'^{V'}_{\Delta'_0}(\Delta')$ may commute if $\Delta_0$ and $\Delta'_0$ are (sharply) causally separated. The present paper aims to construct analogous local conditional localization observables for a system described by a free real scalar quantum field. More precisely, the operators $B^V_{\Delta_0}(\Delta)$ that we shall introduce will be elements of local von Neumann algebras of observables $\gW({\cal O})$ generated by the corresponding local Weyl algebras $\cW(\cal O)$, where ${\cal O}$ is a sufficiently large open double cone such that, in particular, ${\cal O}\supset \Delta_0$.
\section{Elements of free QFT in Minkowski spacetime}\label{QFTsec}

The next sections review the elementary rigorous formulation of quantum field theory for a free real scalar field in Minkowski spacetime, in the Fock representation associated with the Minkowski vacuum state. The exposition is deliberately elementary, pedagogical, and, as far as possible, self-contained (see \cite{KM} for a more advanced review of QFT in curved spacetime in the algebraic formalism). The relevant notions, including in particular some elementary foundational aspects of the Araki--Haag--Kastler approach, are introduced progressively. It is nevertheless assumed throughout that the reader is familiar with the mathematical notions of $*$-algebra, $C^*$-algebra, and von Neumann algebra. The goal is to bring the community working on quantum measurement theory closer to the community working on local quantum field theory.

\subsection{Bosonic Fock space and particle states}
We shall refer here to the mathematical formulation of the free real scalar boson quantum field as presented in Section X.7 of \cite{RS2}, but using different notation in order to make contact with the results in \cite{M23,DM24}. 
The complex {\bf Schwartz space} on $\bR^n$ is denoted by 
$\cS(\bR^n)$ and $\cD(\bR^n):= C_c^\infty(\bR^n)\subset \cS(\bR^n)$ denotes the space 
 of complex {\bf test functions}. The respective real subspaces of real-valued functions are denoted by $\cS_\bR(\bR^n)$ and $\cD_\bR(\bR^n)$.

If $\cH$ is a Hilbert space, we define
$$\cH^{(1)}:= \cH^{1} :=  \cH\:, \quad \cH^{n} := \cH\otimes \cdots \mbox{\scriptsize (n times)}\cdots \otimes \cH\quad \mbox{and}\quad  \cH^{(n)} = S_n\cH^{n} \quad \mbox{if $n>1$}\:,$$
where $S_n : \cH^{n} \to \cH^{n}$ is the orthogonal projector onto the completely symmetric subspace of $\cH^{(n)}$ under the action of the unitary representation of the permutation group of $n$ elements. 
The (separable) Hilbert space in which we develop our theory will be the {\bf bosonic Fock space} $ \gF_s(\cH^{(1)}) $ of scalar particles of mass $m>0$
\beq \gF_s(\cH) :=  \bigoplus_{n=0}^{+\infty} \cH^{(n)} \quad \mbox{where in our case} \quad \cH= \cH^{(1)} := \cH_m :=  L^2(\sV_{m,+}, d\mu_m(p))\:.\label{FoK}\eeq
In (\ref{FoK}),
   the symbol $\oplus$
denotes the orthogonal Hilbert direct sum of Hilbert spaces. The scalar product on the {\bf vacuum subspace} $\cH^{0} := \cH^{(0)} := \bC$ is standard multiplication. 
The {\bf future mass shell} $\sV_{m,+}$, with $m>0$, and the $O(1,3)_+$-{\bf invariant measure} on it are 
$$\sV_{m,+} := \{p \in V_+\:|\: p\cdot p = -m^2\}\:, \quad d\mu_m(p) := \frac{d^3p}{E(p)}$$
The latter identity is valid in every Minkowskian coordinate frame, where one identifies $\bR^3$ with the $3$-space of the spatial part of the four-momenta, and
\beq \label{defp}p \equiv (p^0,\vec{p} ) \quad \mbox{with}\quad p^0=- p_0 := E(p) := \sqrt{m^2+ \vec{p}^2}\:.\eeq

Some relevant terminology is listed below where, from now on, if $\Psi \in \gF_s(\cH^{(1)})$, $\Psi_n$ denotes its component in $\cH^{(n)}$: $\Psi = \oplus_{n=0}^{+\infty} \Psi_n$
\begin{itemize}
\item[(a)] Each normalized $\Psi= \Psi_n \in \cH^{(n)}$ is an {\bf $n$-particle vector state};
\item[(b)]$ \Omega =1 \in  \bC = \cH^{(0)}$ represents the {\bf Minkowski vacuum state};
\item[(c)] $\cH^{(1)}$
is the {\bf one-particle space}\footnote{Denoted by $\cH$ in \cite{M23}.} which determines the whole structure of 
 $\gF_s(\cH^{(1)})$.  
\item[(d)] A relevant dense subspace is the {\bf finite-particle subspace}, 
\beq \label{FZ}\gF_0 := \{ \Psi \in \gF_s(\cH^{(1)})  \:|\: \Psi_n \neq 0 \: \mbox{only for a finite set of $n\in \bN$ depending on $\Psi$}\}.\eeq
 \end{itemize}

As the above terminology suggests, unit vectors $\Psi \in \gF_s(\cH^{(1)})$, when written in terms of states $|\Psi\rangle \langle \Psi| \in \sS(\gF_s(\cH^{(1)}))$, represent pure quantum states of a quantum system of {\em quantum particles} whose mass is $m$ and whose elementary properties, such as the four-momentum, are described in the one-particle Hilbert space $\cH^{(1)}=\cH_m$. The system of particles is associated with a {\em quantum field} of real scalar bosonic type, which we shall introduce shortly.

All the notions presented are Poincar\'e-invariant under the unitary strongly continuous representation 
of the orthochronous Poincar\'e group defined on the Fock space:
\beq U: IO(1,3)_+\ni g \mapsto U_g := U_g^{(0)}\oplus  \bigoplus_{n=1}^{+\infty} U_g^{(1)}\otimes \cdots \mbox{\scriptsize (n times)}\cdots \otimes U^{(1)}_g \in \gB( \gF_s(\cH_m))\label{rep}\eeq
where, if $\psi \in \cH^{(1)} =\cH_m$ and $(\Lambda, a)\equiv g \in IO(1,3)_+$,
\beq\label{rep1}
(U^{(1)}_{(\Lambda, a)}\psi)(p) = e^{-ia\cdot p} \psi\left(\Lambda^{-1}p\right)\:,
\eeq
 and $U^{(0)}$ is the trivial representation of $IO(1,3)_+$ on $\cH^{(0)}= \bC$.
In particular, the Minkowski vacuum state represented by $\Omega$ is {\em Poincar\'e-invariant} by construction.

From now on, $\sV_{m,+}^n := \times_{j=1}^n \sV_{m,+}$ and 
$\cS(\sV^n_{m,+})$ is the space of maps $\psi : \sV_{m,+}^n  \to \bC$
which are in $\cS(\bR^{3n})$ in a given Minkowskian coordinate system where $\sV_{m,+}$ is identified with $\bR^3$ consisting of the spatial components $\vec{p}$ of the four-momenta $p\equiv (E(p), \vec{p})$. We 
therefore write
$\cS(\sV^n_{m,+}) \equiv \cS(\bR^{3n})$ in Minkowskian coordinates.
An analogous definition is given for $\cD(\sV^n_{m,+}) \equiv \cD(\bR^{3n})$. The spaces of distributions $\cS'(\sV^n_{m,+})$ and $\cD'(\sV^n_{m,+})$ are defined correspondingly.
It is easy to see that all these definitions do not depend on the chosen Minkowskian coordinate system.

A pair of dense subspaces is defined according to the previous definitions: the {\bf finite-particle Schwartz subspace}
\beq \gS_0:= \{ \Psi \in \gF_0 \:|\: \Psi_n \in \cS(\sV_{m,+}^n),\:\: n\in \bN \},\label{SF}\eeq
and the subspace of {\bf finite-particle smooth compactly supported vectors}
\beq \gD_0:= \{ \Psi \in \gF_0 \:|\: \Psi_n \in \cD(\sV_{m,+}^n),\:\: n\in \bN \}.\label{DF}\eeq
Evidently $\gD_0\subset \gS_0 \subset \gF_0$.

We leave to the reader the easy proof of the following technical result.\\
\begin{lemma} \label{lemmaH} If $v\in \sV$, consider the one-parameter subgroup $\bR\ni r \mapsto U_{(I, r v)} = e^{ir H^v}$ of the representation (\ref{rep}). The selfadjoint generator $H^v$ satisfies the following.
\begin{itemize} 
\item[(a)] It admits each $\cH^{(n)}$ as a reducing space\footnote{$P_n H^v \subset H^vP_n$ if $P_n$ is the orthogonal projector onto $\cH^{(n)}$.} and $\gS_0$ and $\gD_0$ as invariant subspaces.
\item[(b)] As $\gD_0$ is dense and made of analytic vectors, and $\gD_0\subset \gS_0$, $H^v$ is essentially selfadjoint on $\gS_0$ and $\gD_0$.
\item[(c)] $H^v\Omega =0$ and, for $n=1,2,\ldots$
\beq (H^v \Psi)(p_1,\ldots, p_n) = - \sum_{k=1}^n v\cdot p_k \Psi(p_1,\ldots, p_n)  \quad \mbox{if $\Psi \in  \cH^{(n)} \cap \gS_0$.}\label{Hv}\eeq
\end{itemize}
\end{lemma}

\noindent In the case $v\in \sV_+$, we call $H^v$ the {\bf Hamiltonian} operator associated with $v$.

\subsection{Free  Klein-Gordon quantum  field  in Minkowski spacetime}\label{Secphib}
If $\psi \in \cH^{1}$, consider the unique linear continuous extensions of the operators, for every given $n\in \bN$,
\beq
b_n(\psi) :   \cH^{n} \ni  \psi_1 \otimes \cdots \otimes \psi_n &\mapsto&   \langle \psi| \psi_1 \rangle \psi_2 \otimes \cdots \otimes \psi_n  \in \cH^{n-1}\:, \quad b_0 := 0\nonumber\\
b^\dagger_n(\psi) :   \cH^{n} \ni  \psi_1 \otimes \cdots \otimes \psi_n &\mapsto&  \psi \otimes \psi_1 \otimes \cdots \otimes \psi_n  \in \cH^{n+1}\:.\nonumber
\eeq 
Still indicating by $b_n(\psi)$ and $b^\dagger_n(\psi)$ the said extensions, 
the {\bf annihilation} and {\bf creation operators}, respectively $a(\psi) : \gF_0 \to \gF_0$ and $a^\dagger(\psi) : \gF_0 \to \gF_0$, are defined as the linear extensions 
to $\gF_0$ of the respective maps
\beq \label{aa*}
a(\psi)|_{\cH^{(n)}} :  \cH^{(n)} \ni \Psi&\mapsto& \sqrt{n} S_{n-1} b_n(\psi) \Psi\in \cH^{(n-1)}\:, \quad a(\psi)|_{\cH^{(0)}} := 0\\
a^\dagger(\psi)|_{\cH^{(n)}} :  \cH^{(n)} \ni \Psi &\mapsto&  \sqrt{n+1} S_{n+1} b^\dagger_n(\psi) \Psi \in \cH^{(n+1)}\:.
\eeq 
The operators $a(\psi)$ and $a^\dagger(\psi)$ thus obtained enjoy some elementary properties (see e.g. \cite{RS2}.)\\

\begin{proposition} \label{PROPAA}  $a(\psi) : \gF_0 \to \gF_0$ and $a^\dagger(\psi) : \gF_0 \to \gF_0$, $\psi \in \cH_m$, satisfy the following.
\begin{itemize}
\item[(a)] $\cH^{(1)} \ni \psi \mapsto a(\psi)$ is {\em antilinear} and $\cH^{(1)} \ni \psi \mapsto a^\dagger(\psi)$ is {\em linear}; both are $\bR$-linear. 
\item[(b)] On their dense and invariant domain $\gF_0$
 $$a^\dagger(\psi) = a(\psi)^\dagger|_{\gF_0}\quad \mbox{ and}\quad  a(\psi) = a^\dagger(\psi)^\dagger|_{\gF_0}.$$ 
 \item[(c)] The {\bf Bosonic commutation rules} hold for every $\psi,\psi'\in \cH^{(1)}$
$$[a(\psi), a^\dagger(\psi')] = \langle \psi|\psi'\rangle I,\quad [a(\psi), a(\psi')] = 0,\quad [a^\dagger(\psi), a^\dagger(\psi')] =0\:.$$

\item[(d)]  If $\psi \in \cH^{(1)}$ and $\Psi_n \in \cH^{(n)}$ then
  \beq ||a^\#(\psi)\cdots (\mbox{\scriptsize $k$ times}) \cdots a^\#(\psi) \Psi_n|| \leq \sqrt{(n+1) \cdots (n+k)} ||\psi||^k ||\Psi_n||\:.\label{stime}\eeq
where $a^\#$ represents either $a$ or $a^\dagger$.

\item[(e)] If $\psi \in \cH_m$ and $g\in IO(1,3)_+$ then
\beq
U_g a(\psi) U^{-1}_g = a(U_g^{(1)}\psi)\:, \quad  U_g a^\dagger(\psi) U^{-1}_g = a^\dagger(U_g^{(1)}\psi)\:.
\eeq
\end{itemize}
\end{proposition}

 As is well known already from the original formulation of QFT, trying to define field operators $\hat{\phi}(x)$ localized at each point $x$ of spacetime gives rise to insurmountable mathematical difficulties \cite{Haag}. What one can define is a quantum field operator {\em smeared} with test functions $f$ and denoted by $\hat{\phi}[f]$. 
We therefore move on to define the (free) {\em quantum-field operator} smeared with a {\em test function} $f\in  \cD(\bM)$,
 where $\bM\equiv \bR^4$ according to a Minkowskian coordinate system. We limit ourselves to stating the most relevant elementary technical features of this notion. More information about this classical construction and the physical motivations underpinning this crucial physical tool can be found in the vast literature on the subject (see e.g. \cite{PCT,Araki,Haag}). Regarding the smearing procedure, we have to stress that it is pervasive in rigorous QFT and is the practical procedure used to associate observables with regions of spacetime (where the supports of the smearing functions are localized) in agreement with the basic assumptions of the AHK formulation.

First of all, upon the choice of an origin of $\bM$, the {\bf covariant Fourier transform} of $f\in \cS(\bM)$ is
\beq \hat{f}(p) := \frac{1}{(2\pi)^2}\int_{\bM} e^{-i p\cdot x} f(x) d^4x  \:.\label{Fourier}\eeq
The measure $d^4x$ in the integral is the standard Lebesgue measure in every Minkowskian coordinate system, which turns out to be $IO(1,3)_+$-invariant. \\

\begin{definition}
{\em The {\bf real scalar field operator} of mass $m>0$ smeared with $f\in \cD(\bM)$ is the densely defined operator
\beq
\hat{\phi}[f] := \frac{1}{\sqrt{2}}\left( a({\kappa_m \overline{f}}) + a^\dagger(\kappa_m f)\right): \gF_0 \to \gF_s(\cH_m) \label{PHI}
\eeq
where we have used the 
 $\bC$-linear map
\beq \kappa_m : \cS(\bM)\ni f \mapsto\sqrt{2\pi} \hat{f}|_{\sV_{m,+}} \in L^2(\sV_{m,+}, d\mu_m(p)) = \cH_m\:.\label{kappa}\eeq
and the bar in $\overline{f}$ denotes complex conjugation. } \hfill $\blacksquare$\\
\end{definition}

\begin{remark}\label{REMran}
{\em 
\begin{itemize}
\item[(1)] Three alternative definitions -- {\em all equivalent in Minkowski spacetime} -- are used in the literature, where the space of smearing functions $\cD(\bR)$ \cite{Araki} is respectively replaced by $\cD_\bR(\bM)$ or $\cS(\bM)$ \cite{PCT} or $\cS_\bR(\bM)$. The equivalence is based on two facts. (a) In Minkowski spacetime and with the construction above,
$\hat{\phi}[f] = \hat{\phi}[Re(f)]+i \hat{\phi}[Im(f)]\:.$
(b) As is easy to prove, the real subspace $\kappa_m(\cD_\bR(\bM))$ satisfies
\beq
\overline{\kappa_m(\cD_\bR(\bM))} =\overline{\kappa_m(\cD_\bR(\bM))+ i\kappa_m(\cD_\bR(\bM)) }= \overline{\kappa_m(\cD(\bM))} = \cH_m \:, \label{gggg}
\eeq
where the bar denotes closure in the topology of $\cH_m$. (The same properties are valid if one replaces everywhere $\cD$ by $\cS$.)
Identity (\ref{gggg}) is equivalent to the fact that $\Omega$ is a {\em pure algebraic state} on the abstract $*$-algebra of the field operators (see, e.g., \cite{KM}).
\item[(2)] The smeared quantum fields constructed as above form an elementary system that satisfies the classical {\em G\aa rding-Streater-Wightman axioms} in Minkowski spacetime \cite{PCT,Araki}. \hfill $\blacksquare$\\
\end{itemize}}
\end{remark}

To proceed, consider the (non-homogeneous) {\bf Klein-Gordon equation} in $\bM$
\beq
(\Box - m^2)\phi =f\:.
\eeq
where the {\bf d'Alembert operator} $\Box$ is written as $\partial_\mu \partial^\mu = -\partial^2_{x^0}+\Delta_{\vec{x}}$ in every Minkowskian coordinate system.
That equation admits unique {\bf advanced} and {\bf retarded fundamental solutions}, linear maps $A: \cD(\bM)\to C^\infty(\bM)\quad\mbox{and}\quad R: \cD(\bM)\to C^\infty(\bM)$ respectively, completely defined by the requirement 
that $supp(Rf)\subset J^+(supp(f))$ and $supp(Af)\subset J^-(supp(f))$ and, obviously, $(\Box - m^2)Af= (\Box - m^2)Rf=f$ for every $f\in\cD(\bM)$. Their difference, called the {\bf causal propagator}
$E := A-R$,
has the consequent property that $\int_\bM f(x) (E g)(x) d^4x=0$ if $supp(f)$ and $supp(g)$ are causally separated.\\

\begin{proposition}\label{PROP35} The field operators defined above satisfy the following properties.
\begin{itemize}
\item[(a)] The $\gF_s(\cH_m)$ subspaces $\gF_0$, $\gS_0\subset \gF_0$, and $\gG_0\subset \gS_0$ defined as\footnote{Notice that in principle the ordering
of the factors in $\prod_{k=1}^n \hat{\phi}(f^{(n)}_k)\Omega$  is relevant.}
\beq \gG_0:= \left\{ \Psi \in \gF_0 \:\left|\: \Psi_n \in span\left\{ \prod_{k=1}^n\right.\hat{\phi}(f^{(n)}_k)\Omega,\quad  f^{(n)}_k \in   \cD(\bM) \:, k=1,\ldots, n\right\}\right\}\:, \label{GF}\eeq
are dense, $\hat{\phi}[f]$-invariant, and made of analytic vectors of $\hat{\phi}[f]$ with $f\in \cD(\bM)$.
\item[(b)] If $f\in \cD(\bR^4)$ then $\hat{\phi}[\overline{f}]\subset \hat{\phi}[f]^\dagger$ so that $\hat{\phi}[f]$ is closable.
 More strongly
$$\overline{\hat{\phi}[f]} =\hat{\phi}[f]^\dagger \quad \mbox{if $f\in \cD_\bR(\bM)$}\:,$$
namely $\hat{\phi}[f]$ is essentially selfadjoint if smeared with real functions. According to (a), $\gF_0$, $\gS_0$, and $\gG_0$ are cores of $\hat{\phi}[f]$
for $f\in \cD_\bR(\bM)$.

\item[(c)] $\cD(\bM) \ni f \mapsto \hat{\phi}[f]$ enjoys the following further properties.

\begin{itemize}
%

\item[(c1)] {\bf CCR}. The canonical commutation rules hold: 
 $$[\hat{\phi}[f],\hat{\phi}[g]] = -iE(f,g)I :=  -i\int_\bM f(x) (E g)(x) d^4x  I\quad \mbox{for}\quad f,g \in \cD(\bM)\:,$$
so that $[\hat{\phi}[f],\hat{\phi}[g]] =0$ if $supp(f)$ and $supp(g)$ are causally separated 

\item[(c2)]{\bf KG equation}. It solves the homogeneous Klein-Gordon equation in the distributional sense:
$$\hat{\phi}[(\Box - m^2)f]=0 \quad \mbox{for} \quad f \in \cD(\bM)\:.$$ 
    \item[(c3)] $IO(1,3)_+$-{\bf covariance}. The representation (\ref{rep}) of $IO(1,3)_+$ acts covariantly on the field operator:
    $$U_{g} \hat{\phi}[f]U^{-1}_{g} = \hat{\phi}[g_*f]\quad\mbox{for $g\in IO(1,3)_+$ and $f\in  \cD(\bM)$.}$$
    where $(g_*f)(x) := f(g^{-1}x)$ for every $x\in  \bM$.
    \item[(c4)] The {\bf Weyl generators} $W(f) := e^{i\overline{\hat{\phi}(f)}}$ satisfy the {\bf Weyl relations} for every $f,g \in \cD_\bR(\bM)$
    \beq W(f) W(g) = e^{i E(f,g)/2} W(f+g)\:,\quad W(f)^\dagger= W(-f)\:, \quad W(0) = I  \:. \label{Weyl}
    \eeq
    In particular $[W(f), W(g)] =0$ if $supp(f)$ and $supp(g)$ are causally separated.
\end{itemize}
    \item[(d)] {\bf Reeh-Schlieder property}. 
    \begin{itemize}
    \item[RS1]  If ${\cal O}\subset \bM$ is open, bounded, and non-empty, the unital $*$-algebra ${\cal A}({\cal O})$ generated by the field operators $\hat{\phi}[f]$ for smearing functions such that $supp(f) \subset {\cal O}$ satisfies $\overline{{\cal A}({\cal O})\Omega} = \gF_s(\cH_m)$.
    \item[RS2] If ${\cal O}\subset \bM$ is open, bounded,  and non-empty,  the $C^*$-algebra ${\cal W}({\cal O})$ generated by the Weyl operators $W(f)$ for smearing functions such that $supp(f) \subset {\cal O}$ satisfies $\overline{{\cal W}({\cal O})\Omega} = \gF_s(\cH_m)$. \end{itemize}
\end{itemize}
\end{proposition}
\begin{proof}
(a), (b), and (c3) are consequences of Proposition \ref{PROPAA}, see e.g. \cite{RS2}, paying attention to the use of different notation and taking Remark \ref{REMran} into account. For (c) see \cite{RS2} and \cite{KM} for a generic curved globally hyperbolic spacetime and a quasifree state. A proof of the version of the Reeh-Schlieder property presented in (d)RS1 can be obtained by adapting the more general result stated in Theorem 4-2 of \cite{PCT}; see  Theorem 4.14 in \cite{Araki} for the version RS2 applied to the case of a free scalar field. \end{proof}

Another important feature of the free-field operators introduced above and the associated Weyl algebra is known as {\em Haag duality}. We shall briefly discuss it in Section \ref{secHD}.

Due to Remark \ref{REMran}, in Minkowski spacetime Proposition \ref{PROP35} is still valid\footnote{the CCR require a more delicate adaptation if one smears with non-compactly supported functions, because of the nature of the advanced and retarded solutions of the KG equation.} if one replaces $\cD(\bM)$ (resp. $\cD_\bR(\bM)$) by $\cS(\bM)$ (resp. $\cS_\bR(\bM)$) and changes the statements accordingly.
$(\gF_s(\cH_m),\gG_0, \pi, \Omega)$ is the {\em GNS structure} of the Minkowski-vacuum representation of the unital $*$-algebra $\cA$ called the {\em CCR algebra} generated by abstract field operators $\phi[f]$, where $\pi$ is the $*$-homomorphism induced by $\pi(\phi[h]) = \hat{\phi}[h]$ \cite{KM}.

\subsection{Free Quantum Fields in $\bM$ as Quadratic forms}  The stress-energy operator is a special case of a {\em Wick polynomial}.
 There are at least two procedures to define them in Minkowski spacetime: one can be extended to general globally hyperbolic spacetimes and refers to {\em Hadamard states}, using the powerful machinery of {\em microlocal analysis}
(see e.g. \cite{KM}). The other, the older one, quite familiar to theoretical physicists, is easier to handle when dealing with Minkowski spacetime and the Poincar\'e-invariant vacuum state $\Omega$.
The language used here is that of quadratic forms. In this pedagogical discussion this older approach is more suitable, also because in the specific case of Minkowski spacetime it permits explicit computations.
To this end it is convenient to define a pair of quadratic forms representing the formal operators denoted by $a_{p}$ and $a^\dagger_{p}$ in the physical literature, where $p\in \sV_{m,+}$.

If $p\in \sV_{m,+}$ and $\Psi \in \gS_0$, the operator $a_{p} : \gS_0 \to \gS_0$ such that $a_p(\cH^{(n)} \cap \gS_0) \subset  \cH^{(n-1)} \cap \gS_0$ is defined as the linear extension of 
\beq\label{prop}
(a_{p}\Psi) (k_1,\ldots, k_{n-1}) := \sqrt{n} \Psi(p,k_1,\ldots, k_{n-1})\quad \mbox{for $\Psi\in \cH^{(n)}$ and $a_{p}\Omega:=0$}\:.
\eeq
A {\em quadratic form} is well defined on $\gS_0 \times \gS_0$ as
\beq \langle \Psi'| a^\dagger_{p} \Psi \rangle = \sum_{n=1}^{+\infty}
\int_{\bR^{3n}} \overline{\Psi'_n}(k_1, \ldots, k_n) (a^\dagger_{p}\Psi)_n (k_1,\ldots,k_n) d\mu_m(k_1)\cdots d\mu_m(k_n) \label{quada*}\:.
\eeq
In the integral, $a^\dagger_p : \cH^{(n)} \cap \gS_0 \to \cS'(\sV_{m,+}^{n+1})$ is the map extended by linearity, such that 
\beq
(a^\dagger_{p}\Psi) (k_1,\ldots, k_{n+1}) := \sqrt{n+1} \sum_{l=1}^{n+1} \frac{\delta(p,k_l)\Psi(k_1,\ldots, k_{l-1},k_{l+1}, \ldots,  k_{n+1})}{n+1}\:, \quad \Psi \in \cH^{(n)}   \label{a*p}
\eeq
where the Dirac delta $\delta(p,k)$ refers to the mass shell $\sV_{m,+}$ and its invariant measure, and the integral in (\ref{quada*}) has a distributional meaning, as is appropriate since $\Psi \in \gS_0$. By direct inspection we have the adjunction relation of quadratic forms:
\beq \label{aggF} \langle \Psi'| a^\dagger_{p} \Psi \rangle =  \langle  a_{p} \Psi'| \Psi \rangle  \quad \mbox{for}\quad \Psi,\Psi' \in \gS_0\:.\eeq
According to this identity we can give a general definition.\\

\begin{definition}\label{DEFF}
{\em If $N,M=0,1,2,\ldots$ and the corresponding momenta are $p_1,\ldots, p_N\in \sV_{m,+}$,
$k_1,\ldots, k_M\in \sV_{m,+}$, we define the {\bf normally ordered quadratic form},
\beq \left\langle \Psi'\left| \prod_{j=1}^N a^\dagger_{p_j}\prod_{r=1}^M a_{k_r}\Psi\right. \right\rangle := \left\langle \prod_{j=1}^N a_{p_j} \Psi'\left| \prod_{r=1}^M a_{k_r}\Psi \right.\right\rangle  \quad \mbox{for}\quad \Psi,\Psi' \in \gS_0\:,\eeq
where the right-hand side is a proper inner product and we used definition (\ref{prop}). } \hfill $\blacksquare$\\
\end{definition}

\noindent It is easy to see that, if $\Psi,\Psi' \in \gS_0$, then the following map is a function in $\cS( \sV_{m,+}^{N+M} )$:
$$ \sV_{m,+}^{N+M} \ni (p_1,\ldots, p_N, k_1,\ldots, k_M) \mapsto  \left\langle \Psi'\left| \prod_{j=1}^N a^\dagger_{p_j}\prod_{r=1}^M a_{k_r}\Psi\right. \right\rangle\in \bC \quad$$
Proposition \ref{PROPaa} in the appendix states the most important properties of normally ordered quadratic forms.
  It would not be possible to define, analogously to Definition \ref{DEFF}, quadratic forms corresponding to a symbol like $\prod_{r=1}^M a_{k_r}  \prod_{j=1}^N a^\dagger_{p_j}$ with $M,N \geq 1$ because, from our perspective, $a^\dagger_{p}$ has to be understood in the sense of quadratic forms on $\gS_0$, and objects like 
$\langle a_{p'}  \Psi'| a_{p} \Psi \rangle $ are not defined if $\Psi,\Psi' \in \gS_0$. \\

\begin{definition}  {\em The {\bf quantum-field quadratic form} is the quadratic form
\beq \langle \Psi' |\hat{\phi}(x) \Psi \rangle := \frac{1}{\sqrt{2}(2\pi)^{3/2}}\int_{\bR^3}   e^{ip\cdot x} \langle \Psi' | a_{p} \Psi \rangle +      e^{-ip\cdot x}\langle \Psi' | a^\dagger_{p} \Psi \rangle\:\: d\mu_m(p)\:,  \quad\Psi,\Psi' \in \gS_0\label{PHIx}\eeq
for $x\in \bM$.} \hfill $\blacksquare$
\end{definition}
\begin{notation} {\em The above definition is formally written as
\beq \hat{\phi}(x) := \frac{1}{\sqrt{2}(2\pi)^{3/2}}\int_{\bR^3}   e^{ip\cdot x} a_{p}  +     e^{-ip\cdot x}  a^\dagger_{p} \:\: d\mu_m(p)\label{PHIxf}\:\eeq
and we shall use this notation throughout. \hfill $\blacksquare$.}\\
\end{notation}
The right-hand side of (\ref{PHIx}) defines a smooth, bounded function of $x\in \bM $ by (\ref{fund}).
This function can therefore be smeared with functions $f\in \cS(\bM)$ and $\cD(\bM)$.\\

\begin{proposition}\label{PROP41}  
$\cD(\bM)\ni f \mapsto \langle \Psi' |\hat{\phi}[f] \Psi \rangle$ is a distribution 
in $\cD'(\bM)$ if $\Psi,\Psi' \in \gS_0$ and
\beq \langle \Psi'|\hat{\phi}[f]\Psi \rangle = \int_{\bR^4} \langle \Psi' |\hat{\phi}(x) \Psi \rangle f(x) d^4x \quad \mbox{if $f \in\cD(\bM)$, $\Psi,\Psi' \in \gS_0$.}\label{UI}\eeq
\end{proposition}

\begin{proof} (\ref{PHI}), (\ref{manc}), (\ref{PHIx}) and Fubini--Tonelli theorem yield (\ref{UI}). The map $\cD_\bR(\bR^4) \ni f \mapsto \int_{\bR^4} \langle \Psi' |\hat{\phi}(x) \Psi \rangle f(x) d^4x \in \bC$ defines a distribution in $\cD'(\bR^4)$, for every $\Psi,\Psi' \in \gS_0$, simply because $\bR^4 \ni x \mapsto  \langle \Psi' |\hat{\phi}(x) \Psi \rangle$ is smooth.
\end{proof}

\section{The stress-energy tensor operator of the Klein-Gordon field in $\bM$}\label{SEC2agg}
The construction in the previous section can be extended to {\em normally ordered Wick polynomials} in Minkowski spacetime and with respect to the Poincar\'e-invariant vacuum $\Omega$. 
We only consider two special cases, both of second order: the $\phi^2$ field and the {\em stress-energy tensor} $T_{\mu\nu}$ in Minkowski spacetime.
The completely covariant formulation in curved spacetime for Hadamard states has a rather long history (see \cite{KM} and references therein for a review); a specific discussion, together with recent results on {\em normally ordered} second-order Wick polynomials in curved spacetime referred to Hadamard states, appears in \cite{Ko2,Ko13}.

\subsection{$\phi^2$ and stress-energy-momentum tensor operators}
Classically, the {\bf stress-energy(-momentum) tensor} of a (smooth) real Klein-Gordon scalar field $\phi$ of mass $m>0$ is defined as the symmetric second-order tensor field on $\bM$ whose components in an (arbitrary) coordinate representation are
\beq\label{defT}
T_{\mu\nu}(x) :=\partial_\mu \phi(x) \partial_\nu \phi(x) -  \frac{1}{2} g_{\mu\nu} (\partial_\alpha \phi(x) \partial^\alpha \phi(x) + m^2 \phi(x)^2)\:.
\eeq
Due to the {\em Klein-Gordon equation} $(\Box - m^2)\phi =0$, the {\bf conservation equation}
\beq
\partial_\mu {T^{\mu \nu}}(x)=0\:
\eeq
is valid.
We now study the quantized version of the stress-energy tensor. Since it will be useful in Sect. \ref{sezsplit}, it is also convenient to introduce the quantum version of the square of the field $\phi^2$.\\

\begin{definition}
{\em  If $x \in \bM$, the (normally ordered) {\bf $\phi^2$ quadratic form} at $x$ is the quadratic form
$$
\gS_0\times \gS_0 \ni (\Psi',\Psi) \mapsto  \langle \Psi' |:\spa \hat{\phi}^2 \spa :(x) \Psi \rangle\:,  $$
whereas the 
{\bf stress-energy tensor quadratic form} at $x$ is the assignment, to every Minkowskian reference frame, of a corresponding set of $16$ quadratic forms for $\mu,\nu =0,1,2,3$
$$\gS_0\times \gS_0 \ni (\Psi',\Psi) \mapsto  \langle \Psi' |:\spa \hat{T}_{\mu\nu} \spa :(x) \Psi \rangle\:.$$ The right-hand sides of the formulas above are obtained by replacing, respectively, $\phi(x)$ with $\hat{\phi}(x)$ in $\phi(x)\phi(x)$ and in the expression (\ref{defT}) of $T_{\mu\nu}(x)$, then expanding $\hat{\phi}(x)$ according to (\ref{PHIxf}), then moving the operator
$a$ before the operator $a^\dagger$ in the resulting products, and finally computing the matrix element $ \langle \Psi' | \cdot  \Psi \rangle$ for $\Psi,\Psi' \in \gS_0$:
$$ \langle \Psi' |:\spa\hat{\phi}^2 \spa :(x) \Psi \rangle:=$$
$$ \int_{\sV_{m,+}^2}
\frac{ e^{i (p+k)\cdot x}\langle \Psi' | a_{p}a_{k}\Psi\rangle}{2(2\pi)^3} d^2\mu_m(p,k)+ \int_{\sV_{m,+}^2}  \frac{ e^{-i (p+k)\cdot x}   \langle \Psi' | a^\dagger_{p}a^\dagger_{k}\Psi\rangle}{2(2\pi)^3} d^2\mu_m(p,k)$$ \beq + \int_{\sV_{m,+}^2} e^{i (k-p)\cdot x} \frac{ \langle \Psi' | a^\dagger_{p}a_{k}\Psi \rangle}{2(2\pi)^3}d^2\mu_m(p,k)  +  \int_{\sV_{m,+}^2} e^{i (p-k)\cdot x} \frac{ \langle \Psi' |a^\dagger_{k}a_{p}\Psi \rangle}{2(2\pi)^3}d^2\mu_m(p,k)\:,\label{DECPhi}\eeq
$$ \langle \Psi' |:\spa \hat{T}_{\mu\nu} \spa :(x) \Psi \rangle:=$$
$$ \int_{\sV_{m,+}^2}\sp\sp
\frac{ e^{i (p+k)\cdot x}\langle \Psi' | a_{p}a_{k}\Psi\rangle}{2(2\pi)^3}  t_{\mu\nu}(p,-k)d^2\mu_m(p,k)+ \int_{\sV_{m,+}^2}\sp\sp  \frac{ e^{-i (p+k)\cdot x}   \langle \Psi' | a^\dagger_{p}a^\dagger_{k}\Psi\rangle}{2(2\pi)^3}  t_{\mu\nu}(p,-k)d^2\mu_m(p,k)$$ \beq +\spa \int_{\sV_{m,+}^2}\sp\sp \sp \sp e^{i (k-p)\cdot x} \frac{ \langle \Psi' | a^\dagger_{p}a_{k}\Psi \rangle}{2(2\pi)^3} t_{\mu\nu}(p,k)d^2\mu_m(p,k)  +  \int_{\sV_{m,+}^2}\sp\sp\sp\sp e^{i (p-k)\cdot x} \frac{ \langle \Psi' |a^\dagger_{k}a_{p}\Psi \rangle}{2(2\pi)^3} t_{\mu\nu}(p,k)d^2\mu_m(p,k)\label{DEC}\eeq
where we introduced the symmetric $(0,2)$ tensor
\beq\label{tpk}
t_{\mu\nu}(p,k) :=\frac{1}{2}[k_\mu p_\nu +  p_\mu k_\nu - g_{\mu\nu}(p\cdot k + m^2)]\:.
\eeq}
\hfill $\blacksquare$
\end{definition}
%

By construction, and taking (\ref{fund}) into account, if $\Psi,\Psi'\in \gS_0$,  the maps
 $\bM \ni x \mapsto  \langle \Psi' |:\spa\hat{\phi}^2  \spa :(x) \Psi \rangle$ and $\bM \ni x \mapsto  \langle \Psi' |:\spa \hat{T}_{\mu\nu} \spa :(x) \Psi \rangle$ are smooth bounded functions so that they can be smeared with $f \in \cD(\bR^4)$, giving rise to distributions in $\cD'(\bM)$. Furthermore, the components $  \langle \Psi' |:\spa \hat{T}_{\mu\nu} \spa :(x) \Psi \rangle$ for given $\Psi,\Psi' \in \gS_0$ and $x\in \bM$ define a symmetric $(0,2)$ tensor when the Minkowskian reference frame is changed. Notice that, once defined in Minkowskian coordinates, the quadratic form of the normally ordered stress-energy tensor operator can be defined as a general symmetric $(0,2)$ tensor, independently of the choice of the type of local coordinates, by the standard local tensor transformation law.\\

\begin{proposition}\label{PROP1} Take $f \in \cD(\bM)$ and let us refer to a given Minkowski coordinate system concerning the component indices $\mu,\nu =0,1,2,3$. \\There exist unique operators  $:\spa \hat{\phi}^2 \spa:[f] : \gS_0 \to \gF_s(\cH_m)$, 
 $:\spa \hat{T}_{\mu\nu} \spa:[f] : \gS_0 \to \gF_s(\cH_m)$ respectively called 
 (smeared normally ordered) {\bf $\phi^2$ operator} and
(smeared normally ordered) {\bf stress-energy tensor operator} such that 
\begin{align}
 \langle \Psi' |:\spa \hat{\phi}^2 \spa:[f] \Psi \rangle &= \int_{\bR^4} f(x)   \langle \Psi' |:\spa \hat{\phi}^2 \spa :(x) \Psi \rangle d^4x\quad 
\forall \Psi',\Psi \in \gS_0\:, \label{PP}\\
 \langle \Psi' |:\spa \hat{T}_{\mu\nu} \spa:[f] \Psi \rangle &= \int_{\bR^4} f(x)   \langle \Psi' |:\spa \hat{T}_{\mu\nu} \spa :(x) \Psi \rangle d^4x\quad 
\forall \Psi',\Psi \in \gS_0\:. \label{TT}
\end{align}
Therefore $\cD(\bM)  \ni f \mapsto \langle \Psi' |:\spa \hat{\phi}^2\sp :[f] \Psi \rangle$ and $\cD(\bM)  \ni f \mapsto \langle \Psi' |:\spa \hat{T}_{\mu\nu}\sp :[f] \Psi \rangle$ belong to $\cD'(\bM)$.\\
The following further facts are true.
\begin{itemize}
\item [(a)] $\gS_0$ is invariant under $:\spa \hat{\phi}^2 \spa:[f]$ and $:\spa \hat{T}_{\mu\nu} \spa:[f]$.
\item [(b)]  The said operators admit adjoints and 
$$:\spa \hat{\phi}^2 \spa:[\overline{f}] \subset :\spa\hat{\phi}^2 \spa:[f]^\dagger \:, \quad :\spa \hat{T}_{\mu\nu} \spa:[\overline{f}] \subset :\spa \hat{T}_{\mu\nu} \spa:[f]^\dagger$$
so that, in particular, $:\spa \hat{\phi}^2 \spa:[f]$ and $:\spa \hat{T}_{\mu\nu} \spa:[f]$ are symmetric if $f\in \cD_\bR(\bM)$.

\item[(c)] 
$:\spa \hat{T}_{\mu\nu} \spa:[f] =\: \:  :\spa \hat{T}_{\nu\mu} \spa:[f] \quad \forall f \in \cD(\bM)\:.$
\item[(d)] The representation (\ref{rep}) of $IO(1,3)_+$ acts covariantly on the stress-energy tensor operator:
 if $g\equiv (\Lambda, a)\in IO(1,3)_+$ and defining $(g_*f)(x) := f(g^{-1}x)$ for $x\in \bM$,
\beq
\label{covT}
U_g :\spa \hat{T}_{\mu\nu} \sp:[f] U_g^\dagger &={(\Lambda^{-1})_\mu}^\alpha {(\Lambda^{-1})_\nu}^\beta :\spa \hat{T}_{\alpha\beta} \spa:[g_* f] \quad\mbox{if  $f\in  \cD(\bM)$.}
\eeq
Analogously,
\beq
\label{covT0}
U_g :\spa \hat{\phi}^2 \spa:[f] U_g^\dagger &= \: :\spa \hat{\phi}^2 \spa:[g_* f] \quad\mbox{if  $f\in  \cD(\bM)$.}
\eeq
\item[(e)] The stress-energy tensor operator is conserved in the distributional sense:
\beq
:\spa \hat{T}_{\mu\nu} \spa:[\partial^\mu f] =0\:, \quad \mbox{for $ f \in \cD(\bM)$}. 
\eeq
\item[(f)] $:\spa\hat{T}_{\mu\nu} \spa:[f] u^\mu v^\nu$ and $:\spa \hat{\phi}^2\spa:[f]$ are essentially selfadjoint if
$u,v$ are arbitrary vectors, 
 $f\in \cD(\bM)$ is real.
\end{itemize}
\end{proposition}

\begin{proof} See Appendix \ref{AAA}. \end{proof}

\begin{remark} \label{REMREF} {\em
The self-adjointness of symmetric operators like the averaged stress tensor is a difficult 
issue  in general globally hyperbolic spacetimes. However, essential self-adjointness of $:\spa\hat{\phi}^2\spa:[g]$ was established by Sanders in \cite{Ko2} in general globally hyperbolic spacetimes for any Hadamard reference state and in the  domain $\gG_0$  of this operator
when $g = \sum_{i=1}^Nf^2_i$ for some finite $N$, where each $f_i \in \cD(\bM)$ is real.
 In Minkowski spacetime, and more generally in static spacetimes referring to the unique ground state,  
the situation is much better also for the normally ordered stress-energy tensor $:\spa\hat{T}_{\mu\nu} \spa:[f]$, because one can exploit the self-adjointness of the Hamiltonian $H$ together with commutator arguments as in Sec.X.5 of \cite{RS2}, particularly Thm.X.37. An application of this argument  to establishing essential selfadjointness of the stress--energy tensor of a massive free scalar field  appears\footnote{The author is grateful to an anonymous referee for remarking this fact and for pointing out  \cite{Ko13}.} in Thm. 5.2 of Sanders' work
\cite{Ko13}.} \hfill $\blacksquare$
\end{remark}

\subsection{Locality/commutativity properties of $\hat{\phi}[f]$, $:\spa\hat{\phi}^2\spa : [f]$, and $:\spa \hat{T}_{\mu\nu}\spa:[f]$}\label{sezsplit}
Once we have defined $\hat{\phi}[f]$, $ :\spa \hat{\phi}^2 \spa:[f]$, $ :\spa \hat{T}_{\mu\nu} \spa:[f]$ and the associated quadratic forms, we move on to show another technically fruitful way to compute them in terms of a smearing procedure with compactly supported {\em distributions} instead of compactly supported smooth functions.
First of all, we define the {\bf normally ordered product of two field operators} smeared with $f,g \in \cD(\bM)$ as the operator with domain $\gF_0$
\beq :\spa \hat{\phi}[f]\hat{\phi}[g]\spa: \: := \hat{\phi}[f]\hat{\phi}[g] - \langle \Omega | \hat{\phi}[f]\hat{\phi}[g] \Omega\rangle I\:.\label{NORP}\eeq
This type of definition, when dealing with algebraic states of Gaussian type, can easily be extended to products of many fields by taking advantage of the so-called {\em Wick rule} (see e.g. \cite{KM}). We stick to the elementary case above, since it is sufficient for our purposes.
Accordingly, if $\Psi \in \gF_0$, the maps
$$\cD(\bM)\times \cD(\bM) \ni (f,g)\mapsto \langle \Psi |  \hat{\phi}[f]\hat{\phi}[g] \Psi \rangle\:, \quad 
\cD(\bM)\times \cD(\bM) \ni (f,g)\mapsto \langle \Psi |  :\spa\hat{\phi}[f]\hat{\phi}[g]\spa: \Psi \rangle$$ are respectively the {\bf 2-point function} and 
the {\bf normally ordered 2-point function} of the state represented by $\Psi$.\\

\begin{lemma}\label{lemmad} If $f,g \in \cD(\bM)$,
$:\spa \hat{\phi}[f]\hat{\phi}[g]\spa:$ can be obtained from $\hat{\phi}[f]\hat{\phi}[g]$ by expanding the field operators according to (\ref{PHI}) and moving  
$a$ to the right of  $a^\dagger$:
\beq:\spa \hat{\phi}[f]\hat{\phi}[g]\spa: = \frac{1}{2} \left(a(\kappa_m\overline{f})a(\kappa_m\overline{g}) + a^\dagger(\kappa_m g)a(\kappa_m\overline{f})
+ a^\dagger(\kappa_m f)a(\kappa_m\overline{g})  + a^\dagger(\kappa_m f)a^\dagger(\kappa_m g)\right)\:,\label{mat}\eeq
so that, in particular, $:\spa \hat{\phi}[f]\hat{\phi}[g]\spa: = :\spa \hat{\phi}[g]\hat{\phi}[f]\spa:$.
\end{lemma}

\begin{proof} See Appendix \ref{AAA}.
\end{proof}
The notions introduced allow one to compute $ \langle \Psi | :\spa \hat{\phi}^2\spa:[f] \Psi'\rangle$ and $ \langle \Psi | :\spa \hat{T}_{\mu\nu}\spa:[f] \Psi'\rangle$ by means of a procedure known as {\em point-splitting} which in particular relies upon  Schwartz' kernel theorem.\\

\begin{proposition}\label{PROP2X}
If $\Psi,\Psi' \in \gS_0$, the map $\cD(\bM)\times \cD(\bM) \ni (f,g) \mapsto \langle \Psi| :\spa \hat{\phi}[f]\hat{\phi}[g]\spa: \Psi'\rangle$ is the restriction to $\cD(\bM)\times \cD(\bM)$ of a unique distribution in $\cD'(\bM\times \bM)$ which is a smooth bounded function, denoted by $\bM\times \bM \ni (x,y) \mapsto \langle \Psi | :\spa \hat{\phi}(x)\hat{\phi}(y)\spa: \Psi'\rangle$, symmetric under interchange of its arguments. Furthermore,
\beq\label{ONE}
\int f(x)\delta(x,y)  \langle \Psi | :\spa \hat{\phi}(x)\hat{\phi}(y)\spa: \Psi'\rangle d^4xd^4y =  \langle \Psi | :\spa \hat{\phi}^2\spa:[f] \Psi'\rangle\quad \mbox{if $f\in \cD(\bM)$}.
\eeq
Referring to Minkowskian coordinates,
\beq\label{TWO}
\int f(x)\delta(x,y) D_{\mu\nu}(x,y)  \langle \Psi | :\spa \hat{\phi}(x)\hat{\phi}(y)\spa: \Psi'\rangle d^4xd^4y =  \langle \Psi | :\spa \hat{T}_{\mu\nu}\spa:[f] \Psi'\rangle\quad \mbox{if $f\in \cD(\bM)$,}
\eeq
where we introduced the (formally selfadjoint) second-order differential operator
\beq \label{Dmunu} D_{\mu\nu}(x,y) := \frac{1}{2}\left[\partial_{x^\mu} \partial_{y^\nu} + \partial_{x^\nu} \partial_{y^\mu}
 -g_{\mu\nu} \left( g^{\alpha\beta}\partial_{x^\alpha} \partial_{y^\beta} + m^2\right)\right].\eeq
\end{proposition}

\begin{proof} See Appendix \ref{AAA}.
\end{proof}
The results just obtained imply a first crucial fact about {\em relativistic  locality} in the spirit of the AHK approach. Here the smearing procedure reveals its physical importance.\\

\begin{proposition}\label{PROP4} If ${\cal O} \subset \bM$ is open, $\cT({\cal O})$ denotes the unital $*$-algebra of operators
on $\gF_s(\cH_m)$
with common invariant domain $\gS_0$ generated by (a) $\hat{\phi}$, (b) $:\spa \hat{\phi}^2 \spa:$, (c) $ :\spa \hat{T}_{\mu\nu} \spa:$ smeared with test functions supported in ${\cal O}$. If the open sets ${\cal O}_1\subset \bM$ and ${\cal O}_2\subset \bM$ are causally separated, then
$$[A_1,A_2]=0 \quad \mbox{for $A_1\in \cT({\cal O}_1)$ and $A_2\in \cT({\cal O}_2)$. }$$
\end{proposition}

\begin{proof} See Appendix \ref{AAA}.
\end{proof}

\subsection{Energy inequalities and Quantum Energy Inequalities}
Again taking advantage of Proposition \ref{PROP2X}, we move on to consider energy inequalities. We start with the observation that the classical stress-energy tensor (\ref{defT}) enjoys two important properties.
As is well known, the {\bf four-momentum density} in the Minkowskian reference frame $u \in \sT_+$ is defined as $P_u^\mu(x) := T^{\mu\nu} (x) u_\nu  =T^{\nu\mu} (x) u_\nu$.
Some computations based on the explicit expression (\ref{defT}) imply that
\beq\label{genEIQ2}
T^{\mu\nu}(x) u_\mu u'_\nu \geq 0 \quad   \mbox{if $u, u'\in \sT_+$.}
\eeq
We observe that, in curved spacetime, $u$ must be taken to be a {\em Killing vector} in order to define a conserved four-momentum density $P_u$, whereas $u'$ may be chosen as any future-directed timelike vector. In general, their roles cannot be interchanged in curved spacetime, although (\ref{genEIQ2}) remains valid in either case.
The following elementary fact is true.\\
 \begin{proposition}\label{PE}
 A symmetric tensor $T$ in $\sV^*\otimes \sV^*$ satisfies $T_{\mu\nu} u^\mu u'^\nu \geq 0$ for every pair $u,u' \in \sT_+$ if and only if $J_u \in \overline{\sV_+}$ for every $u\in \sT_+$, where $J^\nu_u := -u^\mu {T_{\mu}}^\nu$. \\
\end{proposition}

\noindent As a consequence, the four-momentum density $P_u(x)$ is causal and future-directed wherever it does not vanish, for every choice of reference frame $u \in \sT_+$ and also, by linearity and taking an obvious limit, for $u\in \overline{\sV_+}$.
%
%
 Inequality (\ref{genEIQ2}) also makes explicit the requirement of {\em positive energy density} of $P_u$ in the reference frame $u'$, which generalizes\footnote{This latter condition, even if valid for every $u\in \sT_+$, does not imply that $J_u$ is causal and future-directed: $T\equiv diag(0,1,0,0)$ in Minkowskian coordinates is a counterexample.} $T^{\mu\nu}(x) u_\mu u_\nu\geq 0$.

It is known that, in general, energy positivity requirements fail when we pass to the quantum regime and, even in curved spacetime and considering a covariant notion of normally ordered stress-energy operator, only lower bounds for the expectation value of the energy are valid. This is true for Hadamard and adiabatic algebraic states (see \cite{F12} for an excellent exhaustive review). 
However, for the Klein-Gordon (real massive) quantum field the inequalities above are still valid in terms of expectation values when one explicitly refers to {\em $n$-particle states} $\Psi \in \cH^{(n)}  \cap \gS_0= \cS(\sV^n_{m,+})$, as already noted and used in \cite{Terno,M23} -- only for $n=1$ -- in more elementary versions of the result below.\\

\begin{proposition}  \label{TEOB1} Consider a real scalar Klein-Gordon field on $\bM$ with mass $m>0$ and the associated 
normally ordered stress-energy tensor operator $:\spa \hat{T}^{\mu\nu}\spa:[f]$ of Proposition \ref{PROP1}, whose components are referred to a given Minkowski coordinate system. Consider a closed subspace $\cK := \oplus_{n\in J} \cH^{(n)}$
for an index set $J \subset \{0,1,2 , \ldots\}$  such that
 $|n-n'| \neq 2$ if $n,n' \in J$ -- in particular $\cK := \cH^{(n)}$ for some $n\in \bN$.
If $f\in \cD_\bR(\bM)$ satisfies $f\geq 0$ and 
$\Psi \in  \cK \cap \gS_0$,
%
%
 then
\beq\label{bff1}
u_\mu u'_\nu \langle \Psi | :\spa \hat{T}^{\mu\nu} \spa: [f] \Psi \rangle \geq 0, \quad \mbox{for every $u, u' \in \overline{\sV_+}$.}
\eeq
and thus 
$u_\mu  \langle \Psi | :\spa \hat{T}^{\mu\nu} \spa: [f] \Psi \rangle$ is causal and future-directed, or vanishes, if $u\in \overline{\sV_+}$.\\
The properties above are also valid if one replaces $ :\spa \hat{T}^{\mu\nu} \spa: [f] $ by $:\spa \hat{T}^{\mu\nu} \spa: (x)$ with $x\in \bM$.
\end{proposition}

\begin{proof} 
According to the structure of the stress-energy operator as presented in the proof of Proposition \ref{PROP1}, $\langle \Psi | :\spa \hat{T}^{\mu\nu} \spa: [f] \Psi' \rangle =0$ if $\Psi \in \cH^{(n)}$, $\Psi'\in \cH^{(n')}$ 
with $n\neq n'$
and 
$|n-n'| \neq 2$. Thus, under our hypotheses, $u_\mu u'_\nu \langle \Psi | :\spa \hat{T}^{\mu\nu} \spa: [f] \Psi'\rangle =\sum_n u_\mu u'_\nu\langle \Psi_n | :\spa \hat{T}^{\mu\nu} \spa: [f] \Psi_n\rangle $.
It is therefore sufficient to prove the thesis for $\cK := \cH^{(n)}$. The case $n=0$ is obvious.
 If $\Psi \in \cH^{(n)} \cap \gS_0$, $n=1,2,\ldots$, define, where $Q:= (q_1,\ldots, q_{n-1}) \in \sV_{m,+}^{n-1}$,
\beq \psi(x,Q) :=\frac{1}{\sqrt{2}(2\pi)^{3/2}} \int_{\sV_{m,+}} \Psi(p,Q) e^{-ip\cdot x} d\mu_m(p)\:.\eeq
By direct inspection, taking advantage of (\ref{TT}) and (\ref{DEC}), where only the last two addends in the expansion matter, one sees that, if $\Psi, \Psi'\in \cH^{(n)} \cap \gS_0$ and $f\in \cD_\bR(\bM)$, then
$$\langle \Psi| :\spa \hat{T}_{\mu\nu}\spa:[f] \Psi' \rangle = \frac{n}{2}\int_{\bM\times \sV_{m,+}^{n-1}} \sp\sp \sp f(x)\left[\partial_\mu \overline{\psi} \partial_\nu \psi'
+ \partial_\nu \overline{\psi} \partial_\mu \psi' - g_{\mu\nu}\left(\partial_\alpha \overline{\psi} \partial^\alpha \psi'
+ m^2 \overline{\psi} \psi'\right)\right] d^4x dQ\:,$$
where $dQ	:= d^{n-1}\mu_m(Q)$.
Now assume $\Psi'=\Psi$. Using $f\geq 0$,
\begin{align} \nonumber
\langle \Psi| :\spa \hat{T}_{00} \spa: [f]\Psi \rangle &=  \frac{n}{2}\int f(x)\left(\sum_{\mu=0}^3\partial_\mu \overline{\psi} \partial_\mu \psi
+ m^2 \overline{\psi} \psi\right)d^4xdQ \geq 0\:, \\ 
\langle \Psi | :\spa \hat{T}_{00} \spa: [f] \Psi \rangle  - \langle \Psi | :\spa \hat{T}_{11} \spa: [f] \Psi\rangle  &=n \int f(x)\left(\sum_{\mu=2}^3\partial_\mu \overline{\psi} \partial_\mu \psi
+ m^2 \overline{\psi} \psi\right)d^4xdQ
\geq 0 \nonumber\:,  \\
\langle \Psi | :\spa \hat{T}_{00} \spa: [f] \Psi \rangle  + \langle \Psi | :\spa \hat{T}_{11} \spa: [f] \Psi\rangle  &=n \int f(x)\sum_{\mu=0}^1\partial_\mu \overline{\psi} \partial_\mu \psi
d^4x dQ
\geq 0\nonumber \:.
\end{align}
so that 
$ \langle \Psi | :\spa \hat{T}_{00} \spa: [f] \Psi \rangle \geq |\langle \Psi | :\spa \hat{T}_{11} \spa: [f] \Psi \rangle|$. This bound is obviously valid in every Minkowskian coordinate system. 
If $u,u' \in \sT_+$, we can choose a Minkowskian coordinate system $x'^0,x'^1,x'^2,x'^3$ such that $u= c\partial'_{0} + s\partial'_{1}$ and $u'= c\partial'_{0} -  s\partial'_{1}$, where $s:= \sinh \chi$, $c:= \cosh \chi$ for some $\chi\in \bR$.
Since $\sinh^2\chi \leq \cosh^2 \chi$ and the two bounds above hold, we have
$u_\mu u'_\nu \langle \Psi | :\spa \hat{T}^{\mu\nu} \spa: [f] \Psi \rangle
=c^2\langle \Psi | :\spa \hat{T}'_{00} \spa: [f] \Psi \rangle- s^2 \langle \Psi | :\spa \hat{T}'_{11} \spa: [f] \Psi \rangle\geq 0$ (when $f\geq 0$).
This result concludes the proof, since the case $u,u'\in \overline{\sV_+}$ is obtained by linearity and continuity; the penultimate statement immediately follows from Proposition \ref{PE}, and the last one is a direct consequence of (\ref{TT}) together with the smoothness of $\bM\ni x \mapsto \langle \Psi| :\spa \hat{T}^{\mu\nu} \spa: (x)\Psi \rangle$.
\end{proof}

Let us pass to the case $\Psi \in \gS_0$, where arbitrary superpositions of components with different particle numbers are allowed.
Now the proof above fails because in the expansion (\ref{DEC}) the first two addends do not vanish in general. Actually, as is in particular suggested in \cite{F12}, this is a general fact. The general result, valid for positive symmetric operators, directly follows from the already mentioned {\em Reeh-Schlieder property} (d) of Proposition \ref{PROP35} and the fact that symmetric positive operators admit positive selfadjoint extensions. 
We prove this rigorously. \\

\begin{proposition}\label{XYZ} Consider the $*$-algebra $\cT({\cal O})$ of operators with common invariant dense domain $\gS_0$, 
smeared with test functions supported in ${\cal O}$, where ${\cal O}\subset \bM$ is a bounded open set as in Proposition \ref{PROP4}. If $A \in \cT({\cal O})$ is symmetric, $A\geq 0$, and $\langle \Omega | A\Omega\rangle =0$, then $A=0$.\\
Consequently, for $u,u'\in \sV$ and $f\in \cD_\bR(\bM)$ with $f \not \equiv 0$ and $f\geq 0$, the inequalities
$u_\mu u'_\nu :\spa \hat{T}^{\mu\nu}\spa:[f]\geq 0$ and $:\spa \hat{\phi}^2\spa:[f]\geq 0$ cannot hold unless the corresponding quadratic-form operator vanishes identically. \end{proposition}

\begin{proof} Since the densely defined operator $A$ is symmetric and $A\geq 0$, it admits its positive Friedrichs self-adjoint extension $A_F\geq 0$.
The rest of the proof is, however, also valid if $A_F$ denotes any {\em positive} self-adjoint extension of $A$. In the cases $u_\mu u'_\nu :\spa \hat{T}^{\mu\nu}\spa:[f]$ and $:\spa \hat{\phi}^2\spa:[f]$, the Friedrichs self-adjoint extensions coincide with the respective closures, since these operators are essentially self-adjoint.
 Since $\Omega \in D(A) \subset D(A_F)\subset D(\sqrt{A_F})$, spectral calculus (see e.g. \cite{Moretti1}) yields
  $||\sqrt{A_F}\Omega||^2 =\langle \Omega|A_F \Omega\rangle = \langle \Omega|A \Omega\rangle =0$. Hence $\sqrt{A_F}\Omega=0$ and therefore $A_F \Omega=A_F^{1/2}A_F^{1/2}\Omega =0$. At this point, consider an open non-empty set ${\cal O}_1\subset \bM$ which is causally separated from ${\cal O}$. According to Proposition \ref{PROP4}, we have $A_F B \Omega=A B \Omega = BA\Omega =0$ for every element $B$ of the sub-$*$-algebra of operators generated by $\hat{\phi}[h]$ with smearing functions satisfying $\operatorname{supp}(h) \subset {\cal O}_1$.  
 If $\Psi\in D(A_F)$, there is a sequence of elements $B_n\Omega \to \Psi$, since the set of vectors $B\Omega$ as above is dense in $\cF_s(\cH_m)$ due to the aforementioned Reeh-Schlieder property ((d) of Proposition \ref{PROP35}). In summary, $B_n\Omega \to \Psi$ and the sequence of images  $A_F B_n\Omega= 0$ trivially converges to $0$.
 Since $A_F=A_F^\dagger$ is closed, we conclude that $A_F \Psi=\lim_n A_F B_n\Omega= 0$ for every $\Psi \in D(A_F)$. In particular, $A =0$,
  since it is a restriction of $A_F$.
%
%
%
%
For the last assertion, take a bounded open set ${\cal O}\supset \operatorname{supp}(f)$ and observe that the displayed quadratic-form operators have vanishing vacuum expectation value:
$\langle \Omega| u_\mu u'_\nu :\spa \hat{T}^{\mu\nu}\spa:[f]\Omega \rangle  = 0$ and 
$\langle \Omega| :\spa \hat{\phi}^2\spa:[f]\Omega \rangle = 0$. Hence, if either of them were positive, the first part of the proposition would force it to vanish identically.
\end{proof}
%

We shall confine our investigation to Minkowski spacetime in the representation of the Poincar\'e-invariant quasifree state, and we shall study the quantum version of condition (\ref{bff1}). We shall prove a theorem of crucial relevance for the construction of our localization POVMs; this result will later be refined in Theorem \ref{TEOFD}, where the form of the quadratic form $b^{\mu\nu}_f$ will be made more explicit through a suitable choice of the smearing function $f$.
The results presented below rely on important known achievements \cite{FS08} (see also \cite{Ko} for a recent extension). Although we restrict ourselves to flat spacetime and to a distinguished reference vacuum state, the results stated in Proposition \ref{TEOB} below should extend to Hadamard quasifree states in curved spacetime.
We shall use a special case of the main result obtained in \cite{FS08}, where the metric-signature convention is the opposite of ours.  
We start with the observation that a {\em normalized} vector $\Psi \in \gS_0$ defines an algebraic state of {\em Hadamard type} \cite{FS08}. Indeed, the Minkowski vacuum $\Omega$ is Hadamard, as is well known. Moreover, every {\em normalized} $\Psi \in \gS_0$ is such that the difference $\Lambda_\Psi-\Lambda_{\Omega}$ of the associated two-point functions defines a smooth integral kernel, by Proposition \ref{PROP2X}. Hence $\Psi$ is Hadamard as well, by definition. 

Take $\Psi \in \gS_0$ with $||\Psi||=1$, and let
$Q$ be an at most first-order differential operator with smooth real coefficients, so that
\beq (Q(f))(x) =  (A(f))(x) + B(x)f(x) \quad \mbox{for $f\in \cD(\bM)$ and $x \in \bM$}\label{QQ}\eeq
where $A$ is a smooth real vector field and $B$ is a smooth real scalar field.
Following the procedure outlined in the proof of Theorem 3.1 of \cite{FS08} -- which is more generally valid in suitably shaped domains of globally hyperbolic spacetimes and for Hadamard states -- and taking into account the comments in Section 3 of \cite{F2}, the result of the quoted theorem takes the following form\footnote{I am grateful to C. Fewster for clarifying this point to me.}. 

A choice of a Minkowskian coordinate system
$$\kappa : \bM \ni q \mapsto (x^0(q),x^1(q), x^2(q),x^3(q))\in  \bR^4$$
defines a bilinear map $\cD(\bM) \times \cD(\bM)\to \bC$, real-valued on $\cD_\bR(\bM)\times \cD_\bR(\bM)$, as follows. Consider the unnormalized Fourier transform, {\em computed in this coordinate system} after identifying $\bM$ with $\bR^4$:
$$\hat{g}(p,p') :=  \int_{\bR^4\times \bR^4} e^{-ip\cdot x}  e^{-ip'\cdot x'}  g(\kappa^{-1}(x), \kappa^{-1}(x')) d^4x d^4x'\:, \quad g \in \cS(\bM \times \bM)\:.$$
Here the dot denotes the standard scalar product in $\bR^4$. At this point, extending the Fourier transform above to distributions in the standard way, define
\beq\label{FSresmod}
b^\kappa_Q(h,h') :=\frac{2}{(2\pi)^4} \int_{\bR_+ \times \bR^3} \widehat{[(h\otimes h') (Q \otimes Q) \Lambda_{\Omega}]_\kappa} (-p,p) d^4p\:, \quad h,h' \in \cD(\bM)\:,
\eeq
where $p \equiv (p_0,p_1,p_2,p_3)$. The right-hand side is well defined for the Minkowski vacuum state $\Omega$, which is Hadamard.
With these definitions, the following inequality is valid for every $f\in \cD_\bR(\bM)$:
\beq\label{BOUND}
\int_\bM f(x)^2 Q\otimes Q (\Lambda_\Psi(x,y)- \Lambda_{\Omega}(x,y))|_{x=y} d^4x \geq -b^\kappa_Q(f,f)\:.
\eeq
The crucial fact in (\ref{BOUND}) is that the lower bound found in this way does not depend on the normalized vector $\Psi \in \gS_0$.  
The estimate is not optimal; moreover, the definition of $b^\kappa_Q$ involves an explicit choice of reference frame, so different choices of $\kappa$ may give different bilinear maps. Each of them, however, satisfies (\ref{BOUND}), whose left-hand side is {\em independent} of the specific choice of $\kappa$.  

A discussion of how to make canonical choices of coordinates in a much more general context is given in \cite{FS08}, to which we refer for further analysis.\\

An immediate consequence of the inequality just obtained is the following. If a finite number $N$ of differential operators $Q_a$ as in (\ref{QQ}) are given, together with constants $d_a\geq 0$, then
\beq\label{BOUNDGGG}
\int_\bM f(x)^2\left(\sum_{a=1}^N d_a Q_a\otimes Q_a\right) (\Lambda_\Psi(x,y)- \Lambda_{\Omega}(x,y))|_{x=y}  d^4x \geq -b^\kappa_{\{d_a,Q_a\}_{a=1}^N}(f,f)\:,
\eeq
where
$$b^\kappa_{\{d_a,Q_a\}_{a=1}^N}(f,f):=\sum_{a=1}^N d_a b^\kappa_{Q_a}(f,f)\:.$$
This follows from the linearity of the integral and of the Fourier transform.\\
 
\begin{remark} {\em It is interesting to observe that the smearing function $f^2$ in (\ref{BOUND}) cannot, in general, be replaced by a non-negative $f\in \cD_\bR(\bM)$ by inserting $\sqrt{f}$ in place of $f$ on the right-hand side of (\ref{BOUND}). Indeed, for $0\leq f\in \cD(\bM)$ one generally has $\sqrt{f}\not \in \cD(\bM)$\footnote{This may fail even if $f$ vanishes with all derivatives at its zeros, as shown by a classical counterexample \cite{NO}.}.
Proposition \ref{TEOB} below is therefore different from Proposition \ref{TEOB1}, precisely because it adopts a stronger hypothesis on the positive smearing functions.
On the other hand, the main result of \cite{FS08} recalled above can easily be extended to the case in which $f^2$ is replaced by a finite sum $\sum_i f_i^2$, with
$f_i \in \cD_\bR(\bM)$. This weaker requirement on the smearing functions was used in \cite{Ko2} to prove essential-selfadjointness results for second-order normally ordered Wick polynomials in curved spacetime with respect to Hadamard states. However, as discussed there, there are functions $f\geq 0$, $f\in \cD(\bM)$, which are not of the form $\sum_i f_i^2$.} \hfill $\blacksquare$\\
\end{remark}

\noindent We now come to the announced proposition. \\

\begin{proposition}\label{TEOB} Consider a real scalar Klein-Gordon field on $\bM$ with mass $m>0$ and the associated 
normally ordered stress-energy tensor operator $:\spa \hat{T}^{\mu\nu}\spa:[f]$ of Proposition \ref{PROP1}, whose components are referred to a given Minkowski coordinate system. 
%
%
If, in this Minkowskian coordinate system, we define
$$b_f^{\mu\nu} := \frac{2}{(2\pi)^4} \int_{\bR_+ \times \bR^3} \widehat{[(f\otimes f) D^{\mu\nu} \Lambda_{\Omega}]_\kappa} (-p,p) d^4p\:, \quad f \in \cD_\bR(\bM)\:,$$
then the following facts hold.
    \begin{itemize}
\item[(a)] If $u, u' \in \overline{\sV_+}$ and $f\in \cD_\bR(\bM)$, then
\beq\label{bff}
u_\mu u'_\nu \langle \Psi | :\spa \hat{T}^{\mu\nu} \spa: [f^2] \Psi \rangle \geq -  u_\mu u'_\nu  b_f^{\mu\nu} ||\Psi||^2\:,\quad \mbox{for every $\Psi \in \gS_0$,}
\eeq
where necessarily $u_\mu u'_\nu  b_f^{\mu\nu} >0$ if $u_\mu u'_\nu  :\spa \hat{T}^{\mu\nu} \spa: [f^2] \neq 0$.
    \item[(b)] If $u \in \sT_+$, $f\in \cD_\bR(\bM)$, and $\Psi \in \gS_0$, the vector 
    $$
-  u_\mu  \left( \langle \Psi |  :\spa \hat{T}^{\mu \nu} \spa: [f^2] \Psi \rangle + b_f^{\mu\nu} ||\Psi||^2 \right) 
$$
is causal and future-directed, or vanishes.
\end{itemize}
These results remain valid if $f^2$ is replaced by $\sum_{i=1}^N f_i^2$ with $f_i\in \cD_\bR(\bM)$, with the obvious redefinition of $ b_f^{\mu\nu}$.
\end{proposition}

\begin{proof} Item (b) is a direct consequence of (a) and Proposition \ref{PE}. We therefore prove (a). From now on we assume $||\Psi||=1$, since validity of (a) in this case immediately implies the general case.
We start by proving that, if $u,u' \in \sT_+$, the differential operator
\beq D(u,u'):= u_\mu u'_\nu D^{\mu\nu}: \cD(\bM \times \bM) \to \cD(\bM \times \bM)\label{dudu}\eeq
(where we use an {\em abstract index notation}, so that no coordinate system has been chosen) can be decomposed as 
\beq D(u,u') = \sum_{a=0}^4  d^{u,u'}_a Q^{u,u'}_a\otimes Q^{u,u'}_a\label{expDp}\eeq
for suitable at most first-order differential operators $Q^{u,u'}_a$ as in (\ref{QQ}) and suitable coefficients $d^{u,u'}_a > 0$.
%
Suppose that $u,u' \in \sT_+$ are given. 
Noticing that $u+u'$ and $u-u'$ are orthogonal,
we can choose $x'^0$ parallel to $u+u'$ and $x'^1$ parallel to $u-u'$, obtaining
%
 $$u= c_{u,u'}\partial_{x'^0} + s_{u,u'}\partial_{x'^1}, \quad u'= c_{u,u'}\partial_{x'^0} -  s_{u,u'}\partial_{x'^1}$$
where $s_{u,u'}:= \sinh \chi_{u,u'}$ and $c_{u,u'}:= \cosh \chi_{u,u'}$ for some $\chi_{u,u'}\in \bR$. We finally define the differential operators $Q^{u,u'}_\mu :=  \partial_{x'^\mu}$ and $Q^{u,u'}_4 := mI$.
With the definitions above,
by direct inspection,
$$ D(u,u')= c_{u,u'}^2D'_{00} -s_{u,u'}^2 D'_{11} = \frac{c_{u,u'}^2-s_{u,u'}^2}{2}(Q^{u,u'}_0\otimes Q^{u,u'}_0 + Q^{u,u'}_1\otimes Q^{u,u'}_1)$$ $$ + \frac{c_{u,u'}^2+s_{u,u'}^2}{2}
(Q^{u,u'}_2\otimes Q^{u,u'}_2 + Q^{u,u'}_3\otimes Q^{u,u'}_3 + Q^{u,u'}_4 \otimes Q^{u,u'}_4)\:.$$
We have proved, in particular, that (\ref{expDp}) holds with $$d^{u,u'}_0= d^{u,u'}_1=\frac{1}{2}\quad \mbox{and}\quad d^{u,u'}_2= d^{u,u'}_3=d^{u,u'}_4 =\frac{1}{2}( \cosh^2 \chi_{u,u'} + \sinh^2 \chi_{u,u'})>0\:.$$
We are now in a position to apply (\ref{BOUNDGGG}) for a chosen coordinate system $\kappa : \bM \ni q \mapsto (x^0(q),x^1(q),x^2(q),x^3(q))\in \bR^4$. Since (\ref{dudu}) and (\ref{expDp}) are valid, the left-hand side of (\ref{BOUNDGGG}) is just 
$$u_\mu u'_\nu \langle \Psi | :\spa \hat{T}^{\mu\nu} \spa: [f^2] \Psi \rangle\:.$$
The right-hand side is 
$$-b^\kappa_{\{d^{u,u'}_a,Q^{u,u'}_a\}_{a=0}^4}(f,f) =
-\frac{2}{(2\pi)^4} \int_{\bR_+ \times \bR^3} \widehat{\left[(f\otimes f) \sum_{a=0}^4  d^{u,u'}_a Q^{u,u'}_a\otimes Q^{u,u'}_a \Lambda_{\Omega}\right]_\kappa} (-p,p) d^4p$$
$$=-\frac{2}{(2\pi)^4} \int_{\bR_+ \times \bR^3} \widehat{[(f\otimes f) u_\mu u'_\nu D^{\mu\nu} \Lambda_{\Omega}]_\kappa} (-p,p) d^4p$$ $$=- u_\mu u'_\nu \frac{2}{(2\pi)^4} \int_{\bR_+ \times \bR^3} \widehat{[(f\otimes f) D^{\mu\nu} \Lambda_{\Omega}]_\kappa} (-p,p) d^4p\:.$$
Inserting these results in (\ref{BOUNDGGG}) yields (\ref{bff}) for $||\Psi||=1$ and $u,u'\in \sT_+$.
By homogeneity in $u$ and $u'$, the same conclusion holds for arbitrary timelike future-directed $u,u'\in\sV_+$. The case $u,u'\in \overline{\sV_+}$ follows by continuity, the cases in which one of the two vectors vanishes being trivial.
%

If $u_\mu u'_\nu  :\spa \hat{T}^{\mu\nu} \spa: [f^2]\neq 0$, then necessarily $u_\mu u'_\nu b_f^{\mu\nu}>0$. Indeed, if $u_\mu u'_\nu b_f^{\mu\nu}\leq 0$, (\ref{bff}) would imply that the operator $u_\mu u'_\nu  :\spa \hat{T}^{\mu\nu} \spa: [f^2]$ is positive on $\gS_0$; since its vacuum expectation value vanishes, this would contradict Proposition \ref{XYZ} unless the operator itself were zero.
The final statement follows by applying the same argument term by term.
\end{proof}


\subsection{Integration of $v^\mu :\spa \hat{T}_{\mu\nu} \spa: [f_x]$ over rest spaces $\Sigma$ and interplay with $H_v$}
Our intention is to define the four-momentum operator associated with the whole rest space $\Sigma$, obtained by integrating over $\Sigma$ the stress-energy tensor operator
smeared with $f\in \cD(\bM)$. This notion will be generalized to subregions $\Delta \subset \Sigma$ in Section \ref{Nen}. 
We start by defining a relevant density to be integrated.
Take a smearing function $f \in \cD(\bM)$ and define the unitary strongly continuous representation of the translations $\sV$ of Minkowski spacetime, $\sV \ni v \mapsto U_{(I, v)} \in \gB(\gF_s(\cH_m))$, the latter being the unitary representation (\ref{rep}) of $IO(1,3)_+$ on the Fock space and, in everything that follows, we have identified the points of $\bM$ with the vectors of $\sV$ as usual $\sV \ni x \mapsto o+x \in \bM$ through the choice of an origin $o\in \bM$. What follows does not depend on this choice. 
Following some ideas in \cite{BF}, we define the operator-valued function
\beq \sV \ni x \mapsto U_{(I, x)}:\spa \hat{T}_{\mu\nu} \spa: [f] U^{-1}_{(I, x)} =:\:  :\spa \hat{T}_{\mu\nu} \spa: [f_x] \quad \mbox{with} \quad f_x(y) := f(y-x)\quad \mbox{for $x,y\in \sV$,}\label{Tfx}\eeq
according to (d) of Proposition \ref{PROP1}.
This $x$-dependent operator $ :\spa \hat{T}_{\mu\nu} \spa: [f_x]$ enjoys a number of crucial properties.\\

\begin{proposition}\label{PROPTTR} Let $f\in \cD(\bM)$ and let $ :\spa \hat{T}_{\mu\nu} \spa: [f_x]$ be defined as in (\ref{Tfx}). The following facts hold
if $\Psi,\Psi'\in \gS_0$.
\begin{itemize}
\item[(a)] The identity holds
\beq
\langle \Psi| :\spa \hat{T}_{\mu\nu} \spa: [f_x]\Psi'\rangle  =
 \int_\bM f(y-x) \langle \Psi|  :\spa \hat{T}_{\mu\nu} \spa: (y) \Psi' \rangle d^4y\:.\nonumber
\eeq
\item[(b)] $\sV\ni x \mapsto \langle\Psi | :\spa \hat{T}_{\mu\nu} \spa: [f_x] \Psi' \rangle$ is smooth, bounded, and satisfies the following bounds
referred to a given Minkowski frame where $x=(x^0,\vec{x})$.

(1) If $\alpha$ is a multi-index for the components of $x$, for every $n\in \bN$ there is a polynomial $P_{|\alpha|+n}$ in the variable $x^0$ such that
\beq\label{BBBO}
\left|\partial^\alpha_x\langle \Psi|  :\spa \hat{T}_{\mu\nu} \spa: [f_x]| \Psi' \rangle \right| \leq \frac{|P_{|\alpha|+n}(x^0)|}{(1+ |\vec{x}|)^n}\:, \quad x \in \bR^4\:.
\eeq
In particular, $\bR^3\ni\vec{x} \mapsto \langle \Psi|  :\spa \hat{T}_{\mu\nu} \spa: [f_{(x^0, \vec{x})}]| \Psi' \rangle$ is a Schwartz function for every $x^0\in \bR$.

(2) There are finite constants $\epsilon >0$, $C>0$ such that 
$$|\langle \Psi|  :\spa \hat{T}_{\mu\nu} \spa: [f_x]| \Psi' \rangle | \leq \frac{C}{(1+ |\vec{x}|)^{3+\epsilon}}\quad \mbox{if $|\vec{x}|> |x^0|$.}$$

\item[(c)] The conservation equation holds (the derivatives being referred to coordinates of $x$)
\beq 
\partial_{\mu} \langle \Psi|  :\spa \hat{T}^{\mu\nu} \spa: [f_x] \Psi' \rangle =0\:.
\eeq
\item[(d)] If $v\in \sV$, $f\in \cD(\bM)$, and $\Sigma$ and $\Sigma'$ are rest spaces of two Minkowskian reference frames $u_\Sigma, u_{\Sigma'} \in \sT_+$, then
\beq \label{INTS}
\int_\Sigma v^\mu  \langle \Psi|  :\spa \hat{T}_{\mu\nu} \spa: [f_x] \Psi' \rangle u^\nu_\Sigma d\Sigma(x)
= \int_{\Sigma'} v^\mu  \langle \Psi|  :\spa \hat{T}_{\mu\nu} \spa: [f_x] \Psi' \rangle u^\nu_{\Sigma'} d\Sigma'(x)\:.
\eeq
Above, $d\Sigma,d\Sigma'$ are the $\bR^3$ Lebesgue measures in spatial Minkowski coordinates adapted to the $3$-planes $\Sigma$ and $\Sigma'$, respectively.
\item[(e)] The bound (\ref{bff}) holds uniformly in $x\in \bM$:
\beq\label{bff22}
u_\mu u'_\nu \langle \Psi | :\spa \hat{T}^{\mu\nu} \spa: [f_x^2] \Psi \rangle \geq -  u_\mu u'_\nu  b_f^{\mu\nu} ||\Psi||^2\:,
\eeq
for $u, u' \in \overline{\sV_+}$, and $\Psi \in \gS_0$, where $b_f^{\mu\nu}$ is defined in Proposition \ref{TEOB}.
\end{itemize}
\end{proposition}

\begin{proof} See Appendix \ref{AAA}.
\end{proof}

We move on to the interplay between the generator $H^v$ of spacetime translations along $v \in \sV$, defined in Lemma \ref{lemmaH}, and integrals of $ :\spa \hat{T}_{\mu\nu} \spa: [f_x] $ over rest spaces of Minkowskian reference frames.\\

\begin{proposition}\label{TEO54}
 If $v\in \sV$, $f\in \cD(\bM)$,
 $\Psi,\Psi' \in \gS_0$, independently of the choice of the rest space $\Sigma$ of a Minkowskian reference frame $u_\Sigma$,
\beq
 \int_{\Sigma} \langle \Psi|  :\spa \hat{T}_{\mu\nu} \spa: [f_x] \Psi' \rangle v^\mu u^\nu_\Sigma \:\: d\Sigma(x)   = \left(\int_\bM fd^4x\right) \langle \Psi|H^v\Psi' \rangle  \label{INTTT2} \:,
\eeq
so that, in particular, the left-hand side is positive if $\Psi=\Psi'$, $f\geq 0$ and $v\in \sV_+$.
%
%
\end{proposition}

\begin{proof} 
 Fix a Minkowskian coordinate system with $\partial_0= u_\Sigma$ so that $\Sigma$ coincides with the plane $x^0=0$. We integrate in $\vec{x}$ over the whole space $\bR^3$ the four terms in (\ref{FINT}). The standard regularization of the spatial Fourier integrals gives the two identities
\[
\int_{\bR^3}e^{i(\vec{k}-\vec{p})\cdot\vec{x}}d^3x=(2\pi)^3\delta^{(3)}(\vec{k}-\vec{p}),\qquad
\int_{\bR^3}e^{i(\vec{p}+\vec{k})\cdot\vec{x}}d^3x=(2\pi)^3\delta^{(3)}(\vec{p}+\vec{k}) .
\]
If $\tilde{p}:=(E(p),-\vec{p})$, this yields
\begin{align*}
& \int_{\Sigma} v^\mu \langle \Psi |  :\spa \hat{T}_{\mu\nu} \spa: [f_x]   \Psi'\rangle u^\nu_\Sigma d\Sigma(x) \\
&=2\pi^2 \int_{\bR^3}\hat{f}(-p-\tilde{p})
\frac{ \langle \Psi | a_{p}a_{\tilde{p}}\Psi'\rangle}{E(p)^2} v^\mu t_{\mu0}(p, -\tilde{p})d^3p \\
&\quad +2\pi^2\int_{\bR^3} \hat{f}(p+\tilde{p})\frac{  \langle \Psi | a^\dagger_{p}a^\dagger_{\tilde{p}}\Psi'\rangle}{E(p)^2} v^\mu t_{\mu0}(p,-\tilde{p})d^3p \\
&\quad +2\pi^2 \int_{\bR^3} \hat{f}(0)\frac{ \langle \Psi | a^\dagger_{p}a_{p}\Psi' \rangle}{E(p)^2}v^\mu t_{\mu0}(p,p)d^3p \\
&\quad +2\pi^2 \int_{\bR^3}\hat{f}(0)\frac{ \langle \Psi |a^\dagger_{p}a_{p}\Psi'\rangle}{E(p)^2}v^\mu t_{\mu0}(p,p)d^3p\: .
\end{align*}
At this point, direct inspection proves that $t_{\mu0}(p,-\tilde{p})=0$, whereas $t_{\mu\nu}(p,p)=p_\mu p_\nu$ and hence
\[
\frac{v^\mu t_{\mu0}(p,p)}{E(p)}= -v\cdot p .
\]
Therefore the first two integrals vanish and the two remaining ones give
$$ \int_{\Sigma} v^\mu \langle \Psi |  :\spa \hat{T}_{\mu\nu} \spa: [f_x]   \Psi'\rangle u^\nu_\Sigma d\Sigma(x)=
-(2\pi)^2\hat{f}(0) \int_{\sV_{m,+}}\langle \Psi |a^\dagger_{p}a_{p}\Psi'\rangle v\cdot p\: d\mu_m(p)\:.$$
Since $(2\pi)^2\hat{f}(0) = \int f(x) d^4x$ and taking (\ref{QuadHv}) into account, the right-hand side of the identity found is $(\int fd^4x) \langle \Psi| H^v \Psi'\rangle$.
\end{proof}

\noindent A surprising fact is that
-- assuming $\int_\bM f d^4x =1$ --
the right-hand side of (\ref{INTTT2}) is independent of the smearing function $f\in \cD_\bR(\bM)$, which however appears in the left-hand side. 
We stress that this is not an evident result, since it proves in particular that, in the left-hand side, two integrals can be interchanged, but one is over the whole spacetime, namely $\bR^4$, and the other over a flat Cauchy surface, i.e. $\bR^3$: it is not a direct application of Fubini's theorem.
It is easy to see, at a heuristic level, that this is due to elementary properties of the smearing procedure and the conservation property (\ref{INTS}) when $\Sigma$ and $\Sigma'$ are parallel.
We expect that (\ref{INTTT2}) (or its formulation in terms of quadratic forms) is also valid in curved spacetime in the presence of a Killing vector field $v$. According to (\ref {INTTT2}) it could be convenient to introduce the notation, formally motivated by the replacement $f(y-x) \to \delta(y-x)$,
$$H^v = \int_{\Sigma}  :\spa \hat{T}_{\mu\nu} \spa: (x)  v^\mu u^\nu_\Sigma \:\: d\Sigma(x)$$
which is used in theoretical physics textbooks.

\section{Relativistic spatial localization observable from the stress-energy operator} \label{SEC3agg}
We are now in a position to construct a relativistic spatial localization observable for $n$-particle states by integrating a normalized notion of $:\spa T_{\mu\nu}\spa:[f^2_x]$ over Borel sets $\Delta$ of rest spaces $\Sigma$ of Minkowskian reference frames $u_\Sigma$. 
The construction we are going to present proves that the relativistic spatial localization observables introduced in \cite{Terno,M23}
are actually rigorously constructed out of QFT notions: they are restrictions of a more general structure, as conjectured in the conclusions of \cite{CDRM}.

\subsection{$n$-particle relativistic spatial localization  from $:\spa \hat{T}_{\mu\nu}\spa:[f^2]$ satisfying CC}

Due to (\ref{INTTT2}), if $f\in \cD_\bR(\bM)$, $\int_\bM f d^4x =1$, $f\geq 0$, and $u\in \sT_+$, extending a definition given in \cite{M23}, an expected expression for the desired effects should be 
\beq \sA_f^u(\Delta) =    \frac{1}{\sqrt{H^u}} \int_{\Delta}  :\spa \hat{T}_{\mu\nu} \spa: [f_x]  u^\mu u^\nu_\Sigma \:\: d\Sigma(x)  \frac{1}{\sqrt{H^u}}\:, \quad \Delta \in \cB(\Sigma)\quad \mbox{(tentatively!)}\label{propos}\eeq
where $\cB(\Sigma)$ denotes the $\sigma$-algebra of Borel sets in the rest space $\Sigma$ of the Minkowskian reference frame $u_\Sigma$. This type of expression was used in \cite{M23}, and already in \cite{Terno} for a first version of this localization notion with $u=u_\Sigma$. However, the notion of stress-energy tensor operator was introduced in those works only in a heuristic manner, without analyzing the crucial smearing procedure, which makes it a mathematically sound object in rigorous QFT and also allows it to comply with the localization notion of the AHK approach. As a matter of fact, the rigorous version of the constructed POVM was actually defined in \cite{M23} in terms of a suitable improvement of Newton-Wigner's PVM and only in the one-particle space. This type of relation with the Newton-Wigner observable will arise later in a more general form, but we shall not use it to {\em define} our notion of localization in our QFT context, contrary to \cite{M23}.

Evident physical issues with (\ref{propos}) arise immediately in our QFT context. First of all, $\sA_f^u(\Delta)$ defined as above is not necessarily positive, 
However, some lower bounds hold, as shown in Proposition \ref{TEOB} when using $f^2$ with $f\in \cD_\bR(\bM)$. Furthermore, due to Proposition  \ref{TEOB1}, if we consider the compressions to $n$-particle spaces $\cH^{(n)}$, namely the operators given by
$ P_n:\spa \hat{T}_{\mu\nu} \spa: [f^2_x] P_n$ and $P_nA_f^u(\Delta) P_n$, they are positive, and the latter is also normalized to $I$ for $\Delta=\Sigma$ because of (\ref{INTTT2})
($P_n \gF_s(\cH_m) \to \gF_s(\cH_m)$ denotes the orthogonal projector onto $\cH^{(n)}$).
We are particularly interested in the specific case of {\em one particle}: $n=1$.
From the causality properties proved in Proposition \ref{TEOB1}
and (c) in Proposition \ref{PROPTTR}, exploiting a procedure already used in \cite{M23,DM24,CDRM},
we expect that the causality condition (CC) in Definition \ref{REMM0444} is satisfied when (a compression of) the operators $\sA_f^u(\Delta)$ is used to define families of POVMs on all rest spaces of $\bM$.
Finally, notice that in (\ref{propos})
$\frac{1}{\sqrt{H^u}}\Omega$ is not defined. To fix this problem we consider the square root of the resolvent $H^{-1}_{u,\epsilon}$, where henceforth we define the {\em $\epsilon$}-regularized Hamiltonian
\beq H^u_{\epsilon}:=(H^u+ \epsilon I)\quad \mbox{for $\epsilon>0$ (arbitrarily small)}\:,\eeq  
and later we shall consider the limit as $\epsilon\to 0^+$. We notice {\em en passant} that the $\alpha$-th power of the written operator, defined via functional calculus, satisfies
\beq (H^u_{\epsilon})^\alpha (\gS_0) = \gS_0 \quad \mbox{for $\alpha \in \bR$ and $\epsilon >0$.} \label{FONDH}\eeq
In fact, in every reducing space $\cH^{(n)}$ the operator $H^u_{\epsilon}$ restricts to a multiplicative operator with a strictly positive, smooth polynomially bounded function. The spectrally defined $\alpha$-power of $H^u_{\epsilon}$ therefore coincides with the corresponding power of this multiplicative operator (e.g. by (f) of Proposition 3.3 in \cite{Moretti2}) in each $\cH^{(n)}$, giving rise to (\ref{FONDH}).
We start from a technical preliminary result, where we explicitly use $f^2$ to smear the stress-energy tensor. Some of the results established below should remain valid even when using $f\geq 0$ in place of $f^2$. However, the proof of the proposition below is easier to carry out if one assumes that $\sqrt{f}$ is smooth, and this is not guaranteed by the requirement that $f\geq 0$ be smooth, as observed below the proof of Proposition \ref{TEOB}.\\  

%
%

\begin{proposition}\label{LLAST} Consider $u\in \sT_+$, $f \in \cD_\bR(\bM)$, $\epsilon>0$ and, 
for a rest space $\Sigma$ of a Minkowskian reference frame $u_\Sigma \in \sT_+$,
define the algebra of sets
$$\cB_0(\Sigma) := \{\Delta \in \cB(\Sigma) \:|\: \: \mbox{either}\:\: |\Delta|< +\infty \:\: \mbox{or}\:\: |\Sigma \setminus \Delta| <+\infty\}.$$
For every $\Delta \in \cB_0(\Sigma)$ there is a bounded operator $ \sA_{f,\epsilon}^u(\Delta)
: \gF_s(\cH_m) \to \gF_s(\cH_m)$
which is uniquely defined by requiring that
\beq \label{sA} \left\langle\Psi \left| \sA_{f,\epsilon}^u(\Delta)  \right.\Psi' \right\rangle = 
 \int_{\Delta}\sp  \left\langle  \frac{1}{\sqrt{H^u_{\epsilon}}}  \Psi \left|   :\spa \hat{T}_{\mu\nu} \spa: [f^2_x] u^\mu u^\nu_\Sigma  \frac{1}{\sqrt{H^u_{\epsilon}}} \Psi' \right. \right \rangle  d\Sigma(x) \:\: \mbox{if $\Psi,\Psi' \in \gS_0$.} \label{FOND1}
\eeq
(Where the right-hand side more generally exists for a generic $\Delta \in \cB(\Sigma)$.)
The following facts hold.
\begin{itemize}
    \item[(a)] $\sA_{f,\epsilon}^u(\Delta)^\dagger  = \sA_{f,\epsilon}^u(\Delta)$.

   \item[(a1)] There is  $C_{f} >0$ such that
 $||P_n \sA_{f,\epsilon}^u(\Delta) P_n|| \leq C_{f}$ for all $n\in \bN$, $\epsilon >0$, $\Delta\in \cB_0(\Sigma)$, where $P_n$ is the orthogonal projector onto $\cH^{(n)}$.
      \item[(b)] If $\Delta = \Sigma$,
\beq \sA^u_{f,\epsilon}(\Sigma)
= \left(\int_\bM f^2 d^4x\right)
\frac{H^u}{{H^{u} + \epsilon I}}\label{NORMa}\:.
\eeq
\item[(c)] If $g\in IO(1,3)_+$ then the covariance relation holds
\beq 
U_g \sA_{f,\epsilon}^u(\Delta) U_g^{-1} =
\sA_{g_*f,\epsilon}^{gu}(g\Delta) 
\eeq
where $U$ is the unitary $IO(1,3)_+$-representation (\ref{rep}).
\item[(d)] The map $\cB_0(\Sigma)\ni \Delta \mapsto \sA_{f,\epsilon}^u(\Delta) $ is weakly $\sigma$-additive.

\item[(e)] The operators $ \sA_{f,\epsilon}^u(\Delta)$ are bounded from below (but not necessarily positive).
\end{itemize}
All the statements remain valid if one replaces $f^2$ by $f= \sum_{i=1}^N f_i^2$ with $f\in \cD_\bR(\bM)$.
\end{proposition}
\begin{proof}
See Appendix \ref{AAA}.
\end{proof}

We are now in a position to state and prove one of the main results of this work.
We construct a positive-energy  relativistic spatial localization observable (Definition \ref{REMM0}) for states with a definite number of particles in the Fock space and which is causal in the sense that it complies with CC as in Definition \ref{REMM0444}. This family of POVMs uniquely arises from the stress-tensor operator smeared with a test function $f^2$ and the choice of a preferred temporal direction $u$. We cannot directly use an identity such as (\ref{propos}) to define an effect because of the states with ``negative probability'' that would arise from (e) in Proposition \ref{LLAST}. Even if we shall return later to that issue, what we do now is remove these annoying states by passing to the compressions $P_nA_f^u(\Delta) P_n$, as already suggested.
 
As before, we use the following notation
\beq {\cal R} := \cup \{\cB(\Sigma)\:|\: \Sigma \subset \bM \: \mbox{spacelike  $3$-plane}\}\eeq where $\cB(\Sigma)$ is the Borel $\sigma$-algebra on $\Sigma$ and $u_\Sigma\in \sT_+$ is the future-oriented unit normal vector to $\Sigma$.\\

\begin{theorem}\label{TEOR}
 Consider a real scalar Klein-Gordon field on $\bM$ with mass $m>0$ and the associated 
normally ordered stress-energy tensor operator $:\spa \hat{T}_{\mu\nu}\spa:[f]$ defined in the Fock space $\gF_s(\cH_m)$ as in Proposition \ref{PROP1}.  
Take $u\in \sT_+$ with associated selfadjoint Hamiltonian $H^u$ (\ref{Hv}), and $f \in \cD_\bR (\bM)$ such that $\int_{\bM} f^2 d^4x =1$.
Then there is a unique map
  \beq\label{A92}
  {\cal R} \ni \Delta \mapsto {A}^u_{f}(\Delta)  \in \gB(\gF_s(\cH_m))
  \eeq
such that
 \begin{itemize}

\item[(a)] $A^u_{f}(\Delta)(\cH^{(n)})\subset \cH^{(n)} $ for $\Delta \in {\cal R}$ and $n\in \bN$;
      
      \item[(b)] ${A}^u_{f}(\Delta)(\cH^{(0)})=0$ for every $\Delta \in {\cal R}$;
  
%
\item[(c)] if $\rho \in \sS(\cH^{(n)})$ for a given $n=0,1,2,\ldots$ and $\Delta \in \cB_0(\Sigma)$, for a spacelike $3$-plane $\Sigma$,
 \beq \label{LIM}
 tr(\rho {A}_{f}^u(\Delta)) = \lim_{\epsilon \to 0^+} 
tr(\rho \sA^u_{f,\epsilon}(\Delta))\:;
 \eeq
  
\item[(d)] for every spacelike $3$-plane $\Sigma$, $\cB(\Sigma) \ni \Delta \mapsto A^u_{f}(\Delta)\spa\rest_{\cH^{(0)\perp}}$ is a (normalized) POVM on $\cH^{(0)\perp}$.
      \end{itemize}
The following further facts are true.
   \begin{itemize}

\item[(e)] 
If $\Delta \in {\cal R}$, $P_n: \gF_s(\cH_m) \to \gF_s(\cH_m)$ is the orthogonal projector onto $\cH^{(n)}$, and $\Psi, \Psi' \in \gS_0$, then
\beq  \langle \Psi| {A}^u_f(\Delta)\Psi'\rangle =\sum_{n=1}^{+\infty}\int_{\Delta} \left\langle   \frac{1}{\sqrt{H^u}}    P_n \Psi \left| :\spa \hat{T}_{\mu\nu} \spa: [f^2_x] \frac{1}{\sqrt{H^u}} \right.      P_n\Psi' \right\rangle u^\mu u^\nu_\Sigma  d\Sigma(x)\:,\label{igt2}\eeq
where $\frac{1}{\sqrt{H^u}}$
is defined by spectral calculus on $\cH^{(0)\perp}$.
  
%
%
%
\end{itemize}

\noindent Furthermore, referring to Definition \ref{REMM0}, if $U$ is the unitary $IO(1,3)_+$-representation (\ref{rep}) and $A_f$ denotes the family of maps $\{A_f^u\}_{u\in \sT_+}$,
\begin{itemize}
\item[(f)] $(\cH^{(0)\perp}, {\cal R}, A_f\sp\rest_{\cH^{(0)\perp}}, U\sp\rest_{\cH^{(0)\perp}})$ is a positive-energy  relativistic spatial localization observable
which is causal according to  Definition \ref{REMM0444}.

\item[(g)] $(\cH^{(n)}, {\cal R}, A_{f}\spa\rest_{\cH^{(n)}}, U\spa\rest_{\cH^{(n)}})$ with $n>0$ 
 is a positive-energy  relativistic spatial localization observable
which is causal according to Definition \ref{REMM0444}.
      \end{itemize}
All the statements remain valid if one replaces $f^2$ by $f= \sum_{i=1}^N f_i^2$ with $f\in \cD_\bR(\bM)$.
\end{theorem}

    \begin{proof} Referring to the proof of Proposition \ref{LLAST} (decomposition (\ref{expand2}) in particular), define 
\beq {A}^u_f(\Delta)  := \sA_{0,3}(\Delta) (=: \sA_{\epsilon,3}(\Delta)|_{\epsilon=0}) \label{DEFAGG}\eeq 
With this definition, (a) and (b) are true, in particular because $P_n \sA_{0,3}(\Delta) = \sA_{0,3}(\Delta) P_n$ and $P_0 \sA_{0,3}(\Delta) = \sA_{0,3}(\Delta) P_0 =0$.
We stress that ${A}^u_f(\Delta)\in \gB(\gF_s(\cH_m))$ and that this operator is defined for every $\Delta \in \cB(\Sigma)$ of every spacelike $3$-plane, as observed in the proof of Proposition \ref{LLAST} (see the discussion below (\ref{expand}) and (\ref{TONE})). 
For future convenience we note that, if $\Psi \in \cH^{(n)}\cap \gS_0$ with $n>0$, 
we have
\beq \langle \Psi| {A}^u_f(\Delta)\Psi \rangle = \langle \Psi| \sA_{0,3}(\Delta) \Psi \rangle = \int_{\Delta}\sp  \left
\langle  \frac{1}{\sqrt{H^u}}  \Psi \left|    :\spa \hat{T}_{\mu\nu} \spa: [f^2_x] \frac{1}{\sqrt{H^u}}  \right. \Psi \right\rangle u^\mu u^\nu_\Sigma  d\Sigma(x)\geq 0 \label{igt}\eeq
where $\frac{1}{\sqrt{H^u}}$
is defined by spectral calculus on $\cH^{(0)\perp}$ and as the zero operator on $\cH^{(0)}$, and where we have taken advantage of Proposition \ref{TEOB1}.
At this point, using $P_n \sA_{0,3}(\Delta) = \sA_{0,3}(\Delta) P_n =P_n \sA_{0,3}(\Delta)P_n$, linearity and polarization immediately yield (e).
\\  Let us pass to the proof of (c). Since both sides vanish separately if $n=0$, we consider the case $n>0$. First we focus on 
the more elementary case of $\rho' = \sum_{i=1}^N p_i |\Psi_i \rangle \langle \Psi_i|$, for $\Psi_i \in \cH^{(n)}$ with a given $n$.
For this type of state it is sufficient to prove that, if $\Delta \in \cB(\Sigma)$ (actually it would be sufficient to consider only the case $\Delta \in \cB_0(\Sigma)$) and $\Psi,\Psi'\in \cH^{(n)}$, with $n>0$,
 $$ \langle \Psi| \sA_{\epsilon,3}(\Delta)\Psi' \rangle = \int_{\Delta}\sp\left\langle  \frac{1}{\sqrt{H^u +  \epsilon I}}  \Psi \left|    :\spa \hat{T}_{\mu\nu} \spa: [f^2_x] \frac{1}{\sqrt{H^u + \epsilon I}}   \right. \Psi' \right\rangle u^\mu u^\nu_\Sigma d\Sigma(x)$$
 \beq \to \int_{\Delta}\sp\left\langle  \frac{1}{\sqrt{H^u}}  \Psi \left|    :\spa \hat{T}_{\mu\nu} \spa: [f^2_x] \frac{1}{\sqrt{H^u}}   \right. \Psi' \right\rangle u^\mu u^\nu_\Sigma d\Sigma(x) = \langle \Psi| \sA_{0,3}(\Delta)\Psi' \rangle\quad\mbox{ if $\epsilon\to 0^+$. }
\label{WANT}\eeq
Let us prove this. Indeed, choosing a Minkowskian coordinate system with $\Sigma$ defined by $x^0=0$ and looking at the expression of $\langle (H^u+  \epsilon I)^{-1/2}\Psi|  :\spa \hat{T}_{\mu\nu} \spa: [f^2_x]| (H^u+  \epsilon I)^{-1/2} \Psi'\rangle$, arising from
the expansion (\ref{FINT}), where only the last two integrals contribute for $\Psi', \Psi\in \gS_0\cap \cH^{(n)}$ with $n>0$, one easily sees that, for every $\vec{x}\in \Sigma$,
$$\langle (H^u+  \epsilon I)^{-1/2}\Psi|  :\spa \hat{T}_{\mu\nu} \spa: [f^2_x]| (H^u+  \epsilon I)^{-1/2} \Psi'\rangle \to 
\langle (H^u)^{-1/2}\Psi|  :\spa \hat{T}_{\mu\nu} \spa: [f^2_x]| (H^u)^{-1/2} \Psi'\rangle$$
as $\epsilon\to 0$. This is a direct consequence of the dominated convergence theorem, when one expands the scalar products as integrals on $\sV_{m,+}^{n}$ using the fact that 
the function $\Psi$ is of Schwartz type (when viewed as a function of the spatial momenta $(\vec{p}_1,\ldots,\vec{p}_n)$), and that $(E^n_{\epsilon, u})^{-1/2}$ is a bounded function of $(\vec{p}_1,\ldots,\vec{p}_n)$. On the other hand, integration by parts proves that,
for every $N=0,1,\ldots$, there are polynomials $Q_N(z_1,\ldots, z_n)$ of degree $N$ in the variables $z_k\in \bR$, such that
$$|\langle (H^u+  \epsilon I)^{-1/2}\Psi|  :\spa \hat{T}_{\mu\nu} \spa: [f^2_x]| (H^u+  \epsilon I)^{-1/2} \Psi'\rangle| $$ $$\leq \frac{
\int_{\bR^{3n}} |E^n_{\epsilon,u}(P)|^{-1/2} |\Phi(P)| |Q_N(\nabla_{\vec{p}_1})  \Phi'(P) E^n_{\epsilon,u}(P)^{-1/2}| d^{3n}p}{(1+ |\vec{x}|)^N}$$
where $P= (\vec{p}_1,\ldots, \vec{p}_n)$ and the Schwartz functions $\Phi,\Phi'\in \cS(\bR^{3n})$ are obtained from $\Psi$ according to (\ref{UM}).
At this point, it is not difficult to see that, due to the special form of the maps $E^n_{\epsilon,u}(p_1,\ldots,p_n)$ and the fact that the function $\Phi$ is Schwartz, there are constants $C_N$ such that 
$$
\int_{\bR^{3n}} E^n_{\epsilon,u}(P)^{-1/2} |\Phi(P)| |Q_N(\partial_{p_1})  \Phi'(P) E^n_{\epsilon,u}(P)^{-1/2}| d^{3n}p \leq C_N < +\infty \quad \mbox{uniformly in $\epsilon \in [0, a)$.} $$
We are now in a position to apply once more the dominated convergence theorem to the left-hand side of (\ref{WANT}) with respect to the $d\Sigma$ integration, proving that (\ref{WANT}) holds. This completes the proof of (c) for $n>0$ in the elementary case considered. 
Let us conclude the proof of (c). The space of density matrices of type $\rho' = \sum_{i=1}^N p_i|\Psi_i \rangle \langle \Psi_i|$, for $\Psi_i \in \cH^{(n)}$ with a given common $n$, is dense in $\sS(\cH^{(n)})$ in the norm $||\cdot||_1$ since $\gS_0 \cap \cH^{(n)}$
is dense in $\cH^{(n)}$ (we leave to the reader the elementary proof based on the fact that the spectral decomposition of a density matrix is a series of operators 
$p_j |\Psi_j\rangle \langle \Psi_j|$, $p_j\geq 0$, which converges in the norm $||\cdot||_1$ and that these operators can be approximated by the previously considered ones in the same topology). As a consequence, if $\rho \in \sS(\cH^{(n)})$, and $\rho'$ is as above, $|tr(\rho A^u_{f}(\Delta)) - tr(\rho \sA^u_{f,\epsilon}(\Delta))|$
$$= |tr(\rho A^u_{f}(\Delta)) - tr(\rho P_n\sA^u_{f,\epsilon}(\Delta)P_n)| = | tr(\rho A^u_{f}(\Delta)) - tr(\rho \sA_{\epsilon, 3}(\Delta))| $$
$$\leq| tr((\rho -\rho') \sA_{\epsilon, 3}(\Delta))| + |tr(\rho' A^u_{f}(\Delta)) - tr(\rho'\sA_{\epsilon, 3}(\Delta))| +| tr((\rho -\rho') A^u_{f}(\Delta)  )|
$$
$$\leq ||\rho-\rho'||_1 \: ||\sA_{\epsilon, 3}(\Delta)||+  |tr(\rho' A^u_{f}(\Delta)) - tr(\rho'\sA_{\epsilon, 3}(\Delta))| +  ||\rho-\rho'||_1\: ||\sA^u_{f}(\Delta)||\:.$$
For every $\eta>0$ we can take $\rho'$ as above such that $ ||\rho-\rho'||_1 ||\: ||\sA^u_{f}(\Delta)||< \eta/3$. On the other hand, since $\sA_{\epsilon, 3}(\Delta)=
P_n\sA^u_{f,\epsilon}(\Delta)P_n$ and (a1) of Proposition \ref{LLAST} holds, we can redefine $\rho'$ so that also $ ||\rho-\rho'||_1 \: ||\sA_{\epsilon, 3}(\Delta)||< \eta/3$ holds. The previous part of the proof for elementary states of type $\rho'$ proves that, for that special $\rho'$,  $|tr(\rho' A^u_{f}(\Delta)) - tr(\rho'\sA_{\epsilon, 3}(\Delta)| <\eta/3$ if $\eta>0$ is sufficiently small.
We have proved that $|tr(\rho A^u_{f}(\Delta)) - tr(\rho \sA^u_{f,\epsilon}(\Delta))|\to 0$ as $\epsilon\to 0^+$, concluding the proof of (c).
It remains to prove (d).
It is clear that the operators $A^u_f(\Delta)$ are bounded and positive, and that $A^u_f(\Sigma)=I$ from (\ref{igt}) and (\ref{INTTT2}).
We only have to show that $\cB(\Sigma)\ni \Delta \mapsto A^u_f(\Delta)\rest_{\cH^{(0)\perp}}$ is weakly $\sigma$-additive.  From (\ref{igt}) and the fact that 
$ \sA_{0,3}(\Delta)P_n = P_n  \sA_{0,3}(\Delta)$
always taking $\Psi \in \cH^{(0)\perp}\cap \gS_0$, we get
$$\langle \Psi | A^u_f(\Delta) \Psi \rangle =\sum_{n\in \bN}
 \int_{\Delta}\sp  \left\langle   \frac{1}{\sqrt{H^u}} \Psi_n \left|    :\spa \hat{T}_{\mu\nu} \spa: [f^2_x] u^\mu u^\nu_\Sigma\right. \frac{1}{\sqrt{H^u}} \Psi_n \right\rangle   d\Sigma(x)\:.$$
Here, for that given $\Psi$, the right-hand side is evidently $\sigma$-additive (notice that only a finite number of components $\Psi_n$ occur and we can use the dominated convergence theorem for (\ref{BBBO})).  Polarization proves that weak $\sigma$-additivity is valid on $ \cH^{(0)\perp}\cap \gS_0$. To conclude, we extend the proof to a generic $\Psi \in \cH^{(0)\perp}$.
Let $\Delta_1, \Delta_2,\ldots$ be a sequence of mutually disjoint sets of $\cB(\Sigma)$ and define 
$S_N := \sum_{j=1}^N A^u_f(\Delta)$. The operators $S_N :\cH^{(0)\perp} \to \cH^{(0)\perp}$ form an increasing sequence of positive operators, so there exists 
a positive operator $S := \sum_{j=1}^{+\infty} A^u_f(\Delta_j) \in \gB(\cH^{(0)\perp})$, where the sum is understood in the strong sense (see, e.g., \cite{Moretti1}). Therefore, in particular, for $\Psi \in \cH^{(0)\perp}$,
$\langle \Psi| S \Psi\rangle  = \sum_{j=1}^{+\infty} \langle \Psi|A^u_f(\Delta_j)\Psi \rangle$.
In the special case $\Psi \in \cH^{(0)\perp}\cap \gS_0$, $\sigma$-additivity requires
$\langle \Psi| S \Psi\rangle  = \sum_{j=1}^{+\infty} \langle \Psi|A^u_f(\Delta_j)\Psi \rangle = \langle \Psi|A^u_f(\cup_{j=1}^{+\infty}\Delta_j)\Psi \rangle$.
Since $\cH^{(0)\perp}\cap \gS_0$ is dense in $\cH^{(0)\perp}$ and $S, A^u_f(\cup_{j=1}^{+\infty}\Delta_j) \in \gB(\cH^{(0)\perp})$, the identity holds for every $\Psi\in \cH^{(0)\perp}$. Eventually, the 
standard argument based on polarization implies $S= A^u_f(\cup_{j=1}^{+\infty}\Delta_j)$, so that 
$\sum_{j=1}^{+\infty} \langle \Psi|A^u_f(\Delta_j)\Psi'\rangle = \langle \Psi|A^u_f(\cup_{j=1}^{+\infty}\Delta_j)\Psi'\rangle$ if $\Psi, \Psi'\in \cH^{(0)\perp}$, concluding the proof of
weak $\sigma$-additivity.

We now turn to uniqueness. If $A(\Delta)\in \gB(\gF_s(\cH_m))$ is another family of operators satisfying (a)-(d), from (c) we conclude that $A(\Delta)= A^u_f(\Delta)$ for every $\Delta \in \cB_0(\Sigma)$.  On the other hand, (d) requires that $\cB(\Sigma) \ni \Delta \mapsto \langle \Psi |A(\Delta) \Psi \rangle \in [0,1]$ be a positive Borel measure for every $\Psi \in \cH^{(0)\perp}$. Since $\cB_0(\Sigma)$ generates $\cB(\Sigma)$, the uniqueness part of Carath\'eodory's extension theorem implies that the said finite measure 
coincides with $\cB(\Sigma) \ni \Delta \mapsto \langle \Psi |A^u_f(\Delta) \Psi \rangle \in [0,1]$. Therefore $A(\Delta)\sp\rest_{\cH^{(0)\perp}} =A^u_f(\Delta)\sp\rest_{\cH^{(0)\perp}}$ for every $\Delta \in \cB(\Sigma)$. Since $\cH^{(0)}$ is invariant for both operators and on it the operators coincide ((a) and (b)), we conclude that $A(\Delta) = A^u_f(\Delta)$ if $\Delta \in \cB(\Sigma)$.

We conclude the proof by establishing (f) and (g). The positive-energy requirment is trivially valid due to the positivity of the Hamiltonian $H^v$ for $v\in \sT_+$.  Covariance of the constructed families of POVMs easily follows from (c) of Proposition \ref{LLAST} and the fact that 
$U$ leaves the spaces $\cH^{(0)}$ invariant. The only non-trivial property is CC. 
Since $\cH^{(0)\perp}\cap \gS_0$ is dense in $\cH^{(0)\perp}$ and $A_f^u(\Delta)$ is continuous, it is sufficient to prove that the family of non-negative numbers 
$\langle \Psi | A_f^u(\Delta)\Psi \rangle$ satisfies CC for $\Psi \in \cH^{(0)\perp}\cap \gS_0$ whose $N$ non-vanishing components are denoted by $\Psi_n$. 
We know that it holds $\langle \Psi| A_f^u \Psi \rangle = \sum_{n=1}^N \int_{\Delta} \langle \frac{1}{\sqrt{H^u}} \Psi_n | u^\mu :\spa \hat{T}_{\mu\nu}\sp:[f^2_x] \frac{1}{\sqrt{H^u}}\Psi_n \rangle  u^\nu_\Sigma d\Sigma(x)$.
The proof relies entirely on the properties of the smooth current which appears in the integrand above
$\bM \ni x \mapsto J_\Psi^\nu(x) := \sum_{n=1}^N \langle \frac{1}{\sqrt{H^u}} \Psi_n | u^\mu :\spa {\hat{T}_\mu}\:^\nu\sp:[f^2_x] \frac{1}{\sqrt{H^u}}\Psi_n \rangle$.  It is causal and future-directed wherever it does not vanish,
due to Proposition \ref{TEOB1},
and conserved in view of (c) of Proposition \ref{PROPTTR}.  The proof of the validity of CC is then the same as for Theorems 35 and 39 in \cite{M23}.
\end{proof}
\subsection{Center  of energy, Newton-Wigner position observable, Heisenberg inequality}\label{secNW}
We suggest a rather direct physical interpretation of the localization observables we have constructed. Results in \cite{M23,DM24} prove that, concerning single particles and localization observables constructed out of the stress-energy tensor (at a formal level in those references), a nice interplay emerges between the {\em first moment} of the POVM on $\Sigma$ and the {\em Newton-Wigner selfadjoint position operator} for the Minkowskian observer $u_\Sigma$ (see below).  Irrespective of the choice of the time direction $u$, and for the centered smearing functions specified below, the three components of the first moment, viewed as symmetric operators, {\em are} the three selfadjoint operators representing the components of the Newton-Wigner position observable on $\Sigma$ (restricted to $\sS(\cH_m)$). This fact is of physical relevance, in particular because the Newton-Wigner observable reduces to the standard notion of position in non-relativistic quantum mechanics when the energy content of the quantum state is negligible with respect to the particle mass. This type of result is quite general, since it holds \cite{DM24} for a wide class of (positive-energy) relativistic position observables studied  by Castrigiano \cite{C23} for a massive boson.

We recall to the reader that, if $\Sigma$ is a rest space of a Minkowskian reference frame $u_\Sigma\in \sT_+$ and $x^0,x^1,x^2,x^3$ are Minkowskian coordinates adapted to $u_\Sigma$ with $\Sigma$ corresponding to $x^0=0$, the three components of the {\bf Newton-Wigner position observable} on $\Sigma$ for our massive scalar boson \cite{M23} are the unique selfadjoint extensions 
of the three symmetric operators in $\cH_m = L^2(\sV_{m,+}, d\mu_m)$, 
\beq
(N_\Sigma^a\psi)(p) = i\sqrt{E_{u_\Sigma}(p)} \frac{\partial}{\partial p^a} \frac{\psi (p)}{\sqrt{E_{u_\Sigma}(p)}} \quad \mbox{for $\psi \in \cS(\sV_{m,+})$  and $a=1,2,3$,} \eeq
where, as usual, $E_u(p) := -p\cdot u$ for $u\in \sT_+$.

As $N_\Sigma^a$ is selfadjoint, it has a PVM and, since the PVMs of the three operators commute, one can define a joint PVM $\{Q(\Delta)\}_{\Delta \in \cB(\Sigma)}$. It has the structure, for $\psi \in L^2(\sV_{m,+}, d\mu_m)$,
\beq (Q_\Sigma(\Delta)\psi)(p) := \int_{\Delta} \sp d\Sigma(x) \sp \int_{\sV_{m,+}} \spa \sp\sp \sp  d\mu_m(k)  \frac{e^{-i(p-k)\cdot x}}{(2\pi)^3}\sqrt{E_{u_\Sigma}(p)}\sqrt{E_{u_\Sigma}(k)}\psi(k)\:. \eeq
A technical summary of the properties of the NW localization observable for a scalar particle and the problems it raises with causality is given in \cite{M23}. We only stress that, in spite of its appealing properties, the NW localization observable cannot be considered a physically sound notion of spatial localization due to its conflict with elementary causality requirements like CC. However, as said above, it could still coincide with the first moment of several more meaningful {\em unsharp} notions of localization  which satisfy CC (Proposition 61 in \cite{DM24}):
$$\int_{\Sigma} x^a \langle \psi | A^u(dx) \psi \rangle = \langle \psi | N^a_\Sigma \psi\rangle\quad \mbox{for  $\psi \in \cS(\sV_{m,+})$ and $a=1,2,3$.}$$
where $A^u$ is a $u$-parametrized family of POVMs $A^u|_{\cB(\Sigma)}$ for every $\Sigma \in {\cal R}$ according to Definition \ref{REMM0} of relativistic spatial localization observable.

If $\Sigma$ is the rest space, at some time, of a Minkowskian reference frame $u_\Sigma \in \sT_+$, and we use a Minkowskian coordinate system adapted to $u_\Sigma$ such that $\Sigma$ is
described by $x^0=0$, the three operators describing the components $X^a_\Sigma$ of the {\bf center of $u$-energy} on $\Sigma$ of our quantum field (the associated particles) are, for $\Psi \in \gS_0$ with $\Psi = \oplus_{n=0}^{+\infty}\Psi_n$,
\begin{align}
X_{u,\Sigma}^a \Psi  = 0 &\oplus N_{\Sigma}^a \Psi_1(p) \oplus
 \sum_{i=1}^2\frac{1}{2} \left\{ N_{\Sigma, i}^a, \frac{E_{u}(p_i)}{E^2_{u}(p_1,p_2)} \right\}  \Psi_2(p_1,p_2)\nonumber\\
&\oplus \cdots \oplus  \sum_{i=1}^n\frac{1}{2} \left\{ N_{\Sigma, i}^a, \frac{E_{u}(p_i)}{E^n_{u}(p_1, \ldots, p_n)} \right\}  \Psi_n(p_1,\ldots, p_n) \oplus \cdots \label{CE}\end{align}
where $\{A,B\}:= AB+BA$ is the standard {\em anticommutator} and, as before,
$$E_u^n(p_1,\ldots, p_n) := -\sum_{j=1}^n u \cdot p_j\:.$$
In the more implicit form, where $Z_i = I\otimes \cdots \otimes  I\otimes Z \otimes I \otimes \cdots \otimes I$ and $Z$ occupies the $i$-th slot,
\beq
X_{u,\Sigma}^a  = \bigoplus_{n=1}^{+\infty}  \sum_{i=1}^n\frac{1}{2} \left\{ N_{\Sigma, i}^a, \frac{H^{u}_i}{H^u} \right\} \label{CEr}
\eeq
(\ref{CE}) and (\ref{CEr}) are nothing but a {\em quantum version} of a classical expression of this sort\footnote{Notice that the above expression becomes the expression of the coordinates of the standard center of mass in the limit of large mass, when $E^n_u \to n u^0 m$.}
$$X^a= \frac{\sum_{n=0}^{+\infty} E^n_u x_n^a}{\sum_{n=0}^{+\infty} E^n_u }$$
which, however, takes the non-commutativity of the involved operators into account, as well as the structure of the Fock space.

 $X_{u,\Sigma}^a$ is evidently symmetric when defined on $\gS_0$
and presumably it is also
essentially selfadjoint. We only comment that its restriction to the one-particle space $\cH^{(1)} =\cH_m$ is essentially selfadjoint because it coincides there with the restriction of the Newton-Wigner operator, which is essentially selfadjoint on $\cS(\sV_{m,+})$ (see e.g. \cite{M23}). 

We have the following result.\ 

\begin{proposition}
Referring to the notions introduced in Theorem \ref{TEOR}, let $\Sigma$ be any $3$-dimensional rest space, and let $x^0,x^1,x^2,x^3$ be Minkowskian coordinates adapted to $u_\Sigma$, with $\Sigma$ described by $x^0=0$. Assume that $f\in \cD_\bR(\bM)$ satisfies
\beq\label{FM0}
\int_\bM f^2 d^4x =1,\qquad \int_\bM x^\alpha f(x)^2 d^4x=0\quad \alpha=0,1,2,3.
\eeq
Then
\beq 
\int_{\Sigma} x^a \langle \Psi | A^u_f(dx) \Psi \rangle  = \langle \Psi | X^a_{u,\Sigma} \Psi\rangle\quad \mbox{for  $\Psi \in \gS_0$ and $a=1,2,3$.}
\eeq
In particular, for one-particle states
 \beq 
\int_{\Sigma} x^a \langle \psi | A^u_f(dx)\psi \rangle= \langle \psi | N^a_\Sigma \psi\rangle\quad \mbox{for  $\psi \in \cS(\sV_{m,+})$ and $a=1,2,3$,}
\eeq
irrespective of $u\in \sT_+$. A sufficient concrete way to guarantee \eqref{FM0} is to take, in the same adapted coordinates,
\beq\label{FMSUFF}
f(x^0,\vec x)^2=h_0(x^0)^2h(\vec x)^2,
\eeq
with the two factors separately normalized, $h_0$ even and $h(\vec x)=H(|\vec x|)$; more generally, the proof only uses the four vanishing first moments in \eqref{FM0}.
\end{proposition}

\begin{proof} Take $\Psi \in \cH^{(0)\perp}\cap \gS_0$ so that $\Psi_n=0$ for $n>N_\Psi$, which makes meaningful the infinite sum below. Set $g:=f^2$ and
\[
c_g^\alpha:=\int_\bM x^\alpha g(x)d^4x,\qquad \alpha=0,1,2,3.
\]
By hypothesis, $c_g^\alpha=0$ for every $\alpha$. According to the Hilbert space isomorphism $J_n$ defined in (\ref{UM}), and using $E^n_{\epsilon=0,u}= E^n_u$, we have from (\ref{NEW}) and referring to (\ref{igt2}),
\begin{align*}
\int_{\Sigma} x^a \langle \Psi | A^u_f(dx) \Psi \rangle 
=&-i\sum_{n=1}^{+\infty}\frac{n}{2\pi} \int_\Sigma  \int_{\bR^{3(n-1)}} \int_{\bR^6}\left( \frac{\partial}{\partial k^a} e^{i (\vec{k}-\vec{p})\cdot \vec{x}} \right) \widehat{g}(p-k)\\
&\quad \times \frac{\overline{\Phi(\vec{p}, \vec{Q})} \Phi(\vec{k},\vec{Q})}{\sqrt{{E^n_{u}(p,Q)} E^n_{u}(k,Q)}} u^\mu t_{\mu 0}(p,k)\frac{d^3pd^3k d^{3(n-1)}q}{\sqrt{E(p)E(k)}} d^3 x .
\end{align*}
For fixed $n$ and $Q$, put
\[
K_n(p,k,Q):=\frac{\overline{\Phi(\vec{p}, \vec{Q})} \Phi(\vec{k},\vec{Q})}{\sqrt{{E^n_{u}(p,Q)} E^n_{u}(k,Q)}}\,u^\mu t_{\mu 0}(p,k)\frac{1}{\sqrt{E(p)E(k)}} .
\]
After integrating over $\vec x$ and integrating by parts in $\vec k$, the previous expression becomes
\begin{align*}
\int_{\Sigma} x^a \langle \Psi | A^u_f(dx) \Psi \rangle
=&\sum_{n=1}^{+\infty} i\frac{n}{2\pi}(2\pi)^3\int_{\bR^{3n}}\left. \frac{\partial}{\partial k^a}\left(\widehat{g}(p-k)K_n(p,k,Q)\right)\right|_{k=p}d^3p d^{3(n-1)}q .
\end{align*}
The derivative acts on both factors. The contribution produced by differentiating $\widehat g(p-k)$ is
\beq\label{EXTRATERM53}
\sum_{n=1}^{+\infty} i\frac{n}{2\pi}(2\pi)^3\int_{\bR^{3n}}\left. \frac{\partial}{\partial k^a}\widehat{g}(p-k)\right|_{k=p} |\Phi(\vec p,\vec Q)|^2\frac{E_u(p)}{E_u^n(p,Q)}d^3p d^{3(n-1)}q .
\eeq
With our Fourier convention,
\[
\left. \frac{\partial}{\partial k^a}\widehat{g}(p-k)\right|_{k=p}
=\frac{i}{(2\pi)^2}\left(c_g^a-\frac{p^a}{E(p)}c_g^0\right)=0,
\]
where the last equality follows from \eqref{FM0}. Hence \eqref{EXTRATERM53} vanishes. We are left with the derivative of $K_n$. Since $\widehat g(0)=(2\pi)^{-2}$ by \eqref{FM0}, and since a direct computation gives
\[
K_n(p,p,Q)=|\Phi(\vec p,\vec Q)|^2\frac{E_u(p)}{E_u^n(p,Q)},\qquad
\left.\frac{\partial}{\partial k^a}\frac{u^\mu t_{\mu0}(p,k)}{\sqrt{E_u^n(p,Q)E_u^n(k,Q)}\sqrt{E(p)E(k)}}\right|_{k=p}
=\frac{1}{2}\frac{\partial}{\partial p^a}\frac{E_u(p)}{E_u^n(p,Q)},
\]
we obtain
\begin{align*}
\int_{\Sigma} x^a \langle \Psi | A^u_f(dx) \Psi \rangle
=&\sum_{n=1}^{+\infty}n \int_{\bR^{3n}}\overline{\Phi(\vec{p}, \vec{Q})}\left( i\frac{\partial}{\partial p^a} \Phi(\vec{p},\vec{Q})\right)
\frac{E_u(p)}{E_u^n(p,Q)}d^3p d^{3(n-1)}q\\
&+\sum_{n=1}^{+\infty}n \int_{\bR^{3n}}\overline{\Phi(\vec{p}, \vec{Q})}  \Phi(\vec{p},\vec{Q})\frac{i}{2} 
\frac{\partial}{\partial p^a} \frac{E_u(p)}{E_u^n(p,Q)} d^3p d^{3(n-1)}q.
\end{align*}
The above expression can be rearranged as
\[
\sum_{n=1}^{+\infty}n \int_{\bR^{3n}}\overline{\Phi(\vec{p}, \vec{Q})} \frac{1}{2}\left\{ i\frac{\partial}{\partial p^a}, \frac{E_u(p)}{E_u^n(p,Q)}\right\}  \Phi(\vec{p},\vec{Q})d^3p d^{3(n-1)}q.
\]
Since $ \Phi(\vec{p},\vec{Q})$ is completely symmetric in its $n$ arguments, the sum above can be rewritten, for every $\Psi \in \gS_0\cap \cH^{(0)\perp}$, as
\beq
 \int_{\Sigma} x^a \langle \Psi | A^u_f(dx) \Psi \rangle = \sum_{n=1}^{+\infty} \int_{\bR^{3n}}\overline{\Phi(\vec{p}_1,\ldots, \vec{p}_n )}\sum_{j=1}^n \frac{1}{2}\left\{ i\frac{\partial}{\partial p_j^a}, \frac{E_u(p_j)}{E_u^n(p_1,\ldots, p_n)}\right\} \Phi(\vec{p}_1,\ldots, \vec{p}_n ) d^{3n}p.\nonumber
\eeq
This is the thesis, since the action of $N^a_\Sigma$ on the representation induced by the isomorphism $J_n$ on the Schwartz functions $\Phi$ representing the state vectors $\Psi$ is just $i\frac{\partial}{\partial p^a}$, whereas the multiplicative operators -- like $\frac{E_u(p)}{E_u^n(p,Q)}$ -- which are only functions of the momenta $\vec{p}_i$, are invariant under the unitary map $J_n$.

\end{proof}

The proof also makes explicit what happens if the centering condition in (\ref{FM0}) is not imposed. With $g:=f^2$ and
\[
c_g^\alpha:=\int_\bM x^\alpha g(x)d^4x\,,\qquad \alpha=0,1,2,3,
\]
the contribution displayed in (\ref{EXTRATERM53}) does not vanish and gives the additional term
\begin{align}
&-\sum_{n=1}^{+\infty}\int_{\bR^{3n}}|\Phi_n(\vec p_1,\ldots,\vec p_n)|^2
\sum_{j=1}^n\frac{E_u(p_j)}{E_u^n(p_1,\ldots,p_n)}
\left(c_g^a-\frac{p_j^a}{E(p_j)}c_g^0\right)d^{3n}p .\label{FMCORR}
\end{align}
Thus, in the absence of (\ref{FM0}), the first moment of the POVM is not exactly the center of $u$-energy, but is shifted by the $f$-dependent quantity (\ref{FMCORR}). In particular, in the one-particle sector, for $\psi\in \cS(\sV_{m,+})$ and $\Phi=J_1\psi$, one obtains
\beq\label{FMCORR1}
\int_\Sigma x^a \langle \psi|A_f^u(dx)\psi\rangle
=\langle \psi|N_\Sigma^a\psi\rangle
-c_g^a\|\psi\|^2+c_g^0\int_{\bR^3}|\Phi(\vec p)|^2\frac{p^a}{E(p)}d^3p .
\eeq
For normalized one-particle states, the correction is therefore the spatial displacement $-c_g^a$ plus the time-displacement contribution $c_g^0\langle p^a/E(p)\rangle_\psi$. The hypotheses (\ref{FM0}) are precisely those which remove this dependence on the spacetime center of the smearing function.

Despite its mathematical interest, it is difficult to accept the proposed interpretation of the general $A^u_f$ from a physical perspective for $n\neq 1$. That is because the center of energy of a quantum field {\em does not seem to have the propensity to localize in space}! Or, at least, it is really difficult to imagine measurement experiments in which the field is eventually localized at a point in space. The preceding proposition should therefore be understood as identifying the first moment of the localization POVM with the center of $u$-energy only after the natural choice of the spacetime origin has been tied to the center of the smearing profile $f^2$.

By contrast, {\em states of single particles do seem to have this propensity}. If we restrict the above result to states $\Psi\in \cH^{(1)}=\cH_m$, and if $f^2$ satisfies (\ref{FM0}), we once again find the result of \cite{M23,DM24}: the first moment of the POVM on $\Sigma$ is the restriction to $\sS(\cH_m)$ of the Newton-Wigner selfadjoint operator. The novelty with respect to the quoted result is that this conclusion is independent of the detailed shape of the smearing function $f$ used to construct the POVM $A^u_f$, provided $f^2$ is normalized and centered in the sense of (\ref{FM0}). If $f^2$ is not centered, the same computation gives instead (\ref{FMCORR1}), so that the Newton-Wigner first moment is recovered only after subtracting the explicit displacement determined by the first spacetime moments of $f^2$. In the $n=1$ case, the proposed physical interpretation has some chance of being meaningful, since single particles do have the propensity to localize in space under localization experiments. The relevance of this result is that the considered single-particle localization observable is now obtained from standard local and quasi-local observables of QFT in a rigorous way.
\beq 
 \langle \psi | {A}_{f}^u(\Delta) \psi \rangle = \lim_{\epsilon \to 0^+} \int_{\Delta}\sp
 \left\langle   \frac{1}{\sqrt{H^u_{\epsilon}}}  \psi \left|     :\spa \hat{T}_{\mu\nu} \spa: [f^2_x]\frac{1}{\sqrt{H^u_{\epsilon}}}  \right. \psi \right\rangle  u^\mu u^\nu_\Sigma  d\Sigma(x)\nonumber 
 \eeq
defines a POVM on every $\Sigma$ when $\psi \in \cH^{(1)}= \cH_m$, and the family of these POVMs is a positive-energy  relativistic spatial localization observable for a single particle which satisfies CC.  The centering assumptions affect the identification of the first moment with $N^a_\Sigma$, but not the definition of the POVM nor the validity of CC.

We finally observe that {\em Heisenberg's inequality} has to be modified according to the new notion of localization (Proposition 61 in \cite{DM24}). The improved expression reads, for $a=1,2,3$,
$$\Delta_\psi x^a \Delta_\psi P_a \geq \frac{\hbar}{2} \sqrt{1 + 4 (\Delta_\psi P_a)^2\langle \psi | \sK_a \psi \rangle}\:.$$
Above, $\Delta_\psi x^a$ is the standard deviation of the distribution of the coordinate $x^a$ in the state represented by the normalized vector $\psi \in \cS(\sV_{m,+})$,
$\Delta_\psi P_a$ is the analogous quantity for the $a$-th component of the momentum, and $\sK_a$ is a selfadjoint positive operator which is a spectral function of $P_a$, in principle depending on $u, f, \Sigma$.
This statement applies to the centered situation in which the first moment agrees with the Newton-Wigner operator; if the centering condition is not imposed, the mean position is shifted according to (\ref{FMCORR1}), while the unsharp nature of the probability distribution is still governed by the same POVM.

See \cite{M23,DM24} for technical discussions on this subject. We only stress that, in the large-mass limit, the above inequality becomes the standard Heisenberg inequality.

\subsection{Large-mass/non-relativistic limit: von Neumann unsharp position measurement}
Restricting ourselves to the one-particle case $\psi \in \cH_m$, we consider, roughly speaking, the large-mass  limit. More precisely, we study wave packets for which the relevant values of the momentum $\vec{p}$ satisfy $|\vec{p}| <\sp< m$ in natural units. In fact, this limit admits a double interpretation. Restoring the speed of light, the condition becomes $|\vec{p}| <\sp< mc$, which may be realized either because $c$ is large, corresponding to the genuinely non-relativistic limit, or because the mass is large while the theory remains relativistic. For most of the issues discussed here, however, there is no need to distinguish between these two possibilities.

We consider a one-particle state $\psi \in \cD(\sV_{m,+})$ such that the values of $\vec{p}$ for which $(E(p),\vec{p}) \in supp(\psi)$ are negligible with respect to $m$ and the value of $p^0=E(p)$ is very close to $m$. Within this approximation we replace $E(p)= \sqrt{\vec{p}^2 +m^2}$ with $m$ and $E_u(p)$ with $mu^0$, etc., at each occurrence. The probability of finding the particle in a bounded measurable set $\Delta \subset \Sigma$, where $\Sigma$ coincides with $x^0=0$, is
$$\langle \psi| A_f^u(\Delta) \psi\rangle = \left\langle   \frac{1}{\sqrt{H^u}} \psi \left|  \int_{\Delta}\sp    :\spa \hat{T}_{\mu\nu} \spa: [f^2_x] u^\mu u^\nu_\Sigma \frac{1}{\sqrt{H^u}}  d\Sigma(x) \right. \psi \right\rangle\:.$$
Expanding the right-hand side as in (\ref{OST}) and replacing $E(p)$ and $E_u(p)$ by $m$ and $u^0m$, respectively, and systematically disregarding terms of order $\vec{p}\cdot \vec{k}$, $\vec{k} p_0$, $\vec{p} k_0$ in $t_{\mu\nu}(p,k)$ in comparison with terms of type $m^2$ or $p_0k_0$ which, in turn, are replaced by $m^2$,
a lengthy but elementary computation yields
\beq \langle \psi| A_f^u(\Delta) \psi\rangle \simeq \int_{\Delta} d\Sigma(x)  \int_\Sigma \: g_f(\vec{y}-\vec{x}) |\phi(\vec{y})|^2  d\Sigma(y) = \int_\Sigma   (\chi_\Delta * g_f)(\vec{y}) |\phi(\vec{y})|^2 d\Sigma(y) \label{vNM} \eeq where\footnote{The normalization factor, which includes the mass, in the classical limit is automatically embodied in the definition of the wavefunction in momentum representation. It appears here as a relic of the normalization with respect to $\mu_m(p)$ on the mass shell and stems from the fact that our discussion obviously does not apply to the massless case.}
$$\phi(\vec{x}) :=\frac{1}{\sqrt{m}} \int_{\bR^3} \frac{e^{i\vec{p}\cdot \vec{x}}}{(2\pi)^{3/2}} \psi(\vec{p}) d^3p\:, \quad g_f(\vec{x}) := \int_\bR f(x^0,\vec{x})^2 dx^0\:.$$
Thus $g_f \in \cD_\bR(\Sigma)$, $g_f\geq 0$ and, if $\int_\bM f^2d^4x=1$, then $\int_\Sigma g_f d\Sigma=1$. If the centering assumptions (\ref{FM0}) hold, then in particular $g_f$ has vanishing spatial first moment. In that case the unsharp measurement kernel is centered on the point to be detected. In the non-relativistic approximation the rest space $\Sigma$ can be associated equivalently with $u$ or $u'$.
Expression (\ref{vNM}) is just a {\em von Neumann model of an unsharp (indirect) position measurement} (see e.g. Section 2.3.1 of \cite{B07}) where $g_f$ is related to the wavefunction of the {\em probe} particle used to measure the position of a particle of the quantum field of mass $m$ (and this may suggest an indication toward a more concrete interpretation of the {\em smearing procedure} in QFT). If $g_f$ is not centered, the same formula describes an unsharp measurement whose apparatus response is displaced by the first spatial moment of $g_f$, in agreement with the relativistic correction (\ref{FMCORR1}). \\

\begin{remark}\label{REMPOSm}
{\em We consider, in a given reference frame adapted to $\Sigma$ -- equivalently associated with $u$ or $u'$ in the proper non-relativistic case $c\to +\infty$, as discussed above -- smearing functions of the form $f(x) =  h(\vec{x}) h'(x^0)$, where $h,h'$ are real, smooth, compactly supported functions, which we can always assume to be normalized,
 $\int_\bR h'^2dx^0 =1$, and $\int_\Sigma h(\vec{x})^2 d\Sigma =1$. In this case $g_f =h^2$ by construction. If, in addition, $h'$ is even and $h$ is radial (or, more generally, $h^2$ has vanishing spatial first moment), then $f^2$ satisfies (\ref{FM0}). The radiality of $h$ is only a simple sufficient condition; the actual requirement is the vanishing of the first moments of $f^2$. We can consider three regimes.
\begin{itemize}
\item[(1)]
As expected from results in \cite{M23}, the considered localization probability in the large-mass limit tends to become the standard one predicted by non-relativistic quantum mechanics in the limit where the centered function $g_f(\vec{x})$ tends to $\delta(\vec{x})$.  
\item[(2)] If, conversely, the centered function $g_f$ tends to become a constant function on larger and larger regions (while preserving the value of its integral), we obtain a more and more imprecise notion of localization in the large-mass limit. That is because the convolution $\chi_\Delta * g_f$ with an almost constant function makes indistinguishable the characteristic functions of a pair of distinct sets $\Delta$ and $\Delta'$. 
\item[(3)] Finally, the large-mass/non-relativistic approximation does not depend on the detailed choice of the temporal function $h'$, since only $\int h'^2dx^0$ enters (\ref{vNM}). However, if one also wants the exact first-moment identity of Proposition 5.3, $h'$ must have vanishing first moment, for instance it may be chosen even about $x^0=0$. With that restriction, $h'$ can still be taken to have arbitrarily large support (preserving the constant value of its integral) or to tend to a delta function in time centered at $x^0=0$. The former regime is a convenient setup for locally minimizing the negative-energy gap, as we shall discuss in Section \ref{Nen}. \hfill $\blacksquare$
\end{itemize}}
\end{remark}

\section{Commutativity of conditional localization POVMs of causally separated laboratories}\label{SECFagg}
The families of effects constructed in Theorem \ref{TEOR}, though arising from (quasi)local observables of QFT, do not satisfy the commutativity requirement for causally separated detection regions $\Delta$ and $\Delta'$. This is not a {\em genuine} physical problem, as discussed in \cite{Mor126}. The propensity of a particle to localize at a single position does not permit, for instance, invoking the no-signaling principle, which would in turn imply commutativity of the effects. What we intend to investigate here is whether commutativity can be restored by passing to a notion of {\em conditional probability} for {\em finite-size laboratories}, as discussed in Section 5.5 of \cite{Mor126}.

From the purely mathematical viewpoint, non-commutativity of effects $A^u_f(\Delta), A^{u'}_f(\Delta')$ for $\Delta,\Delta'$ causally separated is a consequence of a pair of features. 

\begin{itemize}
\item[(a)] First of all, non-commutativity arises from the orthogonal projectors $P_n$ -- and we are interested in the special case $n=1$ -- which appear in (\ref{igt2}). They are not local observables in the sense of AHK.
\end{itemize}

\noindent This issue can easily be circumvented by moving the projectors from the operator to the states. In other words, {\em we} choose to deal only with one-particle states\footnote{The issue is not completely solved from a philosophical viewpoint, since, in order to know whether a field state contains only one particle (or a definite number $n$ of particles), one should collect the whole information about the state from an entire rest space, whereas instruments work locally. However, we shall not address this {\em second-order} issue here.} and, in order to compute the detection probability of a single particle, we directly use the {\em operators} $\sA_{f,\epsilon}^u(\Delta)$ defined in (\ref{sA}), eventually taking the limit as $\epsilon\to 0^+$,
instead of the {\em effects} $A_{f}^u(\Delta)$ as in (\ref{igt2}).
Obviously, {\em when dealing with states in $\sS(\cH^{(n)})$ with $n>0$}, this is equivalent to the other way around, as stated in (\ref{LIM}), since the projectors are already embodied in the states $\rho = P_n\rho P_n$. However, it is worth stressing that the operators $\sA_{f,\epsilon}^u(\Delta)$, though in $\gB(\gF_s(\cH_m))$, are not {\em effects}, contrary to 
$A_{f}^u(\Delta)$, since $\sA_{f,\epsilon}^u(\Delta) \not \geq 0$ in general  {((d) Proposition \ref{COS} below)}, even if they are bounded from below  ((e) Proposition \ref{LLAST}).

\begin{itemize}
\item[(b)] The second source of non-commutativity is the appearance of the non-local operators $\frac{1}{\sqrt{H^u_{\epsilon}}}$ in the above expression: they do not commute
with       
$ :\spa \hat{T}_{\mu\nu} \spa: [f^2_x]$. 
\end{itemize}

\noindent In principle, as anticipated, this issue could be addressed by referring to {\em conditional POVMs} localized in {\em laboratories}, as discussed in the introduction. A laboratory $L(\Delta_0) := (\Delta^\pperp_0)^\pperp\subset \bM$ is defined by assigning a suitable bounded region (typically a non-empty open set with compact closure) $\Delta_0$ in a rest frame $\Sigma$, and one is interested in the (conditional) probability of detecting a particle in subsets $\Delta \subset \Delta_0$. 
%
We expect that, though the localization effects of a single laboratory do not commute, the localization effects associated with a pair of (sharply) causally separated laboratories
based on $\Delta_0$ and $\Delta'_0$
do, if these effects refer to conditional probabilities. That is because these operators are supposed to be constructed in terms of proper local operators, localized in neighborhoods of the laboratories. Causal separation depends not only on the regions $\Delta_0, \Delta_0'$ but also on the support of the smearing function $f^2$ used in the definition of the stress-energy tensor operator.
 
According to \cite{Mor126}, the idea is therefore to define a ``conditioned POVM'' in the laboratory based on $\Delta_0$, whose effects are labeled by regions $\Delta \subset \Delta_0$, 
and have a form of this type
\beq\label{OPop}
V\frac{1}{\sqrt{\sA^u_{f, \epsilon}(\Delta_0)}} \sA^u_{f, \epsilon}(\Delta) \frac{1}{\sqrt{\sA^u_{f, \epsilon}(\Delta_0)}}V^\dagger\:.
\eeq
(More precisely, these POVMs have the interpretation of conditional POVMs for states whose probability of finding the system in $\Delta_0$ is close to $1$ according to (\ref{G}) and the discussion below it.)
As discussed in \cite{Mor126}, there is the possibility that 
effects of this sort, referred to different regions $\Delta_0, \Delta'_0$, included in causally separated neighborhoods, commute.
 
The evident problem is that the operators $\sA^u_{f, \epsilon}(\Delta)$ and $\sA^u_{f, \epsilon}(\Delta_0)$ are {\em not} positive! So the above construction seems pointless because it would imply some sort of ``negative probabilities'' (or even worse, a complex notion of probability, since square roots come into play).
We want to address this issue in the next sections by using energies directly instead of probabilities.

We stress that, if an operator of the form \eqref{OPop}, or of a similar kind, is required to be positive and to belong to a local algebra in the AHK sense, then it should be viewed as the restriction to the one-particle space of an operator defined on the full Hilbert space, as in the discussion of (a). It then follows from the Reeh--Schlieder theorem (Proposition \ref{XYZ}) that positivity on the full Hilbert space is incompatible with the requirement that the vacuum expectation value of the operator vanish.
Accordingly, detectors formalized in this manner necessarily exhibit the well-known phenomenon of \emph{dark counts}.
This is a familiar issue in local quantum physics, and it has recently been revisited in a quantitative framework in \cite{FC26}. In the scattering-theory literature, by contrast, a different standpoint is usually adopted: detector operators are assumed to annihilate the vacuum,
$
B\Omega = 0
$,
and are therefore non-local. Instead, one works with \emph{quasi-local} operators, as is done in Haag--Ruelle scattering theory; see in particular \cite{Haag,Araki}.
In the present paper, we instead consider a genuinely local detector constructed from the local energy density of a quantum field. In this setting, the occurrence of dark counts is unavoidable.
\subsection{Bounds on negative energy  in finite laboratories}\label{Nen}
Instead of considering probabilities, let us focus directly on local energies.
We can try to define detection probabilities starting from a local notion of Hamiltonian, and compare the energy content in $\Delta$ of a one-particle state with the energy content in $\Delta_0\supset \Delta$ of that state.
The operators associated with local energy are obviously affected by the problem of negative energy, which in turn would give rise to ``negative probabilities''. In some sense these are the same ``negative probabilities'' measured by the operators $\sA^u_{f, \epsilon}(\Delta_0)$. From now on, if $\Sigma$ is a rest space, we define the subfamily  of  bounded Borel sets \beq \cB_b(\Sigma) := \{\Delta \in \cB(\Sigma) \:|\: \mbox{$\Delta$ is bounded}\}\:.\eeq

\begin{proposition} \label{COS} Take $u\in \sT_+$, $f \in \cD_\bR(\bM)$ and suppose that $\Delta \in \cB_b(\Sigma)$ for a rest space $\Sigma$.
There is an operator $\sH_f^u(\Delta) : \gS_0 \to \gF_s(\cH_m)$, called the {\bf local Hamiltonian} associated with $\Delta$,
that is uniquely defined by 
\beq\label{HD}
\langle \Psi| \sH_f^{u}(\Delta)\Psi' \rangle =  \int_{\Delta} \langle \Psi| :\spa \hat{T}_{\mu\nu} \spa: [f^2_x] \Psi'\rangle u^\mu   u^\nu_\Sigma  d\Sigma(x) \quad \mbox{for every $\Psi,\Psi'\in \gS_0$.}
\eeq
In particular it holds
\beq\label{PART}
\sH_f^{u}(\Delta) =  :\spa \hat{T}_{\mu\nu} \spa: [f_\Delta] u^\mu   u^\nu_\Sigma \quad \mbox{where} \quad f_\Delta(y) := \int_\Delta f^2(y-x) d\Sigma(x)\:.
\eeq
The following further facts are valid.
\begin{itemize}
\item[(a)] $\sH_f^{u}(\Delta)$ is symmetric, essentially selfadjoint, leaves $\gS_0$ invariant, and satisfies the covariance relation
$$U_g \sH_f^{u}(\Delta)U_g^{-1}= \sH_{g_*f}^{gu}(g\Delta)\quad \mbox{for every $g\in IO(1,3)_+$.}$$
\item[(b)] {If $\sH_f^{u}(\Delta)\neq 0$, then $\sH_f^{u}(\Delta)$} is not positive, but it is bounded from below.
\item[(c)] If $\Delta \subset \Sigma$ is an open set with compact closure and the globally hyperbolic region $L(\Delta):= (\Delta^{\pperp})^{\pperp}$ is the associated laboratory, then 
\beq\label{HD2gen}
\langle \Psi| \sH_f^{u}(\Delta)\Psi' \rangle =  \int_{S} \langle \Psi| :\spa \hat{T}_{\mu\nu} \spa: [f^2_x] \Psi'\rangle u^\mu   u^\nu_S(x)  dS(x) \quad \mbox{for every $\Psi,\Psi'\in \gS_0$.}
\eeq
Above $S\subset L(\Delta)$ is any spacelike smooth Cauchy surface of $ L(\Delta)$ and $u_S(x)$ its future-oriented unit normal vector at $x\in S$, $dS(x)$ being the metric-induced measure on $S$.
\item[(d)] Taking (\ref{FONDH}) into account, if $\epsilon>0$, and $\Psi,\Psi' \in \gS_0$
\beq  \langle \Psi| \sA^u_{f,\epsilon}(\Delta)\Psi'\rangle  = \left\langle \Psi \left| {\frac{1}{\sqrt{H^u_{\epsilon}}} \sH_f^{u}(\Delta) \frac{1}{\sqrt{H^u_{\epsilon}}}}\right. \Psi' \right\rangle  \:.\label{IDf}\eeq
{In particular, if  $\sA^u_{f,\epsilon}(\Delta)\neq 0$, then it  is not positive due to item (b) above.}
\end{itemize}
Everything asserted remains valid if one replaces $f^2$ by $f= \sum_{i=1}^N f_i^2$ with $f\in \cD_\bR(\bM)$.
\end{proposition}

\begin{proof}  See Appendix \ref{AAA}.
\end{proof}
Failure of positivity of $\sH_f^u(\Delta)$ arises from (\ref{PART}) together with  Proposition \ref{XYZ}.
Nevertheless, taking advantage of some crucial results in \cite{F12}, we are about to prove that, referring to a Minkowskian coordinate system adapted to $u$, 
we can make the ``negative-energy gap'' as small as desired, in every reference frame $u'$, generally different from $u$, by choosing the support of the temporal part $h'=h'(x^0)$ of the smearing function $f(x)= h(\vec{x}) h'(x^0)$ suitably large (the coordinates being adapted to $u$). It is worth stressing that this can be done while leaving unchanged the spatial part $h(\vec{x})$ of the smearing function. This function is responsible for the precision of the spatial position measurement according to Remark \ref{REMPOSm} in the non-relativistic  limit. 
In practice, we integrate in space a corrected version of the energy density by adding a term to the normally ordered stress-energy tensor operator depending on a small parameter $\eta>0$. Since we deal with spatial regions of finite extent, the amount of added energy remains finite and can be made arbitrarily small by suitably enlarging the support of the temporal smearing function $h'$.

Henceforth $\sH_{f,\eta}^{u}(\Delta)$
is the unique operator such that, for every $\Psi,\Psi'\in \gS_0$,
 \beq\label{HD2}
\langle \Psi| \sH_{f,\eta}^{u}(\Delta)\Psi' \rangle =  \int_{\Delta}\int_{\bM} \left\langle \Psi\left| \left(:\spa \hat{T}_{\mu\nu} \spa: (y) - \eta g_{\mu\nu}(y) I\right) \right.\Psi'\right\rangle f(y-x)^2 u^\mu   u^\nu_\Sigma d^4y  d\Sigma(x) \:.\quad \eeq
for every $u\in \sV_+$, every bounded measurable $\Delta\subset \Sigma$, and every rest space $\Sigma$ adapted to the reference frame $u_\Sigma$,
$f\in \cD_\bR(\bM)$, and $\eta\geq 0$. It immediately follows -- recalling that $u\cdot u_\Sigma <0$! -- that we can equivalently define
\beq  {\sH}^{u}_{f,\eta}(\Delta)  := \sH_f^{u}(\Delta) + \eta|u\cdot u_\Sigma|  |\Delta| I \:. \label{DEFhhh}\eeq
In particular $ {\sH}^{u}_{f,\eta}(\Delta) $ is symmetric and essentially selfadjoint since $ \eta|u\cdot u'_\Sigma|  |\Delta|\in \bR$ and $ \sH_f^{u}(\Delta) $ is symmetric and essentially selfadjoint.

Notice that the improved stress-energy tensor operator induced by the formal quadratic-form density $:\spa \hat{T}_{\mu\nu} \spa: (y)_\eta  := :\spa \hat{T}_{\mu\nu} \spa: (y) - \eta g_{\mu\nu} I$ is automatically conserved in the usual distributional sense $ :\spa \hat{T}_{\mu\nu} \spa: [\partial^\mu f]_\eta= 0$, since the covariant derivative is metric-compatible (the metric is even constant in Minkowskian coordinates!).
Finally, (a), (b), and (c) of Prop. \ref{COS} remain valid for $\sH_{f,\eta}^{u}(\Delta)$, with the latter obtained by replacing the stress-energy tensor by the improved one, and the covariance relation in (a) reads
 $$U_g \sH_{f,\eta}^{u}(\Delta)U_g^{-1}= \sH_{g_*f, \eta}^{gu}(g\Delta)\quad \mbox{for every $g\in IO(1,3)_+$.}$$
Furthermore, as anticipated, (b) can be substantially improved as follows. \\

\begin{theorem}\label{TEOFD}  Take $u \in \sT_+$, a Minkowskian coordinate system $x^0,x^1,x^2,x^3$ adapted to $u$, and a non-vanishing function $h \in \cD_\bR(\bR^3)$ with  $\int h(\vec{x})^2 d^3x=1$. Finally choose an arbitrarily small $\eta>0$.\\
There exists $h'\in \cD_\bR(\bR)$ with $\int h'(x^0)^2 dx^0=1$ such that, defining $f:= h'h$, we have
\beq\label{ASDED}
  \left\langle \Psi\left| \left(:\spa \hat{T}_{\mu\nu} \spa: [f^2_x] - \eta g_{\mu\nu}I\right) \right.\Psi\right\rangle u^\mu   u'^\nu  \geq c^u_{\eta,f}  |u\cdot u'| ||\Psi||^2 \quad \mbox{if $u'\in \sT_+$, $x\in \bM$, $\Psi \in \gS_0$}
  \eeq
for some finite constant $c^u_{\eta,f}>0$ independent of $u'$. As a consequence
\begin{itemize}
\item[(a)] if $\Psi \in \gS_0$, the smooth conserved current $J_\nu(x) :=-\left\langle \Psi\left| \left(:\spa \hat{T}_{\mu\nu} \spa: [f^2_x] - \eta g_{\mu\nu} I\right) \right.\Psi\right\rangle  u^\mu$ is causal and future-directed wherever it does not vanish;
    \item[(b)] the symmetric and essentially selfadjoint operators ${\sH}^{u}_{f,\eta}(\Delta)$ are positive and monotone:
  \beq \label{NT2}
 {\sH}^{u}_{f,\eta}(\Delta)  \geq  {\sH}^{u}_{f,\eta}(\Delta') \geq 0 \quad \mbox{if $\Delta \supset \Delta'$ are in  $\cB_b(\Sigma)$, for every  rest space $\Sigma$;}
\eeq
    \item[(c)] the operators ${\sH}^{u}_{f,\eta}(\Delta)$ are strictly positive:
  \beq\label{NT}
  {\sH}^{u}_{f,\eta}(\Delta) \geq c^u_{\eta,f}  |u\cdot u_\Sigma|  |\Delta| I  \geq 0  \quad \mbox{if $\Delta\in \cB_b(\Sigma)$, for every  rest space $\Sigma$.}
  \eeq
\end{itemize}
\end{theorem}
\begin{proof} 
The proof of (\ref{ASDED}) relies on the following lemma, which in turn follows from a direct application of general results presented in \cite{F12}.\\
\begin{lemma}\label{LEMMALEMME}
With the main hypotheses, if $h'\in \cD_\bR(\bR)$ and $\Psi \in \gS_0$ with $||\Psi||=1$,
\beq\int_\bR \langle \Psi | :\spa \hat{T}_{\mu\nu}\spa:(x^0,\vec{x}) \Psi \rangle h'(x^0)^2 u^\mu u'^\nu dx^0 \geq -\sqrt{1+ \vec{u}'^2} \frac{1}{16 \pi^3} \int_{m}^{+\infty} 
|\hat{h'}(s)|^2 s^4 Q_3(s) ds
\eeq
where $u^\nu = \delta_0^{\nu}$ and $Q_3 : [1,+\infty) \to [0,1)$ is the smooth, strictly increasing, bounded function 
$$Q_3(x) := \left(1 - \frac{1}{x^2}\right)^{1/2} \left(1 - \frac{1}{2x^2}\right) -\frac{1}{2x^4}\ln\left(x+ \sqrt{x^2-1}\right)$$
which satisfies $Q_3(1)=0$ and $Q_3(x) \to 1$ as $x\to +\infty$.
\end{lemma}
\begin{proof} See Appendix \ref{AAA}.
\end{proof}

\noindent  In the following $f:= h\cdot h'$. 
We therefore have
$$\langle \Psi | :\spa \hat{T}_{\mu\nu}\spa: [f^2]\Psi \rangle u^\mu u'^\nu = \int_{\bR^4} \langle \Psi | :\spa \hat{T}_{\mu\nu}\spa:(x^0,\vec{x}) \Psi \rangle f(x^0,\vec{x})^2 u^\mu u'^\nu dx^0d^3x$$
$$= \int_{\bR^4} \langle \Psi | :\spa \hat{T}_{\mu\nu}\spa:(x^0,\vec{x}) \Psi \rangle h'(x^0)^2h(\vec{x})^2 u^\mu u'^\nu dx^0d^3x \geq -|u\cdot u'| \frac{1}{16 \pi^3} \int_{m}^{+\infty} 
|\hat{h'}(s)|^2 s^4 Q_3(s) ds\:.
$$
As a consequence, since the inequality above is valid for every normalized $\Psi \in \gS_0$ and this space is invariant under the unitary action of the Poincar\'e group, we have in particular that  
$$\langle \Psi | :\spa \hat{T}_{\mu\nu}\spa: [f_x^2]\Psi \rangle u^\mu u'^\nu =  \langle V^{-1}_x\Psi | :\spa \hat{T}_{\mu\nu}\spa: [f^2] V^{-1}_x\Psi \rangle u^\mu u'^\nu 
 \geq -|u\cdot u'| \frac{1}{16 \pi^3} \int_{m}^{+\infty} 
|\hat{h'}(s)|^2 s^4 Q_3(s) ds$$
where $V_x= U_{(I,x)}$ is the unitary representation of spacetime translations (referring to some choice of the origin).
This inequality permits us to prove (\ref{ASDED}). It is sufficient to show that, for every given $\eta>0$, there exists a corresponding  $h'$ such that
\beq \eta > \frac{1}{16 \pi^3} \int_{m}^{+\infty} 
|\hat{h'}(s)|^2 s^4 Q_3(s) ds \label{LIMn}\:.\eeq
If this is true, the required constant $c^u_{\eta,f}$ can be defined as $c^u_{\eta,f} :=( \eta-\eta')$, where $\eta'<\eta$ is positive and sufficiently close to $\eta$ so that (\ref{LIMn}) is still valid with $\eta'$ in place of $\eta$. 
To prove (\ref{LIMn}),
consider a family of smooth compactly supported real functions $h'_\delta(x^0)=\sqrt{\delta} \chi(\delta x^0)$ where $\delta>0$ and $\chi \in \cD_\bR(\bR)$ is such that $\int \chi^2 dx^0 =1$. $\widehat{h'_\delta}$ is Schwartz and $\widehat{h'_\delta}(s)=\frac{1}{\sqrt{\delta}} \widehat{\chi}(s/\delta)$. 
The dominated convergence theorem proves that $\int_{m}^{+\infty} 
|\widehat{h'_\delta}(s)|^2 s^4 Q_3(s) ds=$
$$\delta^{-1}\int_{m}^{+\infty} 
|\widehat{\chi}(s/\delta)|^2 s^4 Q_3(s) ds = \int_{\bR}\chi_{[m/\delta, +\infty)}(y)  \delta^4
|\widehat{\chi}(y)|^2  y^4 Q_3(\delta y) dy\to 0 \quad \mbox{as $\delta \to 0^+$.}$$
Hence (\ref{LIMn}) is valid for a given $\eta>0$ provided one uses $h'=h'_\delta$ with a sufficiently small $\delta>0$.
At this point (a), (b), and (c) follow immediately.
\end{proof}

\subsection{Conditional localization POVMs in laboratories} A suitable candidate for a conditional localization POVM in a laboratory $L(\Delta_0)$, for $\Delta_0 \subset \cB_b(\Sigma)$ with $|\Delta_0|>0$, is expected to be
\beq
B_{\Delta_0}(\Delta) := \frac{1}{
\sqrt{\sH^u_{f,\eta}(\Delta_0)}
} 
\sH^u_{f,\eta}(\Delta)
 \frac{1}{\sqrt{\sH^u_{f,\eta}(\Delta_0)}}\quad \mbox{for ${\cal R} \ni\Delta \subset \Delta_0$} \quad\quad\quad\mbox{(attempt)}
\eeq
where we have chosen $f$ so that, for the given $\eta>0$, (\ref{NT}) is true with $\Delta=\Delta_0$.\\
%
To interpret the square root we need the functional calculus. Since the argument of the square root is an essentially selfadjoint operator, the relevant selfadjoint operator is simply its closure $\overline{{\sH}}^u_{f,\eta}(\Delta_0):= \overline{{\sH}^u_{f,\eta}(\Delta)}$.  Positivity is inherited by the closure. \\

\begin{theorem} \label{TUBB} Take $u\in \sT_+$, an arbitrarily small $\eta>0$, and $f\in \cD_\bR(\bM)$ such that (\ref{NT}) is true. The following facts are valid for every rest space $\Sigma$.
\begin{itemize}
\item[(a)]
If $\Delta_0 \in \cB_b(\Sigma)$ with $|\Delta_0|>0$, then there exists a unique family of operators 
$B_{\Delta_0}(\Delta)^u_{f,\eta} \in \gB(\gF_s(\cH_m))$ for $\Delta\in \cB(\Delta_0)$ such that
\beq\label{BECCA}
 B_{\Delta_0}(\Delta)^u_{f,\eta}  \Psi =\frac{1}{\sqrt{\overline{{\sH}}^u_{f,\eta}(\Delta_0)}} {\sH}^u_{f,\eta}(\Delta)   \frac{1}{\sqrt{\overline{{\sH}}^u_{f,\eta}(\Delta_0)}}\Psi \quad \mbox{for}\quad \Psi \in \gS_{00}
\eeq
where $\gS_{00} := {\sqrt{\overline{{\sH}}^u_{f,\eta}(\Delta_0)}}(\gS_0)$ is dense and $ \frac{1}{\sqrt{\overline{{\sH}}^u_{f,\eta}(\Delta_0)}}\in \gB(\gF_s(\cH_m))$.
\item[(b)] The map $\cB(\Delta_0) \ni \Delta \mapsto B_{\Delta_0}(\Delta)^u_{f,\eta}$ is a normalized POVM absolutely continuous with respect to the Lebesgue measure on $\Sigma$.
\end{itemize}
\end{theorem}

\begin{remark} {\em Before proving the theorem, we observe that if $\Delta_0$ is sufficiently regular that $L(\Delta_0)$
is a globally hyperbolic spacetime in its own right with $\Delta_0$ as a smooth spacelike Cauchy surface, then 
the POVM $B_{\Delta_0}(\Delta)^u_{f,\eta}$ can be extended to a weaker version of relativistic spatial localization observable on the whole spacetime $L(\Delta_0)$.
This can be done with the same technology developed in \cite{DM24}. That mathematical technology, starting from (a) of Theorem \ref{TEOFD},
should also prove that the constructed relativistic spatial localization observable satisfies
a natural generalization of the causal condition CC denoted by GCC therein. If $\Delta_0$ is only Borel, a similar generalization should be possible in the corresponding element of the {\em causal logic} according to \cite{CDRM}, proving therein the validity of an even more general causal condition for achronal sets.
These issues will be investigated elsewhere. \hfill $\blacksquare$}
\end{remark}

\begin{proof}
(a) We need some preliminary lemmas. If $A$ is a symmetric essentially selfadjoint operator as above, its unique selfadjoint extension is its closure $\overline{A}$. The proof of these lemmata uses known properties of the {\em Friedrichs selfadjoint extension} of $A$ and its quadratic form. Notice that the Friedrichs extension exists because $A\geq 0$. By uniqueness, it coincides with $\overline{A}$.\\

\begin{lemma}\label{F2}
If $A :D(A) \to \cH$ is a symmetric essentially selfadjoint operator with $A\geq 0$, then 
\begin{itemize}
\item[(a)] $\overline{\sqrt{\overline{A}}(D(A))} = Ker(\sqrt{\overline{A}})^\perp= Ker(\overline{A})^\perp$;
\item[(b)] if $A \geq cI$ with $c>0$, then $\overline{A} \geq cI$.
In this case: $\overline{\sqrt{\overline{A}}(D(A))}= Ran(\sqrt{\overline{A}}) = \cH$ and $(\overline{A})^{-\alpha}\in \gB(\cH)$ for $\alpha>0$.
\end{itemize}
\end{lemma}

\begin{proof} See Appendix \ref{AAA}.
\end{proof}

\begin{lemma}\label{LEMMAZd}
Consider a pair of symmetric essentially selfadjoint operators $A^0 \geq A\geq 0$ with $D(A)=D(A^0)$.
The following facts hold.
\begin{itemize}
\item[(a)] It holds
$$\left\langle \sqrt{\overline{A^0}} x \left| \sqrt{\overline{A^0}} \right. x \right\rangle \geq \langle\sqrt{\overline{A}} x | \sqrt{\overline{A}}  x \rangle  \quad \mbox{ for $x \in D(\sqrt{\overline{A_0}})$}$$
so that,\\
 $$\langle x |\overline{A^0}  x \rangle \geq \langle x |\overline{A} x \rangle \geq 0 \quad \mbox{for $x \in D(\sqrt{\overline{A^0}})$}$$
and these inequalitiees are in particular valid for $x\in D(A^0) \subset D(\sqrt{\overline{A^0}})$;
\item[(b)] if $Ker (\sqrt{\overline{A^0}}) (= Ker (\overline{A^0}))= \{0\}$, then
 $$||z||^2 \geq  \left\langle \frac{1}{\sqrt{\overline{A^0}}} z \left| A  \frac{1}{\sqrt{\overline{A^0}}}z\right.\right\rangle= \left\langle \frac{1}{\sqrt{\overline{A^0}}} z \left| \overline{A}  \frac{1}{\sqrt{\overline{A^0}}}z\right.\right\rangle \geq 0\quad \mbox{for $z \in \sqrt{\overline{A^0}}(D(A^0))$;} $$ 
\item[(c)] if $A^0 \geq cI$ with $c>0$ then $Ker(\overline{A^0}) = \{0\}$, $\sqrt{\overline{A^0}}(D(A^0))$ is dense, and $1/\sqrt{\overline{A^0}} \in \gB(\cH)$,
so that the quadratic form above uniquely defines an effect in $\cH$ by continuous extension. Therefore this operator is the unique continuous extension of 
$$\frac{1}{\sqrt{\overline{A^0}}} A  \frac{1}{\sqrt{\overline{A^0}}} : \sqrt{\overline{A^0}}(D(A^0)) \to \cH\:.$$
\end{itemize}
\end{lemma}

\begin{proof} See Appendix \ref{AAA}.
\end{proof}
\noindent Returning to the proof of (a), the claim follows immediately from these two lemmas by using $A^0 = \sH^u_{f,\eta}(\Delta_0)$
and $A =  \sH^u_{f,\eta}(\Delta)$, both defined on the common dense domain $\gS_0$, since $\sH^u_{f,\eta}(\Delta_0) \geq \sH^u_{f,\eta}(\Delta)\geq 0$ and 
$\sH^u_{f,\eta}(\Delta_0) \geq c I$ with $c>0$ in view of Theorem \ref{TEOFD}.\\
Concerning (b), we observe that the operators $B_{\Delta_0}(\Delta)^u_{f,\eta}$ are effects, as follows from (b) of Lemma \ref{LEMMAZd}. Furthermore, obviously, 
$B_{\Delta_0}(\Delta_0)^u_{f,\eta}=I$ by construction, since the identity is the unique continuous extension of this quadratic form. The properties of $\sigma$-additivity and absolute continuity are proved by taking advantage of the Vitali-Hahn-Saks theorem, as in the proof of (a) of Proposition 5.5 in \cite{Mor126}.
\end{proof}

\begin{corollary}\label{CORC}
Under the hypotheses of Theorem \ref{TUBB}, $\frac{1}{\sqrt{\overline{{\sH}}^u_{f,\eta}(\Delta_0)}} {\sH}^u_{f,\eta}(\Delta)   \frac{1}{\sqrt{\overline{{\sH}}^u_{f,\eta}(\Delta_0)}}$ is essentially selfadjoint on its natural domain $\gS_{00} = {\sqrt{\overline{{\sH}}^u_{f,\eta}(\Delta_0)}}(\gS_0)$, which therefore is a core:
\beq\label{BECCA2}
 B_{\Delta_0}(\Delta)^u_{f,\eta} =\overline{\frac{1}{\sqrt{\overline{{\sH}}^u_{f,\eta}(\Delta_0)}} {\sH}^u_{f,\eta}(\Delta)   \frac{1}{\sqrt{\overline{{\sH}}^u_{f,\eta}(\Delta_0)}}}\:.
\eeq
The same identity is valid if ${\sH}^u_{f,\eta}(\Delta) $ is replaced above  by any other symmetric  extension of ${\sH}^u_{f,\eta}(\Delta)$.
\end{corollary}

\begin{proof} Observe that the dense subspace $\gS_{00} := {\sqrt{\overline{{\sH}}^u_{f,\eta}(\Delta_0)}}(\gS_0)$ is just the domain of $\frac{1}{\sqrt{\overline{{\sH}}^u_{f,\eta}(\Delta_0)}} {\sH}^u_{f,\eta}(\Delta)   \frac{1}{\sqrt{\overline{{\sH}}^u_{f,\eta}(\Delta_0)}}$, since $\frac{1}{\sqrt{\overline{{\sH}}^u_{f,\eta}(\Delta_0)}}$ is bijective. The unique everywhere-defined continuous extension of the symmetric operator $\frac{1}{\sqrt{\overline{{\sH}}^u_{f,\eta}(\Delta_0)}} {\sH}^u_{f,\eta}(\Delta)   \frac{1}{\sqrt{\overline{{\sH}}^u_{f,\eta}(\Delta_0)}}$ is its closure, which is selfadjoint. Hence this symmetric operator is essentially selfadjoint. In particular, its domain is a core. If we replace ${\sH}^u_{f,\eta}(\Delta)$ by some symmetric extension $E$ of it, we obtain a symmetric extension $\frac{1}{\sqrt{\overline{{\sH}}^u_{f,\eta}(\Delta_0)}}E  \frac{1}{\sqrt{\overline{{\sH}}^u_{f,\eta}(\Delta_0)}}\supset \frac{1}{\sqrt{\overline{{\sH}}^u_{f,\eta}(\Delta_0)}} {\sH}^u_{f,\eta}(\Delta)   \frac{1}{\sqrt{\overline{{\sH}}^u_{f,\eta}(\Delta_0)}}$. Since the right-hand side is essentially selfadjoint, its closure is selfadjoint and thus maximally symmetric. As a consequence, the closure of the left-hand side, which is symmetric as well, must coincide with the closure of the right-hand side. This completes the proof. 
\end{proof}


To conclude this investigation we examine the interplay between  the POVM $B_{\Delta_0}(\Delta)^u_{f,\eta}$ and the global POVM $A_f^u(\Delta)$ on the given $\Sigma\supset \Delta_0$.
We prove that, in fact, the map $\cB(\Delta_0)\ni \Delta \mapsto B_{\Delta_0}(\Delta)^u_{f,\eta}$ can be interpreted as a conditional POVM constructed out of a non-normalized POVM which is an approximation of $A^u_f$ with the desired precision, in the sense clarified below, and also using suitable unitary operators in accordance with  (\ref{BD222}).

First of all, we construct positive bounded operators out of $\sA^u_{f,\epsilon}(\Delta)$ (which arbitrarily approximate the effects $A_f^u(\Delta)$ as stated in (\ref{LIM}) but are not positive), exploiting the improved local Hamiltonians $\sH^u_{f,\eta}$. In other words we consider the unique everywhere-defined bounded extension of the operator 
$$\frac{1}{\sqrt{H^u_\epsilon}} \sH^u_{f,\eta}(\Delta)\frac{1}{\sqrt{H^u_\epsilon}} :\gS_0\to \gS_0\:.$$
which turns out to have the explicit form \beq\label{LaLa}
\sA^u_{f, \epsilon,\eta}(\Delta) :=   \sA^u_{f,\epsilon}(\Delta) + \eta |u\cdot u_\Sigma| |\Delta| \frac{1}{H^u_\epsilon} \in \gB(\gF_s(\cH_m))
\eeq
for $\epsilon>0$ and the remaining parameters fixed as before.
By construction, for every given arbitrarily small $\eta>0$ we can tune the (temporal part of the) function $f$ in such a way that $\sA^u_{f, \epsilon,\eta}(\Delta) \geq 0$ for all $\Delta \in \cB_b(\Sigma)$ and every $\epsilon>0$. As asserted, these positive operators, when traced against finite-particle states, approximate the effects $A^u_{f}(\Delta)$ with the desired precision in a given laboratory based on $\Delta_0 \in \cB_b(\Sigma)$. \\

\begin{proposition} \label{PF1} Take $u \in \sT_+$, $\Delta_0 \in \cB_b(\Sigma)$, an arbitrarily small $\eta>0$, and $f\in \cD_\bR(\bM)$ such that (\ref{NT}) is true.  The $\rho, \Delta$-uniform bound holds
\beq |tr(\rho  A^u_{f}(\Delta))  - \lim_{\epsilon\to 0^+} tr(\rho \sA^u_{f, \epsilon,\eta}(\Delta)) |  \leq 
 \eta \frac{|u\cdot u_\Sigma | |\Delta_0|}{m}\label{STIMa}\eeq
for every $\rho \in \sS(\cH^{(n)})$ with $n>0$ and every $\Delta \in  \cB(\Delta_0)$.
 \end{proposition}
 
 \begin{proof}
It holds
$| tr(\rho {\sA}^u_{f, \epsilon}(\Delta)) -  tr(\rho {\sA}^u_{f, \epsilon, \eta}(\Delta)) |  \leq 
 \eta |\Delta| |u\cdot u_\Sigma| tr(\rho \frac{1}{H^u_{\epsilon}}) \leq \eta |u\cdot u_\Sigma|  \frac{|\Delta|}{m}$ directly from (\ref{LaLa}).
Taking the limit as $\epsilon\to 0^+$, the bound remains valid and (\ref{LIM}) produces the thesis since $|\Delta| \leq |\Delta_0|$. \end{proof}

At this point, we can prove the following result, where the structure of a conditional POVM as in (\ref{oneG}) emerges.\\

\begin{theorem} Take $u \in \sT_+$, an arbitrarily small $\eta>0$ and $f\in \cD_\bR(\bM)$ such that (\ref{NT}) is true. Define the POVM $\cB(\Delta_0)\ni \Delta \mapsto  B_{\Delta_0}(\Delta)^u_{f,\eta}$ for a given $\Delta_0\in \cB_b(\Sigma), \Sigma$ rest space, as in (a) of Theorem \ref{TUBB}.  For every $\epsilon>0$, there is a unitary $V^u_{f,\epsilon,\eta,\Delta_0}$ 
such that 
\beq\label{QP}
B_{\Delta_0}(\Delta)^u_{f,\eta}  = V^u_{f,\epsilon,\eta,\Delta_0} \frac{1}{\sqrt{\sA^u_{f, \epsilon,\eta}(\Delta_0)}}  \sA^u_{f, \epsilon,\eta}(\Delta)  \frac{1}{\sqrt{\sA^u_{f, \epsilon,\eta}(\Delta_0)}} V^{u\dagger}_{f,\epsilon,\eta,\Delta_0}\:.
\eeq
\end{theorem}

\begin{proof}  First of all, we observe that (\ref{LaLa}) and ${\sH}^{u}_{f,\eta}(\Delta) \geq c^u_{\eta,f}  |u\cdot u_\Sigma| |\Delta| I  \geq 0 $ imply 
$
\sA^u_{f, \epsilon,\eta}(\Delta) \geq \frac{c^u_{\eta,f} |u\cdot u_\Sigma |}{\epsilon}   |\Delta| I \geq 0
$ if $\epsilon>0$ and, in particular, 
 $\sA^u_{f, \epsilon,\eta}(\Delta)^{-1} \in \gB(\gF_s(\cH_m))
$.
Furthermore, since $\overline{{\sH}}^u_{f,\eta}(\Delta_0)$ is selfadjoint and $1/\sqrt{H^u_\epsilon} \in \gB(\gF_s(\cH_m))$ we have
$$\sqrt{\overline{{\sH}}^u_{f,\eta}(\Delta_0)}\frac{1}{\sqrt{H^u_\epsilon}} = \left(\frac{1}{\sqrt{H^u_\epsilon}}  \sqrt{\overline{{\sH}}^u_{f,\eta}(\Delta_0)} \right)^\dagger\:\:\mbox{so that}\:\:
\left( \sqrt{\overline{{\sH}}^u_{f,\eta}(\Delta_0)} \frac{1}{\sqrt{H^u_\epsilon}} \right)^\dagger=
\overline{\frac{1}{\sqrt{H^u_\epsilon}} \sqrt{\overline{{\sH}}^u_{f,\eta}(\Delta_0)} }\:,$$
 where we also used the fact that both sides in the first identity 
are (closed) densely defined operators.
As a consequence, if $\Psi \in \sqrt{H^u_\epsilon}(\gS_0)$ so that $(\sqrt{H^u_\epsilon})^{-1}\Psi \in \gS_0 =  D({\sH}^u_{f,\eta}(\Delta_0))\subset D(\overline{{\sH}}^u_{f,\eta}(\Delta_0)) \subset D(\sqrt{\overline{{\sH}}^u_{f,\eta}(\Delta_0)})$, we find
$$\left\langle \Psi \left| \left( \sqrt{\overline{{\sH}}^u_{f,\eta}(\Delta_0)} \frac{1}{\sqrt{H^u_\epsilon}} \right)^\dagger \left( \sqrt{\overline{{\sH}}^u_{f,\eta}(\Delta_0)} \frac{1}{\sqrt{H^u_\epsilon}} \right)\Psi\right. \right\rangle
=\langle \Psi | \sA^u_{f, \epsilon,\eta}(\Delta_0)\Psi\rangle \leq C||\Psi||^2\:.$$
This implies that $\sqrt{\overline{{\sH}}^u_{f,\eta}(\Delta_0)} \frac{1}{\sqrt{H^u_\epsilon}}$, defined on its natural domain
$D:= \sqrt{H^u_\epsilon}(\gS_0)$,
extends uniquely by continuity to an operator $K\in  \gB(\gF_s(\cH_m))$.
The range of this operator is dense since $\overline{{\sH}}^u_{f,\eta}(\Delta_0)$ is strictly bounded below and (a) of Lemma \ref{F2} holds. Moreover, since $K^\dagger K = \sA^u_{f, \epsilon,\eta}(\Delta_0)$ and the latter operator is strictly bounded below, and $K$ is continuous, its range is closed. In conclusion $K\in \gB(\gF_s(\cH_m))$ is also bijective and thus its inverse is bounded. The polar decomposition theorem and $K^\dagger K = \sA^u_{f, \epsilon,\eta}(\Delta_0)$ eventually imply
that $K= V^u_{f,\epsilon,\eta,\Delta_0}  \sqrt{\sA^u_{f, \epsilon,\eta}(\Delta_0)}$ for a partial isometry $V^u_{f,\epsilon,\eta,\Delta_0}$, which we shall indicate by $V$ to simplify the notation. Actually $V$ is unitary just  because $K$ is bijective.
Finally, since the involved operators are bijective,
$(K\spa\rest_D)^{-1}= \sqrt{H^u_\epsilon} \frac{1}{\sqrt{\overline{{\sH}}^u_{f,\eta}(\Delta_0)}}$ where the natural domain of this composition is just $K(D) = \gS_{00}$ defined in (\ref{BECCA}).
To conclude,
$$\overline{{\sH}}^u_{f,\eta}(\Delta) \frac{1}{\sqrt{\overline{{\sH}}^u_{f,\eta}(\Delta_0)}} =
\overline{{\sH}}^u_{f,\eta}(\Delta) \frac{1}{\sqrt{H^u_\epsilon}}\sqrt{H^u_\epsilon} \frac{1}{\sqrt{\overline{{\sH}}^u_{f,\eta}(\Delta_0)}} 
=\overline{{\sH}}^u_{f,\eta}(\Delta) \frac{1}{\sqrt{H^u_\epsilon}} (K\spa\rest_D)^{-1}$$
$$= \overline{{\sH}}^u_{f,\eta}(\Delta) \frac{1}{\sqrt{H^u_\epsilon}} \frac{1}{\sqrt{\sA^u_{f,\epsilon,\eta}(\Delta_0)}}V^\dagger 
=\sqrt{H^u_\epsilon} \sA^u_{f,\epsilon,\eta}(\Delta)\frac{1}{\sqrt{\sA^u_{f,\epsilon,\eta}(\Delta_0)}}V^\dagger \:.$$
Using this identity on the right-hand side of the composition (\ref{BECCA2}) and adjusting the left-hand side similarly, we find that the two operators appearing on the two sides of (\ref{QP}) have the same matrix elements on the dense subspace $\gS_{00}$. Since both operators are bounded and everywhere defined, the thesis follows.
\end{proof}

\subsection{Local Weyl and von Neumann algebras and Haag duality} \label{secHD} 
We shall consider the Weyl unitaries $W(f) := e^{i\overline{\hat{\phi}(f)}}$ for $f\in \cD_\bR(\bM)$ introduced in Proposition \ref{PROP35}, and we refer to the elementary {\em theory of von Neumann algebras} \cite{Takesaki,L}.  
 
 As usual, if $\cH$ is a Hilbert space, $\gG':=\{ B \in \gB(\cH)\:|\: [A,B]=0,\:\:\forall A\in \gG\}$ henceforth denotes the {\bf commutant} of a set $\gG \subset \gB(\cH)$.\\ 
 
\begin{definition} 
{\em  Let  ${\cal O}\subset \bM$ be an open set. 
\begin{itemize} 
\item[(a)] The {\bf local Weyl algebra} $\cW({\cal O})\subset \gB(\gF_s(\cH_m))$ associated with ${\cal O}$  is the unital $C^*$-algebra generated by the Weyl operators $W(f)$    with $f\in \cD_\bR(\bM)$ and $supp(f)\subset {\cal O}$. 
(In other words, it is  the operator-norm closure in $\gB(\gF_s(\cH_m))$ of linear combinations of products of the aforementioned unitiary operators $W(f)$.) 
\item[(b)] $\gW({\cal O}):= \cW({\cal O})''$ is the {\bf local von Neumann algebra} associated with ${\cal O}$. \hfill $\blacksquare$\\ 
\end{itemize}} 
\end{definition} 
 
\noindent By definition, the {\bf isotony property} holds for both families of algebras: 
\beq 
\cW({\cal O}) \subset \cW({\cal O}_1)\quad \mbox{and}\quad 
\gW({\cal O}) \subset \gW({\cal O}_1)\quad \mbox{if ${\cal O}\subset {\cal O}_1$} \:. \label{ISOT} 
\eeq 
The {\em Weyl relations} (\ref{Weyl}) imply that the linear subspace 
$$span\{W(f) \:|\: f\in \cD_\bR(\bM),\: supp(f)\subset {\cal O}\}\subset \cW({\cal O})$$ actually is a unital $*$-algebra that is dense in $\cW({\cal O})$ in the operator norm. Hence it is also dense in the strong operator topology. At this point, von Neumann's {\em double commutant theorem} \cite{Takesaki,L} implies that 
$$\gW({\cal O}) =  
\overline{span\{W(f) \:|\: f\in \cD_\bR(\bM),\: supp(f)\subset {\cal O}\}}^s$$ 
$$ = \overline{span\{W(f) \:|\: f\in \cD_\bR(\bM),\: supp(f)\subset {\cal O}\}}^w $$  
\beq = (span\{W(f) \:|\: f\in \cD_\bR(\bM),\: supp(f)\subset {\cal O}\})''\label{WO}\:.\eeq  
where the closures $\overline{\cdot}^s$ and $\overline{\cdot}^w$ refer to the strong and weak operator topologies, respectively.

To conclude this short summary, we present a version of {\em Haag duality} in the more modern formulation discussed in \cite{Camassa}.  This is one of the celebrated mathematical relations in AQFT \cite{Haag,Araki}. For the free scalar field, it was first established by Araki in \cite{hd}.\\ 
 
\begin{definition} 
{\em If $p,q \in \bM$ are such that $p-q\neq 0$ is future directed, the {\bf open double cone} generated by them is the set  
${\cal O} := Int(J^-(p) \cap J^+(q))$. \\ 
Referring to the notion of causal complement (\ref{CComp}), we define $A^c := Int(A^\pperp)$ if $A\subset \bM$.  \hfill $\blacksquare$}\\ 
\end{definition} 
 
\noindent Observe that 
 every bounded set in $\bM$ is contained in a sufficiently large open double cone. 
 In particular, a laboratory $L(\Delta_0)$, based on a bounded non-empty open set $\Delta_0\subset \Sigma$, is always contained in a sufficiently large open double cone ${\cal O}$ generated by $p,q \in \bM$ with $p-q$ normal to $\Sigma$ and ${\cal O}\cap \Sigma \supset \Delta_0$. Open double cones themselves are laboratories based on $\Delta_0 :={\cal O}\cap \Sigma$ as above. Open double cones are in particular causally complete, ${\cal O} = ({\cal O}^\pperp)^\pperp$. Finally, the family of open double cones is a topological basis for $\bM$. We leave the elementary proof of these geometric facts to the reader. \\ 
 
\begin{proposition} 
The {\bf Haag duality} relation holds for the von Neumann algebras generated by the Weyl algebra of a real massive Klein-Gordon field in Minkowski spacetime: if ${\cal O}\subset \bM$ is an open double cone, then 
\beq 
\gW({\cal O})' = \gW({\cal O}^c) \:.\label{Hdu} 
\eeq 
\end{proposition} 
\begin{proof}  
See Theorem 4.8 in \cite{Camassa} for $M=\bM$. 
\end{proof} 
\subsection{Commutativity of conditional POVMs as local AHK operators}  
 
We conclude this work, making particular use of Haag duality, by proving that the conditional POVMs with effects $B_{\Delta_0}(\Delta)^u_{f,\eta}$ satisfy the commutativity property expected in the AHK approach.

We remind the reader that a closed densely-defined operator $A$ is said to be {\bf affiliated} with a von Neumann algebra $\gR\subset \gB(\cH)$ \cite{Takesaki,L} if $UA\subset AU$ for every unitary\footnote{That is equivalent to the apparently stronger  requirement $UAU^\dagger =A$ for every $U\in \gR'$.} $U\in \gR'$. 
It turns out, by the very definition of von Neumann algebra and the fact that the unitaries of a von Neumann algebra generate the algebra, that if $A\in \gB(\cH)$, affiliation with $\gR$ is equivalent to $A\in \gR$. It is not difficult to prove that, if $A: D(A) \to \cH$ is selfadjoint, then $A$ is affiliated with $\gR$ if and only if the projectors $P^{(A)}(E)$ of the PVM of $A$ are elements of $\gR$. 
 
We now proceed to prove that the closure of the stress-energy tensor operator is affiliated with every local von Neumann algebra associated with its smearing function. This requires a couple of lemmata.\\

\begin{lemma}\label{LEMNs} If $f,f_1\in \cD(\bM)$, with $f$ real, satisfy $supp(f)\subset {\cal O}$, $supp(f_1)\subset {\cal O}_1$ with the open sets ${\cal O}$ and ${\cal O}_1$ causally separated, then  
\begin{itemize} 
\item[(a)] $W(f) \Psi \in D\left(\overline{:\spa \hat{T}_{\mu\nu}\spa:[f_1]}\right)$ for $\Psi \in \gS_0$; 
    \item[(b)] $W(f) \overline{:\spa \hat{T}_{\mu\nu}\spa:[f_1]} \subset  \overline{:\spa \hat{T}_{\mu\nu}\spa:[f_1]} W(f)$. 
    \end{itemize} 
    The same results are valid if one replaces everywhere $:\spa \hat{T}_{\mu\nu}\spa:[f_1]$ by $:\spa \hat{T}_{\mu\nu}\spa:[f_1]+ cI$ for any $c\in \bC$. 
\end{lemma} 
 
\begin{proof} 
From (a) of Proposition \ref{PROP35}, the vectors in $\gF_0$, and thus also those in the subspace $\gS_0$, are analytic vectors for the selfadjoint operators $\overline{\hat{\phi}(f)}$ when $f$ is real (see Theorem X.41 in \cite{RS2} for a detailed proof); in particular,  
$W(tf) \Psi = \sum_{n=0}^{+\infty} \frac{i^n t^n\hat{\phi}[f]^n}{n!} \Psi$ for every $t\in \bR$.  Now observe that  
$$:\spa \hat{T}_{\mu\nu}\spa:[f_1]  \sum_{n=0}^{N} \frac{i^n \hat{\phi}[f]^n}{n!} \Psi = \sum_{n=0}^{N} \frac{i^n \hat{\phi}[f]^n}{n!} :\spa \hat{T}_{\mu\nu}\spa:[f_1] \Psi$$ 
where we used Proposition \ref{PROP4}. Since $:\spa \hat{T}_{\mu\nu}\spa:[f_1] \Psi\in \gS_0$ because $:\spa \hat{T}_{\mu\nu}\spa:[f_1] (\gS_0) \subset \gS_0$ ((a) of Proposition \ref{PROP1}), the limits of both sides exist as $N\to +\infty$, and this proves the first assertion because $:\spa \hat{T}_{\mu\nu}\spa:[f_1]$ -- which is defined on $\gS_0$ -- is closable. In particular, we also find that 
$\overline{:\spa \hat{T}_{\mu\nu}\spa:[f_1]} W(f)\Psi= W(f) \overline{:\spa \hat{T}_{\mu\nu}\spa:[f_1]}\Psi$ if $\Psi \in \gS_0$. To prove the second assertion, take $\Phi \in D\left(\overline{:\spa \hat{T}_{\mu\nu}\spa:[f_1]}\right)$. There is a sequence $\gS_0 \ni \Psi_n \to \Phi$ such that $:\spa \hat{T}_{\mu\nu}\spa:[f_1] \Psi_n \to \overline{:\spa \hat{T}_{\mu\nu}\spa:[f_1] }\Phi$ by definition of closure, so that  
$W(f) \overline{:\spa \hat{T}_{\mu\nu}\spa:[f_1]} \Psi_n \to W(f) \overline{:\spa \hat{T}_{\mu\nu}\spa:[f_1] }\Phi$. On the other hand,  
$W(f) \overline{:\spa \hat{T}_{\mu\nu}\spa:[f_1]} \Psi_n=  \overline{:\spa \hat{T}_{\mu\nu}\spa:[f_1]} W(f) \Psi_n$, and the limit of both sides exists by construction and thus must coincide  
with $\overline{:\spa \hat{T}_{\mu\nu}\spa:[f_1]} W(f) \Phi$, since the operator under consideration is closed. In summary, if $\Phi \in D\left(\overline{:\spa \hat{T}_{\mu\nu}\spa:[f_1]}\right)$, then  
$\overline{:\spa \hat{T}_{\mu\nu}\spa:[f_1]} W(f)\Phi = W(f) \overline{:\spa \hat{T}_{\mu\nu}\spa:[f_1]} \Phi$. This is (b). The last statement is obvious if one observes that $\overline{A+cI}= \overline{A}+cI$ and $D(A+cI)=D(A)$, for a closable operator $A$ and any constant $c\in \bC$. 
\end{proof}

\begin{lemma} \label{WER} Consider an open double cone ${\cal O}$. 
If $f\in \cD(\bM)$ and $supp(f)\subset {\cal O}$, then $\overline{:\spa \hat{T}_{\mu\nu}\spa:[f] +cI}$ is affiliated with $\gW({\cal O})$ for any $c\in \bC$. 
\end{lemma} 
 
\begin{proof} Let us start with the case $c=0$. 
A densely defined closed operator $A$ is affiliated with a von Neumann algebra $\gR$ if $UA \subset AU$ for every unitary $U\in \gR'$.   
In our case we can exploit Haag duality, $\gW({\cal O})' = \gW({\cal O}^c)$. 
A unitary $U \in \gW({\cal O})' = \gW({\cal O}^c)$ is the strong limit of finite linear combinations of Weyl operators converging to $U$, 
with smearing functions supported in ${\cal O}^c$, due to (\ref{WO}). 
From Lemma \ref{LEMNs}, for $\Psi \in D\left( \overline{:\spa \hat{T}_{\mu\nu}\spa:[f]}\right)$, it holds that 
$W_n\overline{:\spa \hat{T}_{\mu\nu}\spa:[f]}\Psi = \overline{:\spa \hat{T}_{\mu\nu}\spa:[f]}W_n \Psi$. 
Since $\overline{:\spa \hat{T}_{\mu\nu}\spa:[f]}$ is closed and $W_n\to U$ strongly, taking the limit as $n\to+\infty$ we find the desired relation 
$U\overline{:\spa \hat{T}_{\mu\nu}\spa:[f]}\Psi = \overline{:\spa \hat{T}_{\mu\nu}\spa:[f]}U \Psi$. The case $c\neq 0$ is obvious once one observes that $\overline{A+cI}= \overline{A}+cI$ for a closable operator $A$ and any constant $c\in \bC$. 
\end{proof} 
 
It is evident that this affiliation property with local von Neumann Weyl algebras is also valid for $:\spa \hat{\phi}^2\spa:[f]$ and all elements of the local algebra ${\cal T}({\cal O})$ defined in Proposition \ref{PROP4}. In principle, apart from unexpected technicalities, the same argument should apply to all {\em Wick polynomials} defined through the {\em Wick rule}, whether normally ordered or constructed via the locally covariant Hadamard procedure (for the massive scalar field in Minkowski spacetime), including derivatives as well.   
 
 The final result on the commutativity of conditional POVMs is stated in (c) below.

We start with a physically natural definition, which reflects the fact that the identity ${\sH}^u_{f,\eta}(\Delta)= :\spa \hat{T}_{\mu\nu} \spa:[f_\Delta] u^\mu u^\nu_\Sigma  + \eta |u\cdot u_\Sigma| |\Delta| I $ holds, as follows from its definition and from (\ref{PART}),  
where, by construction, $supp(f_\Delta) \subset \Delta+ supp(f^2)$. 
Thus, the standard localization notion for local observables can also be used for ${\sH}^u_{f,\eta}(\Delta)$.\\ 
 
\begin{definition} {\em The operator $\sH^u_{f,\eta}(\Delta)$ defined in (\ref{DEFhhh})  is {\bf localized} in the open double cone ${\cal O}$ if 
$\Delta+ supp(f^2) \subset {\cal O}$.} \hfill $\blacksquare$\\ 
\end{definition} 
 
\begin{theorem} Take $u\in \sT_+$, an arbitrarily small $\eta>0$, and $f\in \cD_\bR(\bM)$ such that (\ref{NT}) is true.  The following facts hold for a rest space $\Sigma$ and an open double cone ${\cal O}\subset \bM$. 
\begin{itemize} 
\item[(a)] If ${\sH}^u_{f,\eta}(\Delta)$ is localized in ${\cal O}$, then its closure $\overline{{\sH}}^u_{f,\eta}(\Delta)$ is affiliated with $\gW({\cal O})$. In particular, the spectral projectors of $\overline{{\sH}}^u_{f,\eta}(\Delta)$ belong to $\gW({\cal O})$. 
\item[(b)]  If   ${\sH}^u_{f,\eta}(\Delta_0)$ is localized in ${\cal O}$, then  the effects of the POVM $\cB(\Delta_0) \ni \Delta \mapsto B_{\Delta_0}(\Delta)^u_{f,\eta}$ defined in (\ref{BECCA2}) satisfy  
 $B_{\Delta_0}(\Delta)^u_{f,\eta} \in \gW({\cal O})$. 
\item[(c)] Let $\tilde{\cal O}$ be another open double cone, take $\tilde{\Delta}_0 \in \cB_b(\tilde{\Sigma})$,  suppose that  
${\sH}^{\tilde{u}}_{\tilde{f},\tilde{\eta}}(\tilde{\Delta}_0)$ is localized in $\tilde{\cal O}$, 
 and consider the analogous POVM $\cB(\tilde{\Delta}_0) \ni \tilde{\Delta} \mapsto {B}_{\tilde{\Delta}_0}(\tilde{\Delta})^{\tilde{u}}_{\tilde{f},\tilde{\eta}}$ for given $\tilde{u},\tilde{f}, \tilde{\eta}$  such that (\ref{NT}) is true. 
If ${\cal O}$ and $\tilde{\cal O}$ are causally separated, then 
\beq 
[ {B}_{\tilde{\Delta}_0}(\tilde{\Delta})^{\tilde{u}}_{\tilde{f},\tilde{\eta}},   B_{\Delta_0}(\Delta)^u_{f,\eta}] =0\:, \quad \mbox{for every $\Delta \in \cB(\Delta_0)$ and $\tilde{\Delta} \in \cB(\tilde{\Delta}_0)$}. 
\eeq 
More generally, this commutativity relation holds after replacing $ {B}_{\tilde{\Delta}_0}(\tilde{\Delta})^{\tilde{u}}_{\tilde{f},\tilde{\eta}}$ with any operator $B\in \gM(\tilde{O})$.
\end{itemize} 
\end{theorem} 
 
\begin{proof} (a) First observe that  
the densely defined symmetric positive operator ${\sH}^u_{f,\eta}$ satisfies 
${\sH}^u_{f,\eta}= :\spa \hat{T}_{\mu\nu} \spa:[f_\Delta] u^\mu u^\nu_\Sigma  + \eta |u\cdot u_\Sigma| |\Delta| I $ as follows from its definition and from (\ref{PART}). By construction, $supp(f_\Delta) \subset \Delta+ supp(f^2)$. At this point Lemma \ref{WER} proves that $\overline{{\sH}}^u_{f,\eta}(\Delta)$ is affiliated with $\gW({\cal O})$ as well and thus, in particular, the spectral projectors of $\overline{{\sH}}^u_{f,\eta}(\Delta)$ belong to $\gW({\cal O})$.\\ 
(b) Since ${\cal O} \supset \Delta_0+ supp(f^2)\supset \Delta+ supp(f^2)$ for $\Delta \subset \Delta_0$, (a) proves that the spectral projectors of all $\overline{{\sH}}^u_{f,\eta}(\Delta)$ and $\overline{{\sH}}^u_{f,\eta}(\Delta_0)$ belong to $\gW({\cal O})$. Referring to Corollary \ref{CORC}, this fact easily implies, via spectral calculus (e.g. Proposition 3.78 in \cite{Moretti2}), that if $U \in \gW({\cal O})'$ and $\Psi \in \gS_{00}$, 
$$U\frac{1}{\sqrt{\overline{{\sH}}^u_{f,\eta}(\Delta_0)}} \overline{{\sH}}^u_{f,\eta}(\Delta)   \frac{1}{\sqrt{\overline{{\sH}}^u_{f,\eta}(\Delta_0)}}\Psi = 
\frac{1}{\sqrt{\overline{{\sH}}^u_{f,\eta}(\Delta_0)}} \overline{{\sH}}^u_{f,\eta}(\Delta)   \frac{1}{\sqrt{\overline{{\sH}}^u_{f,\eta}(\Delta_0)}} U\Psi\:.$$ 
Since the dense subspace $\gS_{00}$ is a core for the everywhere-defined bounded operator $B_{\Delta_0}(\Delta)^u_{f,\eta} =\overline{\frac{1}{\sqrt{\overline{{\sH}}^u_{f,\eta}(\Delta_0)}} \overline{{\sH}}^u_{f,\eta}(\Delta)   \frac{1}{\sqrt{\overline{{\sH}}^u_{f,\eta}(\Delta_0)}}}$, we obtain that $UB_{\Delta_0}(\Delta)^u_{f,\eta}= B_{\Delta_0}(\Delta)^u_{f,\eta}U$ if $U \in \gW({\cal O})'$, so that 
$B_{\Delta_0}(\Delta)^u_{f,\eta}\in  \gW({\cal O})'' =  \gW({\cal O})$. \\ 
(c) If ${\cal O}$ and $\tilde{\cal O}$ are causally separated, then $\tilde{\cal O} \subset {\cal O}^c$ so that $\gW(\tilde{\cal O}) \subset \gW({\cal O}^c) =  \gW({\cal O})'$ by (\ref{Hdu}). The thesis follows from (b). 
\end{proof}

\begin{remark}
 \begin{itemize}
{\em \item[(1)] 
Essential selfadjointness played a crucial role in establishing that the spectral measure of the stress--energy tensor belongs to a corresponding local von Neumann algebra. For other elements $A\in {\cal T}({\cal O})$ (the local algebra defined in Proposition \ref{PROP4}) essential selfadjointness is much more difficult to establish. In general we only know that these elements are symmetric on the dense domain $\gS_0$. 
With an argument similar to the one used for the stress--energy tensor one may prove that $\overline{A}$ is affiliated with  $\gM({\cal O})$, but it is difficult to go further.
Nevertheless, if we also know that $A\geq 0$  (or that $A$ is only bounded below), we can conclude that it admits its Friedrichs selfadjoint extension $A_F$ (observe that $A$ and $\overline{A}$ admit the same Friedrichs extension $A_F$). 
At this point, a crucial result of Kadison applies.\\ 
\begin{proposition}[Kadison] \label{KaD} 
Let $\gR$ be a von Neumann algebra on the complex Hilbert space $\cH$ and $A:D(A) \to \cH$ a  closed,  symmetric, densely-defined operator on $\cH$ which is positive: $\langle x|A x\rangle \geq 0$ if $x\in D(A)$. If $A$ is affiliated with $\gR$, then its Friedrichs selfadjoint extension $A_F$ is affiliated with $\gR$. 
The PVM of $A_F$  therefore belongs to $\gR$. 
\end{proposition} 
\begin{proof} Corollary 5 in \cite{Kadison}. 
\end{proof} 
Using this result, one has that, even if  $A\in {\cal T}({\cal O})$ is not essentially selfadjoint but $A\geq 0$ (or more weakly $A$ is bounded below),  its Friedrichs selfadjoint extension $A_F$ is affiliated with $\gM({\cal O})$. In particular, its spectral measure belongs to $\gM({\cal O})$ together with every (bounded) spectral function of $A_F$.  If $A\geq 0$ is even essentially selfadjoint,  its unique selfadjoint extension  must coincide with $A_F$.
\item[(2)]  As already observed, the fact that $B_{\Delta_0}(\Delta)^u_{f,\eta}\in  \gW({\cal O})$ in addition to positivity  $B_{\Delta_0}(\Delta)^u_{f,\eta} \geq 0$, when we  
use that effect on states in the whole Hilbert space instead of those in $\cH^{(1)}$,  
 immediately implies a {\em dark count} phenomenon -- which is  mathematically expressed by 
$\langle \Omega | B_{\Delta_0}(\Delta)^u_{f,\eta} \Omega \rangle >0$ -- by the general version of the Reeh-Schlieder theorem 
(d) RS2 in Proposition \ref{PROP35}.  (See \cite{FC26} for a recent physical discussion of the subject.) \hfill $\blacksquare$}
\end{itemize}
\end{remark}

\section{Conclusions and outlook}\label{CONC}  
In this work we have constructed a class of relativistic spatial localization observables
$A^u_f$
within the standard framework of quantum field theory, by exploiting the stress--energy--momentum tensor operator smeared with test functions $f^2$. The construction provides, for every fixed timelike direction, a family of POVMs defined on spacelike hypersurfaces, which are well-behaved on each $n$-particle sector and satisfy a natural relativistic causality condition of Castrigiano type.

A central outcome of the analysis is that these localization observables arise from local or quasi-local quantum-field-theoretic quantities, thereby placing previous constructions, originally formulated at a more heuristic level, on a rigorous footing. In particular, in the one-particle sector, the resulting localization scheme reduces to the observable introduced in \cite{M23} (which extended and made rigorous an original physical model due to Terno \cite{Terno}), and its first moment reproduces the Newton--Wigner position operator, independently of the shape of the smearing function under suitable normalization and centering conditions.

An important aspect concerns the analysis of the role of energy positivity. While the stress--energy tensor fails to define positive operators on the full Fock space due to the Reeh--Schlieder theorem, we have shown that suitable lower bounds can be established by means of quantum energy inequalities. This analysis makes it possible to control deviations from positivity and to define regularized families of positive operators $A^u_{f,\epsilon}$ that approximate the localization effects with arbitrary precision on states with a fixed particle number. The price to pay is, of course, that the operators $A^u_{f,\epsilon}$ are non-local; however, they arise as restrictions (more precisely as {\em compressions} to the $n$-particle spaces)  of non-local, bounded from below but  non-positive operators $\sA^u_{f,\epsilon}$ acting on the full Fock space.

The second main result of the paper is the construction of conditional localization observables associated with finite laboratories. Taking energy inequalities into account once again, by introducing suitably modified local energy operators and considering their unique  selfadjoint extensions, we have defined conditional POVMs. These POVMs are related in a natural way to the spatial localization observables constructed above. However, certain intertwining unitary operators $V^u_{f,\epsilon,\eta,\Delta_0}$ appear in the relation between the local effects and the (approximated and positive) relativistic spatial localization effects $\sA^u_{f,\epsilon,\eta}$, as in (\ref{QP}). The role of these unitary operators, although compatible with the analysis developed in \cite{Mor126}, is not yet fully understood and deserves further investigation.

Within this framework, and under appropriate causal separation assumptions on the laboratories, we have shown that the effects of conditional localization observables commute and belong to local von Neumann algebras, in agreement with the Araki--Haag--Kastler description of locality and in accordance with the analysis in \cite{Mor126}. This provides a concrete realization, in a quantum-field-theoretic setting, of the idea that commutativity of localization observables should be recovered only at the level of conditional measurements performed in spacetime regions of finite extent.

Several directions for further investigation naturally emerge from the present analysis. From a mathematical viewpoint, it would be desirable to extend the construction beyond Minkowski spacetime, in particular to globally hyperbolic curved spacetimes and Hadamard states, where the stress--energy tensor and quantum energy inequalities are available in a suitably generalized form. Another relevant issue concerns the dependence on the smearing function and the extent to which different choices lead to physically equivalent localization schemes.

On the physical side, a more detailed analysis of measurement procedures implementing the proposed observables would be of interest, especially in connection with indirect measurement models and detector-based formulations.

Finally, the interplay between localization, energy conditions, and causality constraints suggests that the framework developed here may be useful in addressing more general questions concerning the operational meaning of localization in relativistic quantum systems and its compatibility with the structure of local quantum field theory.

There remain, however, some open issues that deserve further investigation.

Castrigiano (see especially \cite{Castrigiano2,C23}) introduced several relativistic spatial localization observables for bosons and fermions which appear to arise as restrictions to the one-particle space of a more general structure defined on the full Fock space. In that framework, the relevant local observables seem to be given by the electric current rather than the stress--energy tensor. However, it is not clear how a smearing procedure applied to such local observables at the level of the Fock space could reproduce, after spatial integration, the corresponding localization effects. Clarifying this point would be of particular interest and will be addressed elsewhere.

Another related issue concerns the possible existence of connections with the localization observables introduced by Lechner and de Oliveira \cite{LdO}, in which modular theory plays a central role. Understanding whether, and in which sense, the present construction can be related to that approach may shed further light on the structural aspects of localization in quantum field theory.

\section*{Acknowledgments} The author is grateful to D. Castrigiano, C. De Rosa, C. Fewster, N. Pinamonti, and A. Schenkel for many useful discussions, over the years, on the issues addressed in this work. The author is also grateful to an anonymous referee for a careful reading of this quite technical paper, for the content of Remark \ref{REMREF}, and for many helpful suggestions to improve the content of this paper. This work was written within the activities of INdAM-GNFM.

\section*{Declaration statements}
{\bf Conflict of Interest}: The authors declare that they have no conflict of interest.\\
{\bf Ethical Statement}: This work does not involve human participants, animals, or sensitive data, and therefore no ethical approval was required.\\
{\bf Informed Consent}: Not applicable.\\
{\bf Data Availability}: No datasets were generated or analyzed in this study. All relevant information is contained within the article.\\
{\bf Funding}: This research received no external funding.

\section*{Appendix}
\appendix
%
%

\section{Appendix: Properties of normally ordered quadratic forms }

\begin{proposition}\label{PROPaa}
{The normally ordered quadratic forms as in Definition \ref{DEFF} satisfy the following elementary properties.
\begin{itemize}
\item[(1)] Directly from the definition,
\beq \left\langle \Psi'\left| \prod_{j=1}^N a^\dagger_{p_j}\prod_{r=1}^M a_{k_r}\Psi\right. \right\rangle =
\overline{\left\langle \Psi\left|\prod_{r=1}^M a^\dagger_{k_r} \prod_{j=1}^N a_{p_j}\Psi'\right. \right\rangle}
\:.\label{quadAGG}\eeq 
\item[(2)] If $\sigma : \{1,\ldots, N\} \to \{1,\ldots, N\}$ and $\pi : \{1,\ldots, M\} \to \{1,\ldots, M\}$ are arbitrary permutations (i.e., bijective functions):
\beq \left\langle \Psi'\left| \prod_{j=1}^N a^\dagger_{p_j}\prod_{r=1}^M a_{k_r}\Psi\right. \right\rangle =
 \left\langle \Psi'\left| \prod_{j=1}^N a^\dagger_{p_{\sigma(j)}}\prod_{r=1}^M a_{k_{\pi(r)}}\Psi\right. \right\rangle\:.
\eeq 
\item[(3)] Taking (\ref{rep1}) and (\ref{rep}) into account:
for $(\Lambda, a)\in IO(1,3)_+$,
\beq \left\langle U_{(\Lambda, a)}\Psi'\left| \prod_{j=1}^N  a^\dagger_{p_j}\prod_{r=1}^M a_{k_r} U_{(\Lambda, a)}\Psi\right. \right\rangle =
\left\langle \Psi'\left| \prod_{j=1}^N e^{ia\cdot p_j} a^\dagger_{\Lambda^{-1} p_j}\prod_{r=1}^M e^{-ia\cdot k_r}a_{\Lambda^{-1} k_r} \Psi\right. \right\rangle\:.\label{KOV}
\eeq 
\item[(4)] 
 If $\Psi,\Psi' \in \gS_0$, defining $K:= N+M$,  the map \beq\label{fund} \sV_{m,+}^{N+M} \ni (p_1,\ldots, p_N, k_1,\ldots, k_M) \mapsto  \left\langle \Psi'\left| \prod_{j=1}^N a^\dagger_{p_j}\prod_{r=1}^M a_{k_r}\Psi\right. \right\rangle\in \bC \:\:\mbox{is 
  $\cS( \sV_{m,+}^{K} )$}\eeq
\item[(5)] From (\ref{prop}), for $f_j,g_r\in \cS(\sV_{m,+})$, $\Psi,\Psi' \in \gS_0$:
$$
\left\langle \Psi'\left| \prod_{j=1}^N a^\dagger(f_j)\prod_{r=1}^M a( g_r)\Psi\right. \right\rangle  =$$\beq \label{manc}
=  \int \prod_{j=1}^N  f_j(p_j)
\prod_{r=1}^M \overline{g_r(k_r)} \left\langle \Psi'\left| \prod_{j=1}^N a^\dagger_{p_j}\prod_{r=1}^M a_{k_r}\Psi\right. \right\rangle \; d^N\mu_m(p) d^M\mu_m(k)\label{PPRO}
\eeq
\item[(6)] More generally, if $h: \sV_{m,+}^{N+M} \to \bC$ is measurable and polynomially bounded, then the quadratic\footnote{ Theorem X.44  of \cite{RS2} establishes that a quadratic form as in (\ref{Tf}) uniquely defines a closable operator if  $h \in L^2$. However,   we shall need a more general type of quadratic form which does {\em not} satisfy the hypotheses of that theorem.} form on $\gS_0$
$$
\gS_0\times \gS_0 \ni (\Psi',\Psi) \mapsto 
$$
\beq
 \int h(p_1, \ldots , p_N, k_1, \ldots , k_M)  \left\langle \Psi'\left| 
\prod_{j=1}^N a^\dagger_{p_j}\prod_{r=1}^M a_{k_r}\Psi\right. \right\rangle \; d^N\mu_m(p) d^M\mu_m(k)\label{Tf}\:.
\eeq
\item[(7)] By direct inspection and referring to Lemma \ref{lemmaH}, if $v \in \sV$ and $\Psi,\Psi' \in \gS_0$, then
\beq
\langle \Psi | H^v \Psi'\rangle = -\int_{\sV_{m,+}} v\cdot p \langle \Psi| a^\dagger_p a_p \Psi'\rangle d\mu_m(p) \:.  \label{QuadHv}
\eeq
\end{itemize}}
\end{proposition}

\section{Appendix: Proofs of technical propositions}\label{AAA}

\noindent {\bf Proof of Proposition \ref{PROP1}}. First statement and (a). We notice that, if operators  $:\spa\hat{\phi}^2\spa:[f] : \gS_0 \to \gF_s(\cH_m)$,
 $:\spa \hat{T}_{\mu\nu} \spa:[f] : \gS_0 \to \gF_s(\cH_m)$ exist such that (\ref{PP}) and (\ref{TT}) hold, i.e., they define the corresponding quadratic forms on $\gS_0$, then they must be unique, since $\gS_0$ is dense. 
We choose a Minkowskian reference frame, which we shall use henceforth.  We indicate by  $\hat{f} =\hat{f}(k^0, \vec{k})$  the Fourier transform (\ref{Fourier}) of $f\in \cS_\bR(\bR^4)$. We want to associate a corresponding operator with each addend in the decomposition (\ref{DEC}), whose sum amounts to $:\spa \hat{T}_{\mu\nu} \spa:[f]$. We start with the definition of an operator representing 
$$:\spa \hat{T}_{\mu\nu} \spa:[f]_{aa} = (\mbox{formally})=\int \left( \int e^{i (p+k)\cdot x} f(x) d^4x\right)
\frac{ a_{p}a_{k}}{2(2\pi)^3}  t_{\mu\nu}(p,-k)d^2\mu_m(p,k) $$ $$ = \int d^2\mu_m(p,k)\:\frac{ \hat{f}(-E(p)-E(k), -\vec{p}-\vec{k})}{4\pi }
 t_{\mu\nu}(p,-k) a_{p}a_{k} \:. $$
Notice that the smooth function $$\sV_{m,+}^2 \ni (p,k) \mapsto  h_{\mu\nu}(p,k) :=\frac{\hat{f}(E(p)+E(k), \vec{p}+\vec{k})}{4\pi} t_{\mu\nu}(p,-k)\in \bC$$ is a function in $\cS(\sV_{m,+}^2)$, since, for instance, as $\hat{f} \in \cS(\bR^4)$ and $t_{\mu\nu}(p,k)$ is polynomially bounded, $$|h_{\mu\nu}(p,k)| \leq C_N ( (\vec{p}+\vec{k})^2+(E(p)+E(k))^2)^{-N/2}
\leq  C'_N ( (E(p)+E(k))^2)^{-N/2}
\leq C''_N  (\vec{p}^2+\vec{k}^2)^{-N/2}$$ for every $N \in \bN$. Furthermore, the same argument 
 is valid for every derivative, because $\hat{f}\in \cS(\bR^4)$ and the derivatives of the smooth functions $E(p)$, $E(k)$, $t_{\mu\nu}(p,-k)$ are polynomially bounded.
Let us define $:\spa \hat{T}_{\mu\nu} \spa:[f]_{aa}: \gS_0 \to \gS_0$  as the linear extension of the map $\cH^{(n)} \cap \gS_0 \to \cH^{(n-2)} \cap \gS_0$, which vanishes if $n<2$,
 and where we use the notation $d^2\mu_m(p) := d^2\mu_m(p_1,p_2)$
\beq (:\spa \hat{T}_{\mu\nu} \spa:[f]_{aa} \Psi)({k}_1, \ldots , {k}_{n- 2})  := \sqrt{n(n-1)} \int \sp h_{\mu\nu}({p}_1,  {p}_2) \Psi({p}_1,  {p}_2,
{k}_{1}, \ldots , {k}_{n-2}) d^2\mu_m(p)\label{tag2}.\eeq
 This is a well-defined operator $\gS_0 \to \gS_0$ with the property that, if $f \in \cD(\bM)$, a direct use of Fubini--Tonelli theorem yields
$$\langle \Psi' | :\spa \hat{T}_{\mu\nu} \spa:[f]_{aa} \Psi \rangle =
\int \left.\left( \int e^{i (p+k)\cdot x} f(x) d^4x\right)\right|_{\sV_{m,+}^2}
\sp \frac{\langle \Psi'| a_{p}a_{k}\Psi\rangle}{2(2\pi)^3}  t_{\mu\nu}(p,-k)d^2\mu_m(p,k)\:.
 $$
Analogously, referring to the complex-conjugated function $\overline{h}_{\mu\nu} \in \cS(\sV_{m,+}^2)$ we can define
$:\spa \hat{T}_{\mu\nu} \spa:[f]_{a^\dagger   a^\dagger}: \gS_0 \to \gS_0$  as the linear extension of the map $\cH^{(n)}  \cap \gS_0\to \cH^{(n+2)} \cap \gS_0$
\beq ( :\spa \hat{T}_{\mu\nu} \spa:[f]_{a^\dagger a^\dagger} \Psi)({k}_1, \ldots , {k}_{n+2})  := \sqrt{(n+1)(n+2)} S_{n+2} \left( \overline{h_{\mu\nu}(k_1,k_2)} \Psi(
k_{3}, \ldots , k_{n+2}) \right).\eeq
$S_n : \cH^n\to \cH^{(n)} $ is the symmetrization projector.
 With this definition we have, as before, that if $f\in \cD(\bM)$
$$\langle \Psi' | :\spa \hat{T}_{\mu\nu} \spa:[f]_{a^\dagger a^\dagger} \Psi \rangle =
\int \left.\left( \int e^{-i (p+k)\cdot x} f(x) d^4x\right)\right|_{\sV_{m,+}^2}\sp
\frac{\langle \Psi'| a^\dagger_{p}a^\dagger_{k}\Psi\rangle}{2(2\pi)^3}  t_{\mu\nu}(p,-k)d^2\mu_m(p,k)\:.
 $$
We now have to construct the operators associated with the operators $a_p^\dagger a_k$ and $a^\dagger_ka_p$ in the quadratic form of the stress-energy tensor. Let us start by defining the smooth function $$\sV_{m,+}^2 \ni (p,k) \mapsto  s_{\mu\nu}^+(p,k) :=\frac{\hat{f}(E(p)-E(k), \vec{p}-\vec{k})}{4\pi} t_{\mu\nu}(p,k)\in \bC\:,$$ which, contrary to $h$, is polynomially bounded but not necessarily in $\cS(\sV_{m,+}^2)$, due to the minus sign in the temporal entry of $\hat{f}$. Nevertheless the map $(k_1,\ldots, k_n) \mapsto  \int_{\bR^{3}}\spa s^+_{\mu\nu}(k_1,p) \Psi(p,
{k}_{2}, \ldots , \vec{k}_{n}) d\mu_m(p)$ is in $\cS(\sV_{m,+}^l)$ as proved in Lemma \ref{bastard}  below.
The wanted operator is therefore the linear extension of the operator $\cH^{(n)} \cap \gS_0 \to \cH^{(n)} \cap \gS_0$ such that it vanishes for $n=0$ and 
\beq ( :\spa \hat{T}_{\mu\nu} \spa:[f]^+_{a^\dagger a} \Psi)({k}_1, \ldots , {k}_{n})  :=n S_n \left(\int_{\bR^{3}}\sp s^+_{\mu\nu}(k_1,p) \Psi(p,
{k}_{2}, \ldots ,{k}_{n}) d\mu_m(p)\right)\label{tag20} \:.\eeq
This is a well-defined operator $\gS_0 \to \gS_0$ with the property that
$$\langle \Psi' | :\spa \hat{T}_{\mu\nu} \spa:[f]^+_{a^\dagger a} \Psi \rangle =
\int \left.\left( \int e^{i (k-p)\cdot x} f(x) d^4x\right)\right|_{\sV_{m,+}^2}
\sp \frac{\langle \Psi'| a^\dagger_{p}a_{k}\Psi\rangle}{2(2\pi)^3}  t_{\mu\nu}(p,k)d^2\mu_m(p,k)\:.
 $$
With the same procedure, defining
$$\sV_{m,+}^2 \ni (p,k) \mapsto  s^-_{\mu\nu}(p,k) :=\frac{\hat{f}(E(k)-E(p), \vec{k}-\vec{p})}{4\pi} t_{\mu\nu}(p,k)\in \bC\:,$$
and
\beq ( :\spa \hat{T}_{\mu\nu} \spa:[f]^-_{a^\dagger a} \Psi)_{n}({k}_1, \ldots , {k}_{n})  :=n S_n \left( \int_{\bR^{3}}\sp s^-_{\mu\nu}(k_1,p) \Psi(p,
{k}_{2}, \ldots ,{k}_{n}) d\mu_m(p)\right)\label{tagp2} \eeq
we find
$$\langle \Psi' | :\spa \hat{T}_{\mu\nu} \spa:[f]^-_{a^\dagger a} \Psi \rangle =
\int \left.\left( \int e^{-i (k-p)\cdot x} f(x) d^4x\right)\right|_{\sV_{m,+}^2}
\sp \frac{\langle \Psi'| a^\dagger_{k}a_{p}\Psi\rangle}{2(2\pi)^3}  t_{\mu\nu}(p,k)d^2\mu_m(p,k)
 $$
$$=\langle \Psi' | :\spa \hat{T}_{\mu\nu} \spa:[f]^+_{a^\dagger a} \Psi \rangle\:.$$
The wanted operator, which also satisfies (a) by construction, is therefore 
\beq :\spa \hat{T}_{\mu\nu} \spa:[f] :=  :\spa \hat{T}_{\mu\nu} \spa:[f]_{aa}+   :\spa \hat{T}_{\mu\nu} \spa:[f]_{a^\dagger a^\dagger }
+ 2 :\spa \hat{T}_{\mu\nu} \spa:[f]^+_{a^\dagger a}\:.\label{DECTT}
\eeq
The case of $\phi^2$ is essentially identical: it is sufficient to remove the tensor $t_{\mu\nu}$ from each step of the above proof.
The final statement is obvious from the fact that the quadratic form of the stress-energy tensor is a smooth function of $x$ and that (\ref{PP}) and (\ref{TT}) hold.\\
(b)  The identities
$$
 \overline{\langle \Psi' |:\spa \hat{\phi}^2 \spa :(x) \Psi \rangle} =  \langle \Psi |:\spa \hat{\phi}^2 \spa :(x) \Psi' \rangle \:, \quad  \overline{\langle \Psi' |:\spa \hat{T}_{\mu\nu} \spa :(x) \Psi \rangle} =  \langle \Psi |:\spa \hat{T}_{\mu\nu} \spa :(x) \Psi' \rangle$$
and consequently
 easily follow from (\ref{DECPhi}) and (\ref{DEC}), taking (\ref{aggF}) into account. At this point, (\ref{TT}) proves that the associated operator satisfies $:\spa \hat{T}_{\mu\nu} \spa:[\overline{f}] \subset :\spa \hat{T}_{\mu\nu} \spa:[f]^\dagger $ since the domain of $ :\spa \hat{T}_{\mu\nu} \spa:[f]$ is dense and thus the operator $:\spa \hat{T}_{\mu\nu} \spa:[f]$ admits an adjoint. The case of $\phi^2$ is identical.
As an immediate consequence of what has been established, $:\spa \hat{T}_{\mu\nu} \spa:[f]$ and $:\spa \hat{\phi}^2 \spa:[f]$ are symmetric if $f$ is
 real. \\
(c) From (\ref{DEC}), since $t_{\mu\nu}$ is symmetric we have $\langle \Psi' |:\spa \hat{T}_{\mu\nu} \spa:(x) \Psi \rangle =\langle \Psi' |:\spa \hat{T}_{\nu\mu} \spa:(x) \Psi \rangle$ for every $\Psi,\Psi'\in \gS_0$ and $x\in \bM$.
At this point, the wanted identity follows from (\ref{TT}) using density of the domain.\\
(d) Starting from the decomposition (\ref{DEC}), taking (\ref{KOV}) into account, and using $IO(1,3)_+$-invariance of the scalar product and the measure on the mass shell, we find that it holds 
$\langle U^{-1}_g\Psi | :\sp \hat{T}_{\mu\nu}\sp:(x) U_g^\dagger \Psi'\rangle  ={(\Lambda^{-1})_\mu}^\alpha {(\Lambda^{-1})_\nu}^\beta \spa \langle \Psi |  :\sp \hat{T}_{\alpha\beta} \sp:(gx)\Psi' \rangle$ for $x\in \bM\:, \Psi,\Psi' \in \gS_0$.
At this point (\ref{TT}) gives (\ref{covT}) as a consequence of domain density and Poincar\'e invariance of the measure of $\bM$.\\
(e) It is sufficient to prove that
$\partial^\mu \langle \Psi' |:\spa \hat{T}_{\mu\nu} \spa:(x) \Psi \rangle =0$ for $x \in \bM, \Psi,\Psi'\in \gS_0$.
The thesis then immediately arises from (\ref{TT}) by integrating by parts, using the fact that the quadratic form is a bounded function in spacetime. Concerning the above identity, which is actually equivalent to the thesis, it is an immediate consequence of (\ref{DEC}) when passing the derivative under the symbol of integration (which is permitted by the rapid decay of the integrated functions and their derivatives, exploiting the dominated convergence theorem) and noticing that, on shell,
$(p+k)^\mu t_{\mu\nu}(p,-k) =0$ and $(p-k)^\mu t_{\mu\nu}(p,k)=0$. \\
(f) Essential self-adjointness of $:\spa\hat{\phi}^2\spa:[f]$ defined on $\gS_0$, with $f \in \cD(\bM)$ real 
arises (for the massive scalar boson) 
from  \cite{LangerholcSchroer1965}.
In Appendix 1 the authors  prove
that every vector in $\gS_0$ is an analytic
vector for $\alpha:\spa\hat{\phi}^2\spa:[f] + \beta \hat{\phi}[f]$, with  $\alpha,\beta \in \bC$. Hence, by Nelson's
analytic-vector theorem, $:\spa\hat{\phi}^2\spa:[f]$ is essentially self-adjoint on that
domain whenever $f$ is real.\\
Thm.5.2 in \cite{Ko13} proves that $:\spa\hat{T}_{ab} \spa:[f^{ab}]$ is essentially selfadjoint in every globally hyperbolic static spacetime using the unique ground state as reference state. Minkowski spacetime and Minkowski vacuum are  special cases.
$f^{ab}$ is any real valued smooth compactly supported tensor field. In our case $f^{ab}= f u^av^b$.
In \cite{Ko13}, essential selfadjointness is established in $\gG_0$ whereas the domain of our operators is $\gS_0$.
Since $\gG_0 \subset \gS_0$, the selfadjoint closure $\overline{:\spa\hat{T}_{ab} \spa:[f^{ab}] \rest_{\gG_0}}$ is  extended by  the closure $\overline{:\spa\hat{T}_{ab} \spa:[f^{ab}]}$ itself. The latter is symmetric, being the closure of a symmetric operator. Since selfadjoint operators are maximally symmetric, the two closures must coincide. In other words  $\overline{:\spa\hat{T}_{ab} \spa:[f^{ab}]}$ is selfadjoint, which means that the symmetric operator $:\spa\hat{T}_{ab} \spa:[f^{ab}]$ is essentially selfadjoint.\\
This completes the proof. \hfill $\Box$\\

\begin{lemma}\label{bastard}
Let $E(q):= \sqrt{\vec{q}^2+m^2}$ for a fixed constant $m>0$, consider $f \in \cS(\bR^4)$,
and define
$$h(\vec{k},\vec{p}) := \int_{\bR^4} e^{i(\vec{k}-\vec{p})\cdot \vec{x}}  e^{-i(E(k)-E(p))x^0} f(x^0, \vec{x}) d^4x\:, \quad 
T_{\mu\nu}(\vec{k},\vec{p}):= t_{\mu\nu}((E(k), \vec{k}),(E(p),\vec{p}))\:, $$ where 
$$t_{\mu\nu}(p,k) :=\frac{1}{2}[k_\mu p_\nu +  p_\mu k_\nu - g_{\mu\nu}(p\cdot k + m^2)]\:.$$
If $\psi=\psi(\vec{k}, K)$, with $K:=(\vec{k}_1, \ldots, \vec{k}_N )$, is a function in $\cS(\bR^{3(N+1)})$, $N\geq 0$, then the function
$$\phi(\vec{k},K) := \int_{\bR^3} h(\vec{k}, \vec{p}) T_{\mu\nu}(\vec{k}, \vec{p})  \psi(\vec{p}, K) \frac{d^3p}{E(p)}$$
belongs to $\cS(\bR^{3(N+1)})$.
\end{lemma}

\begin{proof} It is easy to see that $\phi$ is smooth as a consequence of the mean value theorem and the dominated convergence theorem when computing the $\vec{k}$ and $K$ derivatives of any order and passing them under the symbol of integration, since these derivatives are locally uniformly bounded by integrable functions of the variable $\vec{p}$ only.   Furthermore, taking advantage of Fubini--Tonelli theorem,
$$\phi(\vec{k},K) =  \int_{\bR^3} \int_{\bR^4} e^{i(\vec{k}-\vec{p})\cdot \vec{x}}  e^{-i(E(k)-E(p))x^0} f(x^0, \vec{x})  {T}_{\mu\nu}(\vec{k}, \vec{p}) \frac{\psi(\vec{p}, K) }{E(p)}d^4x d^3p \:.$$
For any choice of the indices $\mu,\nu$, $\phi$ can be rewritten as 
\beq \label{espa}\phi(\vec{k},K) =  \int_{\bR^3} \int_{\bR^4} e^{i(\vec{k}-\vec{p})\cdot \vec{x}}  e^{-i(E(k)-E(p))x^0} f'(x^0, \vec{x}) \psi'(\vec{p}, K) d^4x d^3p\:,\label{114} \eeq
where the terms of $\hat{T}_{\mu\nu}$ have been embodied in the new Schwartz functions $f'$ and $\psi'$ by integrating by parts in the variables $\vec{x}$ and $x^0$. These functions depend on the choice of the indices $\mu,\nu$. The new function $\psi'$ also embodies the factor $E(p)^{-1}$.\\
At this point the thesis is equivalent to showing that 
 \beq |k^{\alpha} \partial^{|\alpha|}_{k^\alpha}  K^{\beta} \partial^{|\beta|}_{K^\beta} \phi(\vec{k},K)| \leq C_{\alpha\beta} \quad \mbox{for}\quad  (\vec{k},K) \in \bR^{3(N+1)}\label{115}\eeq
for every choice of the multiindices $\alpha, \beta$ and associated constants $ C_{\alpha\beta} <+\infty$. Concerning 
the operators of type  $K^\beta \partial^{|\beta|}_{K^\beta}$, passing the operator under the integral symbol,  their action is simply to change $\psi'$ to a different function in $\cS(\bR^{3(N+1)})$. Therefore we focus attention on the operators $k^{\alpha} \partial^{|\alpha|}_{k^\alpha} $. \\
The action of an operator $ \partial^{|\alpha|}_{k^\alpha} $ on $\phi$ produces a linear combination of functions as in (\ref{espa}) with the following changes in the integrand: (a) factors $x^\beta$,
(b) factors given by products of $k$-derivatives (of order $1$ at least) of the function $E(k)$.  The terms of type (a) can be accommodated into a new definition of the Schwartz function $f'$. We end up with a linear combination of functions 
$$Q(\vec{k})\int_{\bR^3} \int_{\bR^4} e^{i(\vec{k}-\vec{p})\cdot \vec{x}}  e^{-i(E(k)-E(p))x^0} f''(x^0, \vec{x}) \psi'(\vec{p}, K) d^4x d^3p$$
where $Q(\vec{k})$ is a polynomial in $k$-derivatives (of first order at least) of the function $E(k)$. These derivatives, and thus the polynomial itself, are bounded by some $c< +\infty$  as it is easy to prove by direct inspection. The action of the multiplicative operator $k^\alpha$ in the components of $\vec{k}$ can be worked out by integrating by parts, obtaining factors $x^\gamma$ which can be embodied in a new function $f''\in \cS(\bR^4)$ and factors $p^\omega$ which can be included in a new function $\psi''\in \cS(\bR^{3(N+1)})$.
In summary
$$k^{\alpha} \partial^{|\alpha|}_{k^\alpha}  K^{\beta} \partial^{|\beta|}_{K^\beta} \phi(\vec{k},K) =
\sum_{j=1}^J Q_j(\vec{k})\int_{\bR^3} \int_{\bR^4} e^{i(\vec{k}-\vec{p})\cdot \vec{x}}  e^{-i(E(k)-E(p))x^0} f_j(x^0, \vec{x}) \psi_j(\vec{p}, K) d^4xd^3p\:,$$
so that, if $| Q_j(\vec{k})| \leq c_j < +\infty$ as said above,
$$|k^{\alpha} \partial^{|\alpha|}_{k^\alpha}  K^{\beta} \partial^{|\beta|}_{K^\beta} \phi(\vec{k},K)| \leq \sum_{j=1}^J c_j
\int_{\bR^3} \int_{\bR^4} |f_j(x^0, \vec{x}) \psi_j(\vec{p}, K)| d^4xd^3p < +\infty$$
concluding the proof.
\end{proof}

\noindent {\bf Proof of Lemma \ref{lemmad}}. It is sufficient to prove that the matrix elements of the two sides of (\ref{mat})
computed on elements $\Psi_n \in \cH^{(n)} $, $\Psi_l \in \cH^{(l)} $ are identical. This identity amounts to proving that 
$\langle\Psi_n | [a(\kappa_m \overline{f}),a^\dagger (\kappa_m g)] \Psi_l \rangle = \langle\Omega | [a(\kappa_m \overline{f}), a^\dagger (\kappa_m g)]\Omega \rangle \langle\Psi_n |  \Psi_l \rangle$. This identity is true as an easy consequence of the definitions of $a(\psi)$ and $a^\dagger (\psi)$. The last identity is now obvious.
\hfill $\Box$\\

\noindent {\bf Proof of Proposition \ref{PROP2X}}.  According to Propositions \ref{PROP35} and \ref{PROP41}, the map $\cD(\bM) \ni g \mapsto \langle \Psi |:\spa \hat{\phi}[f]\hat{\phi}[g]\spa: \Psi' \rangle = \langle \hat{\phi}[\overline{f}]\Psi | \hat{\phi}[g] \Psi'\rangle -  \langle \hat{\phi}[\overline{f}]\Omega | \hat{\phi}[g] \Omega\rangle 
 \langle \Psi | \Psi'\rangle$ is continuous for every given $f\in \cD(\bM)$. 
With a similar argument,  $\langle \Psi |:\spa \hat{\phi}[f_n]\hat{\phi}[g]\spa: \Psi' \rangle \to 0$ if $f_n \to 0$ in $\cD(\bM)$ for every given $g\in \cD(\bM)$. 
Schwartz' kernel theorem implies that there exists a unique distribution in $\cD'(\bM\times \bM)$ which extends the map $\cD(\bM)\times \cD(\bM) \ni (f,g) \mapsto \langle \Psi| :\spa \hat{\phi}[f]\hat{\phi}[g]\spa: \Psi'\rangle$.
We denote its kernel by 
$\langle \Psi |:\spa \hat{\phi}(x)\hat{\phi}(y)\spa: \Psi' \rangle$. At this point, using the definition (\ref{PHI}), (\ref{mat}), and (\ref{PPRO}) one finds
that the bounded smooth $x$-$y$ symmetric function (as the quadratic forms $(p,k) \mapsto \langle \Psi' | a^\#_{p}a^\#_{k}\Psi \rangle$ are Schwartz functions on $\sV_{m,+}^2$)
$$K(x,y) := \int
\frac{ e^{i (p\cdot x +k\cdot y)}\langle \Psi' | a_{p}a_{k}\Psi\rangle}{2(2\pi)^3} d^2\mu_m(p,k)+ \int  \frac{ e^{-i (p\cdot x + k\cdot y) } \langle \Psi' | a^\dagger _{p}a^\dagger _{k}\Psi\rangle}{2(2\pi)^3} d^2\mu_m(p,k)$$ \beq + \int e^{i (k\cdot y-p\cdot x)} \frac{ \langle \Psi' | a^\dagger _{p}a_{k}\Psi \rangle}{2(2\pi)^3}d^2\mu_m(p,k)  +  \int e^{i (p\cdot x-k\cdot y)} \frac{ \langle \Psi' |a^\dagger _{k}a_{p}\Psi \rangle}{2(2\pi)^3}d^2\mu_m(p,k)\:,\label{DECPhdi}\eeq
satisfies 
$$\int f(x)g(y) K(x,y) d^4xd^4y =\langle \Psi |:\spa \hat{\phi}[f]\hat{\phi}[g]\spa: \Psi' \rangle\:. $$
By uniqueness, $\langle \Psi |:\spa \hat{\phi}(x)\hat{\phi}(y)\spa: \Psi' \rangle = K(x,y)$, which is a smooth bounded function, so that it can be smeared 
with a distribution with compact support such as $f(x) \delta(x,y)$ for $f\in \cD(\bM)$. At this point (\ref{ONE}) immediately arises by comparing the action of this distribution on $ K(x,y)$ and (\ref{DECPhi}). A straightforward modification of this argument yields (\ref{TWO}) from (\ref{DEC})
when using a distribution $D^*_{\mu,\nu}(x,y)f(x) \delta(x,y) =D_{\mu,\nu}(x,y)f(x) \delta(x,y)$. \hfill $\Box$\\

\noindent {\bf Proof of Proposition \ref{PROP4}}.  We start by proving the commutativity for 
two square fields $ :\spa\hat{\phi}^2\spa:[\overline{h}]$.
By direct use of (c3) in Proposition \ref{PROP35}, we have that 
$$[\hat{\phi}[h_1]\hat{\phi}[h_2], \hat{\phi}[h'_1]\hat{\phi}[h'_2]]=0$$ if $supp(h_i)\subset {\cal O}_1$ and  $supp(h'_i)\subset  {\cal O}_2$. Definition (\ref{NORP}) immediately implies that  also $$[:\spa \hat{\phi}[h_1]\hat{\phi}[h_2]\spa:, :\spa\hat{\phi}[h'_1]\hat{\phi}[h'_2]\spa:]=0$$ so that
$$\langle:\spa \hat{\phi}[\overline{h_2}]\hat{\phi}[\overline{h_1}]\spa: \Psi | :\spa\hat{\phi}[h'_1]\hat{\phi}[h'_2]\spa: \Psi' \rangle 
- \overline{\langle  :\spa\hat{\phi}[h_1]\hat{\phi}[h_2]\spa: \Psi' | :\spa \hat{\phi}[\overline{h'_2}]\hat{\phi}[\overline{h'_1}]\spa:  \Psi \rangle}=0$$
  Proposition \ref{PROP2X} finally entails
$$\int ( \langle:\spa \hat{\phi}[\overline{h_2}]\hat{\phi}[\overline{h_1}]\spa: \Psi | :\spa\hat{\phi}(x)\hat{\phi}(y)\spa: \Psi' \rangle
- \overline{\langle  :\spa\hat{\phi}[h_1]\hat{\phi}[h_2]\spa:  \Psi' | :\spa \hat{\phi}(y)\hat{\phi}(x)\spa: \Psi \rangle}) h_1'(x) h'_2(y) d^4xd^4y
 =0$$
A standard argument based on the continuity of the function before  $h_1'(x) h'_2(y)$ in the integrand, 
 arbitrariness of the functions $h'_1,h'_2 \in \cD( {\cal O}_2)$ and the product topology of  $ {\cal O}_2\times  {\cal O}_2$ yields
$$\langle:\spa \hat{\phi}[\overline{h_2}]\hat{\phi}[\overline{h_1}]\spa: \Psi | :\spa\hat{\phi}(x)\hat{\phi}(y)\spa: \Psi' \rangle
- \overline{\langle  :\spa\hat{\phi}[h_1]\hat{\phi}[h_2]\spa:  \Psi' | :\spa \hat{\phi}(y)\hat{\phi}(x)\spa: \Psi \rangle}
 =0\quad \mbox{if $(x,y) \in  {\cal O}_2\times  {\cal O}_2$.}$$ 
As the function on the left-hand side is smooth, smearing both sides with $h'(x)\delta(x,y)$ where $supp(h')\subset {\cal O}_2$ produces
$$\langle:\spa \hat{\phi}[\overline{h_1}]\hat{\phi}[\overline{h_2}]\spa: \Psi | :\spa\hat{\phi}^2\spa:[h'] \Psi' \rangle
-\overline{\langle:\spa \hat{\phi}[h_1]\hat{\phi}[h_2]\spa: \Psi | :\spa\hat{\phi}^2\spa:[\overline{h'}] \Psi' \rangle }
 =0\:.$$
Taking the complex conjugate and repeating the argument for the pair $h_1,h_2$, and taking $supp(h) \subset  {\cal O}_1$, we end up with 
$$\langle \Psi |[ :\spa\hat{\phi}^2\spa:[h] , :\spa\hat{\phi}^2\spa:[h'] ] \Psi' \rangle =0$$
 which is the thesis for $ :\spa\hat{\phi}^2\spa:$ because $\Psi$ ranges over a dense set. The same procedure applies to the proof of commutativity of $ :\spa\hat{\phi}^2\spa:$ and $:\spa \hat{T}_{\mu\nu}\spa:$ and to the case of two operators of type $:\spa \hat{T}_{\mu\nu}\spa:$, just by applying the operators $D_{\mu\nu}(x,y)$ and $D_{\alpha\beta}(x',y')$ at the obvious steps of the proof and taking (\ref{TWO}) into account. The cases involving operators $\hat{\phi}$ are simplified versions of the above proof.   At this point the main proof follows straightforwardly. \hfill $\Box$\\

\noindent {\bf Proof of Proposition \ref{PROPTTR}}.  The proofs of (a) and (c) are straightforward, by applying the relevant definitions and taking advantage of (\ref{TT}) and (e) in Proposition \ref{PROP1}.  The second bound in (b) will be discussed below. The first bound in (b) easily arises by integrating by parts
the explicit expression of ($x$-derivatives of) $\langle \Psi|  :\spa \hat{T}_{\mu\nu} \spa: [f_x]| \Psi' \rangle$, which reads $$ \int_{\bR^6}\sp\sp e^{i (\vec{p}+\vec{k})\cdot \vec{x}} e^{-i (E(p)+E(k))x^0}\hat{f}(-p-k)
\frac{ \langle \Psi | a_{p}a_{k}\Psi'\rangle}{4\pi }  t_{\mu\nu}(p,-k)\frac{d^3pd^3k}{E(p)E(k)} $$ $$ +\int_{\bR^6}\sp\sp  e^{-i (\vec{p}+\vec{k})\cdot \vec{x}} e^{i (E(p)+E(k))x^0} \hat{f}(p+k)\frac{  \langle \Psi | a^\dagger _{p}a^\dagger _{k}\Psi'\rangle}{4\pi}  t_{\mu\nu}(p,-k)\frac{d^3pd^3k}{E(p)E(k)}$$ $$ +\spa \int_{\bR^6} \sp \sp e^{i (\vec{k}-\vec{p})\cdot \vec{x}} e^{i (E(p)-E(k))x^0} \hat{f}(p-k)\frac{ \langle \Psi | a^\dagger _{p}a_{k}\Psi' \rangle}{4\pi} t_{\mu\nu}(p,k)\frac{d^3pd^3k}{E(p)E(k)} $$ \beq +\int_{\bR^6}\sp\sp e^{i (\vec{p}-\vec{k})\cdot \vec{x}} e^{i (E(k)-E(p))x^0}\hat{f}(k-p)\frac{ \langle \Psi |a^\dagger_{k}a_{p}\Psi' \rangle}{4\pi} t_{\mu\nu}(p,k)\frac{d^3pd^3k}{E(p)E(k)} \label{FINT}\:,\eeq
where $p \equiv (E(p), \vec{p})$, $k \equiv (E(k), \vec{k})$ with $E(p) = p^0= \sqrt{\vec{p}^2 +m^2}$ and $E(k) = k^0= \sqrt{\vec{k}^2 +m^2}$.
The integrated functions of $\vec{p}$ and $\vec{k}$ are of Schwartz type in $\bR^3\times \bR^3 \equiv \sV_{m,+}^2$ by construction.  Thus, in particular, one can pass the $x$-derivatives under the symbol of integration. Let us discuss the proof of (d). First of all we observe that the integrals on both sides are finite since the integrated functions are Schwartz, in view of the bound (\ref{BBBO}), which is valid in two Minkowskian coordinate systems comoving with the reference frames $u_\Sigma$ and $u_{\Sigma'}$ respectively.  Define $J^\nu :=  v_\mu\langle \Psi|  :\spa \hat{T}^{\mu\nu} \spa: [f_x] \Psi' \rangle$. We observe that $J$ is smooth and conserved, $\partial_\mu J^\mu(x)=0$, by (c).
Now consider the case where $\Sigma$ and $\Sigma'$ are parallel. In this case we can use a common system of Minkowskian coordinates, where $\Sigma$ is the plane $x^0=0$ and $\Sigma'$ is the analogous plane at $x^0=T>0$ (or {\em vice versa}).
If $C_r \subset \bR^4$ is the cylinder with axis $\vec{x}=0$ and basis $B_r:= \{\vec{x}\in \bR^3 \:|\: |\vec{x}|<r\}$ in $\Sigma$, we can use the divergence theorem for the compact cylinder $G_r \subset C_r$ bounded by $B_r$ and the analogous basis $B'_r$ in $\Sigma'$. 
We move on to compute the flux of $J$ across $\partial G_r$. This flux must be $0$ since $J$ is conserved, taking the divergence theorem into account. The flux of $J$ on the lateral surface of $G_r$ 
is bounded by $c T r^2$,
so that it tends to $0$ as $r\to +\infty$ due to (\ref{BBBO}). The outward fluxes on the bases respectively tend to 
 $-\int_\Sigma v^\mu  \langle \Psi|  :\spa \hat{T}_{\mu\nu} \spa: [f_x] \Psi' \rangle n^\nu_\Sigma d\Sigma$ and
$\int_{\Sigma'} v^\mu  \langle \Psi|  :\spa \hat{T}_{\mu\nu} \spa: [f_x] \Psi' \rangle n^\nu_{\Sigma'} d\mu_{\Sigma'}$ in view of the dominated convergence theorem.   This proves the thesis for the case under consideration.  It remains to consider the case where $\Sigma$ and $\Sigma'$ are not parallel and thus meet in a $2$-plane.
Let us arrange a Minkowskian reference frame with the origin contained in this intersection and $\Sigma$ described by the plane $x^0=0$ in $\bR^4$.
With a further spatial rotation of the coordinates if necessary, we can describe the $3$-plane $\Sigma'$ as the plane $x^0= \gamma x^1$ in $\bR^4$ where $\gamma \in (0,1)$, since $\Sigma'$ is spacelike. We assume that we deal with vectors $\Psi,\Psi' \in \gS_0$ such that their (finite smooth) components $\Psi^n, \Psi'^n$ have compact support in the corresponding spaces $\bR^{3n} \equiv \sV_{m,+}^n$. At this point, using the same argument as in the proof of Lemma 23 in  \cite{CDRM} for the integrals (\ref{FINT}), we have that the second bound in (b) is satisfied:
 $$|\langle \Psi|  :\spa \hat{T}_{\mu\nu} \spa: [f_x]| \Psi' \rangle | \leq \frac{C}{(1+ |\vec{x}|)^{3+\epsilon}}\quad \mbox{if $|\vec{x}|> |x^0|$, for some constants $\epsilon>0$, $C<+\infty$.}$$
We can now apply the divergence theorem to the cylindrical compact solid with a basis $B_r$ on $\Sigma$ as before, lateral surface normal to $\Sigma$, and a basis on $\Sigma'$ made of two parts (one above $\Sigma$ and the other below it). The area of the lateral surface is bounded by $c r^3$ for some constant $c>0$. Due to the found estimate, since this lateral surface stays in the region $|\vec{x}|> |x^0|$ because $0<\gamma <1$,  the flux of $J$ across the lateral surface is bounded by $cr^3\frac{C}{(1+ r)^{3+\epsilon}}$ which vanishes as $r\to +\infty$.
The remaining part of the flux, in the limit as $r\to +\infty$, yields (\ref{INTS}) for the compactly supported vectors $\Psi,\Psi' \in \gD_0$.
 To conclude the proof, we have to remove this compactness requirement. We observe that the space of smooth compactly supported functions $\cD(\bR^{3n})$ is dense in the Schwartz topology of $\cS(\bR^{3n})\equiv \cS(\sV_{m,+}^n)$. Therefore, $\Psi,\Psi'\in \gS_0$ admit sequences $\gD_0 \ni \Psi^n \to \Psi$, $ \gD_0 \ni \Psi'^n \to \Psi'$ (as $n\to +\infty$) such that the convergence, component by component, is in the relevant Schwartz topology.
As a consequence $$e^{i (E(p)+E(k))x^0} \hat{f}(p+k)\frac{  \langle \Psi_n | a^\dagger_{p}a^\dagger_{k}\Psi'_n\rangle}{2(2\pi)^3E(p)E(k)}  t_{\mu\nu}(p,-k)$$ $$ \to e^{i (E(p)+E(k))x^0} \hat{f}(p+k)\frac{  \langle \Psi | a^\dagger_{p}a^\dagger _{k}\Psi'\rangle}{2(2\pi)^3E(p)E(k)}  t_{\mu\nu}(p,-k) \quad \mbox{as $n\to +\infty$}$$
 in the Schwartz topology 
of functions of $(\vec{p},\vec{k}) \in \bR^6$,  and the same result is valid 
for the other similar terms in  (\ref{FINT}).
Using (\ref{FINT}), and the fact that the Fourier transform (of functions of the variables $\vec{u}:=\vec{k}\pm \vec{p}$ and $\vec{v}:=\vec{k}\mp \vec{p}$) is continuous with respect to the Schwartz topology,  we have that, as $n\to +\infty$,
$$\int_\Sigma n^\mu  \langle \Psi^n|  :\spa \hat{T}_{\mu\nu} \spa: [f_x] \Psi'^n \rangle n^\nu_\Sigma d\Sigma(x) \to \int_\Sigma n^\mu  \langle \Psi|  :\spa \hat{T}_{\mu\nu} \spa: [f_x] \Psi' \rangle n^\nu_\Sigma d\Sigma(x)\:,$$
and the same is valid for $\Sigma'$ when referring to a Minkowskian coordinate system adapted to it.  This fact concludes the proof of (\ref{INTS}) for the general case $\Psi,\Psi'\in \gS_0$. 
Concerning (e), the proof immediately arises from (\ref{bff22}), since $u_\mu u'_\nu \langle \Psi | :\spa \hat{T}^{\mu\nu} \spa: [f_x^2] \Psi \rangle
=$
$$
 u_\mu u'_\nu \langle U_{(I,x)}^{-1} \Psi | :\spa \hat{T}^{\mu\nu} \spa: [f^2]  U_{(I,x)}^{-1}\Psi \rangle
 \geq -  u_\mu u'_\nu  b_f^{\mu\nu} || U_{(I,x)}^{-1}\Psi||^2 = -  u_\mu u'_\nu  b_f^{\mu\nu} ||\Psi||^2 \:.
$$
This completes the proof.
\hfill $\Box$\\

\noindent {\bf Proof of Proposition \ref{LLAST}}.
First of all, we observe that
$ \frac{1}{\sqrt{H^u_{\epsilon}}} :\spa \hat{T}_{\mu\nu} \spa: [f^2_x]   \frac{1}{\sqrt{H^u_{\epsilon}}} \Psi$ is well defined if $\Psi \in \gS_0$ because the bounded operators
$\frac{1}{\sqrt{H^u +\epsilon I}}$ are everywhere defined and leave $\gS_0$ invariant ((a) Lemma \ref{lemmaH}).
Further comments on the proof are in order.
Existence of $\sA_{f,\epsilon}^u(\Sigma)$ and (b) are trivial consequences of Proposition \ref{TEO54}, just referring to state vectors of the form 
$ \frac{1}{\sqrt{H^u + \epsilon I}} \Psi $ instead of $\Psi \in \gS_0$, remembering that $ \frac{1}{\sqrt{H^u + \epsilon I}}$ leaves $\gS_0$ invariant. Given that, if  $\Delta\in \cB_0(\Sigma)$ and
$|\Delta|:= +\infty$,
we define  $\sA_{f,\epsilon}^u(\Delta) :=\sA_{f,\epsilon}(\Sigma)- \sA_{f,\epsilon}^u(\Sigma \setminus \Delta)$ provided $\sA_{f,\epsilon}^u(E)$ exists when $|E|<+\infty$, as we shall prove below. Notice that with this definition (\ref{FOND1}) is satisfied in all cases.  
In summary, it is sufficient to prove that $\sA_{f,\epsilon}^u(\Delta)$ does exist when $\Delta$ has finite measure.
Uniqueness follows from  (\ref{FOND1}): $\sA_{f,\epsilon}^u(\Delta)$ is defined on a dense domain. Since it is bounded, it is unique.
 Regarding (d), additivity next follows from (\ref{FOND1}) directly, and weak $\sigma$-additivity then easily follows from the $\sigma$-additivity
of the right-hand side of (\ref{FOND1}). This, in turn, is an obvious consequence of the fact that $\Sigma \ni \vec{x}\mapsto |\langle \Psi|  :\spa \hat{T}_{\mu\nu} \spa: [f^2_x]| \Psi' \rangle | $
is integrable in view of (b)(1) of Proposition \ref{PROPTTR}.
Concerning (a), we observe that $\sA_{f,\epsilon}^u(\Delta)$, if any, is selfadjoint because it is bounded and Hermitian on a dense domain ($\gS_0$). Hermiticity on $\gS_0$ follows from 
(\ref{FOND1}) since the quadratic form in the integral is real for $\Psi=\Psi'$, so that the integral and $\langle \Psi| \sA_{f,\epsilon}^u(\Delta) \Psi\rangle$ itself are real as well. An elementary polarization argument implies that $\sA_{f,\epsilon}^u(\Delta)$ is Hermitian thereon. We shall prove (a1) later.
(c) is a consequence of (d) of Proposition \ref{PROP1}, (\ref{FOND1}), and continuity arguments if $|\Delta|<+\infty$, and it directly follows from (\ref{NORMa}), additivity, the definition of $H^u$ in the remaining 
cases, and elementary spectral calculus.

To go on with the rest of the proof of existence for the case $|\Delta|<+\infty$ we shall make use of a couple of lemmata whose proofs appear below in this appendix.\\

\begin{lemma}\label{LEMMAR} Consider a map $K: \cD_1 \times \cD_2 \to \bC$ which is linear in the right argument and antilinear in the left one, where $\cD_i \subset \cH_i$ is a dense subspace in a complex Hilbert space $\cH_i$, for $i=1,2$. Assume that $|K(\psi_1,\psi_2)| \leq C||\psi_2||_2 ||\psi_1||_1$ for some $C< +\infty$ and every pair  $(\psi_1,\psi_2) \in   \cD_1\times \cD_2$. Then there exists a unique bounded operator $A_K : \cH_2 \to \cH_1$ such that $\langle \psi_1 |A_K \psi_2 \rangle_1 = K(\psi_1,\psi_2)$ for every pair $(\psi_1,\psi_2) \in   \cD_1\times \cD_2$. Finally $||A_K|| \leq C$.
\end{lemma}
\begin{proof} See below.
\end{proof}
\begin{lemma} \label{LEMMAS}
Suppose that the direct Hilbert decomposition holds $\cH = \oplus_{j\in J} \cH_j$ for a complex Hilbert space and mutually orthogonal closed subspaces $\cH_j$. Assume that there are bounded operators $A_j : \cH_j \to \cH$ for $j\in J$ such that $Ran(A_j) \perp Ran(A_k)$ if $j\neq k$ and 
$\sup_{j \in J}||A_j|| <+\infty$. Then there exists a unique $A \in \gB(\cH)$ such that $A\spa\rest_{\cH_j}= A_j$ for $j\in J$. Furthermore 
$||A||= \sup_{j \in J}||A_j||$. If $J$ is countable, then $A\psi = \lim_{N\to +\infty} \sum_{j\leq N} A_j P_j\psi$, where $P_j: \cH \to \cH$ is the orthogonal projector onto $\cH_j$ and $\psi \in \cH$.
\end{lemma}
\begin{proof} See below.
\end{proof}
\noindent Let us move on to prove the main thesis. Referring to a Minkowskian coordinate system $x^0,x^1,x^2,x^3$ where $\Sigma$ corresponds to $x^0=0$, we decompose the quadratic form
\beq\int_{\Delta}\sp u^\mu \left\langle \frac{1}{\sqrt{H^u + \epsilon I}} \Psi \left|  :\spa \hat{T}_{\mu\nu} \spa: [f^2_x] \right.  \frac{1}{\sqrt{H^u + \epsilon I}} \Psi' \right \rangle u^\nu_\Sigma \:\: d\Sigma(x) = \sum_{i=1}^3 T^\Delta_{\epsilon,i}(\Psi,\Psi')\label{expand}\eeq
according to the decomposition (\ref{FINT}), where  we collected together the last two terms 
and 
we omitted the indices $f,u$ which are supposed to be given. The quadratic forms $T^\Delta_{\epsilon,i}(\Psi,\Psi')$ are well defined for the generic case $\Delta \in \cB(\Sigma)$, not only for $\Delta \in \cB_0(\Sigma)$. Furthermore, they are also well defined if $\epsilon = 0$ provided  $\Psi,\Psi' \in \cH^{(0)\perp}$, and interpreting 
$\frac{1}{\sqrt{H^u}}$ in the spectral sense in $\cH^{(0)\perp} $ and as the zero operator in $\cH^{(0)} $.
At least in the case $\Delta \in \cB_0(\Sigma)$ and $\epsilon>0$, we expect to find a decomposition corresponding to (\ref{expand})  of the  operator $\sA_{f,\epsilon}^u(\Delta)$ we are looking for,  where, again, we omit the indices $f,u$ which are supposed to be given,
\beq
\sA_{f,\epsilon}^u(\Delta) = \sum_{i=1}^3 \sA_{\epsilon, i}(\Delta)\:.\label{expand2}
\eeq
We start by focusing on the case $i=3$ where, as we shall see, the operator $\sA_{\epsilon,3}(\Delta)$ turns out to be defined also for $\epsilon=0$ and $\Delta \in \cB(\Sigma)$, possibly unbounded with unbounded complement.
Now $T^\Delta_{\epsilon, 3}(\Psi,\Psi'):= T^\Delta_{\epsilon, 3}(\Psi,\Psi')_+ + T^\Delta_{\epsilon,3}(\Psi,\Psi')_-$ where 
$$ T^\Delta_{\epsilon,3}(\Psi,\Psi')_+:=  \int_\Delta \int_{\bR^6}\sp  e^{i (\vec{k}-\vec{p})\cdot \vec{x}}  \widehat{f^2}(p-k)\frac{ \langle W^\epsilon_u\Psi | a^\dagger_{p}a_{k}W^\epsilon_u\Psi' \rangle}{4\pi} u^\mu t_{\mu 0}(p,k)\frac{d^3pd^3k}{E(p)E(k)} d^3 x\:,$$ \beq  T^\Delta_{\epsilon, 3}(\Psi,\Psi')_- :=  \int_\Delta\int_{\bR^6} 
\sp\sp e^{i (\vec{p}-\vec{k})\cdot \vec{x}}\hat{f^2}(-p+k)\frac{ \langle W^\epsilon_u\Psi |a^\dagger_{k}a_{p}W^\epsilon_u\Psi' \rangle}{4\pi}u^\mu t_{\mu0}(p,k)\frac{d^3pd^3k}{E(p)E(k)} d^3x \nonumber\:,\eeq
and where $W^\epsilon_u := (H^u +  \epsilon I)^{-1/2} \in \gB(\gF_s(\cH_m))$.  Evidently $T^\Delta_{\epsilon,3}(\Psi,\Psi')_+= T^\Delta_{\epsilon,3}(\Psi,\Psi')_-$.
The quadratic form vanishes unless both vectors belong to the same subspace $\Psi,\Psi' \in \cH^{(n)}\cap \gS_0$. We henceforth restrict ourselves to the study of the quadratic form on $(\cH ^{(n)}\cap \gS_0)\times (\cH^{(n)}\cap \gS_0)$ and we want to prove that it is induced by an operator $\sA_{\epsilon, 3,n}(\Delta)\in \gB(\cH^{(n)})$, taking advantage of Lemma \ref{LEMMAR}. If $n=0$ there is nothing to prove because 
the quadratic form vanishes: $\sA_{\epsilon, 3,0}(\Delta)=0$. We therefore assume $n>0$ so that $\Psi = \Psi(p,Q)$ and $\Psi' = \Psi'(k,Q)$ are functions in $\cS(\sV_{m,+}^n)$ and $Q:= (q^2,\ldots, q^{n})$.   $T^\Delta_3(\Psi,\Psi')$ can be rearranged to 
$$\frac{n}{2\pi} \int_\Delta \int_{\bR^6}  \int_{\bR^{3(n-1)}} \sp  e^{i (\vec{k}-\vec{p})\cdot \vec{x}}  \widehat{f^2}(p-k)\frac{\overline{\Psi(p, Q)} \Psi'(k,Q)}{\sqrt{{E^n_{\epsilon, u}(p,Q)} E^n_{\epsilon,u}(k,Q)}} u^\mu t_{\mu 0}(p,k)\frac{d^3pd^3k d^{3(n-1)}q}{E(p)E(k)\prod_i E(q_i)} d^3 x$$
where, for $\epsilon\geq 0$,
$$E^n_{\epsilon,u}(p,Q) := -u\cdot \left(p + \sum_{r=2}^n q_r\right) +  \epsilon \:.$$
Using the Hilbert space isomorphism $$ J_n : L^2(\sV^{n}_{m,+}, d^n\mu_m) \ni \Psi (p_1,\ldots p_n) \mapsto 
 \Phi (\vec{p}_1,\ldots \vec{p}_n)$$ \beq  := \frac{\Psi (p_1,\ldots p_n)}{ \sqrt{E(p_1) \cdots E(p_n)} } \in L^2(\bR^{3n}, d^{3n}p) \label{UM}\eeq
and swapping two integrations of Schwartz functions, the written integral can be rephrased as (where $\vec{Q}:= (\vec{q}_2, \ldots, \vec{q}_n)$)  $T^\Delta_{\epsilon,3}(\Psi,\Psi')= S_{\epsilon,3}(\Phi,\Phi')$
\beq  :=\frac{n}{2\pi} \int_\Delta  \int_{\bR^{3(n-1)}}\sp  \int_{\bR^6} e^{i (\vec{k}-\vec{p})\cdot \vec{x}}  \widehat{f^2}(p-k)\frac{\overline{\Phi(\vec{p}, \vec{Q})} \Phi'(\vec{k},\vec{Q})}{\sqrt{{E^n_{\epsilon,u}(p,Q)} E^n_{\epsilon,u}(k,Q)}} u^\mu t_{\mu 0}(p,k)\frac{d^3pd^3k d^{3(n-1)}q}{\sqrt{E(p)E(k)}} d^3 x\:. \label{NEW}\eeq
Finally, the convolution theorem and the fact that $f$ is real imply that 
$$(2\pi)^2 \widehat{f^2}(p-k) = \int_{\bR^4} \hat{f}(p-k -u) \hat{f}(u) d^4 u= \int_{\bR^4} \hat{f}(p-u) \hat{f}(u-k) d^4 u=
\int_{\bR^4}\overline{\hat{f}(u-p)}\hat{f}(u-k) d^4u\:.$$
According to this result, if $\Psi,\Psi' \in \cH^{(n)} \cap \gS_0$ and $\Phi,\Phi'$ are respectively associated by means of the unitary map (\ref{UM}), then $T^\Delta_{\epsilon,3}(\Psi,\Psi')= S_{\epsilon, 3}(\Phi,\Phi')$ can be rearranged to
\beq T^\Delta_{\epsilon,3}(\Psi,\Psi') = n\int_\Delta \sp\int_{\bR^{3(n-1)}} \int_{\bR^6} \int_{\bR^4} \sp \frac{\overline{\hat{f}(u-p)\Phi(\vec{p}, \vec{Q})} e^{i(\vec{k}-\vec{p})\cdot \vec{x}} \hat{f}(u-k)\Phi'(\vec{k},\vec{Q})}{(2\pi)^3\sqrt{{E^n_{\epsilon,u}(p,Q)} E^n_{\epsilon,u}(k,Q)}}\nonumber \\
\times \frac{u^\mu t_{\mu 0}(p,k)}{\sqrt{E(p)E(k)}}d^4u d^3pd^3k d^{3(n-1)}q d^3 x\:.\label{OST}\eeq
We now consider the reduced quadratic form
$$S'_3(\Psi,\Psi') := n\int_\Delta \sp  \int_{\bR^{3(n-1)}}  \int_{\bR^6}\int_{\bR^4} \sp \frac{\overline{\hat{f}(u-p)\Phi(\vec{p}, \vec{Q})} e^{i(\vec{k}-\vec{p})\cdot \vec{x}} \hat{f}(u-k)\Phi'(\vec{k},\vec{Q})}{(2\pi)^3}d^4u d^3pd^3k d^{3(n-1)}q d^3 x\:.$$
We observe that the internal integrals in $d^4u$ and $d^3pd^3k$ can be swapped by Fubini's theorem because
$$ \int_{\bR^6}\int_{\bR^4} \sp |\hat{f}(u-p)\Phi(\vec{p}, \vec{Q}) \hat{f}(u-k)\Phi'(\vec{k},\vec{Q})|d^4u d^3pd^3k $$
$$=  \int_{\bR^6}  |\Phi(\vec{p}, \vec{Q})\Phi'(\vec{k},\vec{Q})|\int_{\bR^4} \sp |\hat{f}(u+k-p) \hat{f}(u)| d^4u d^3pd^3k$$
$$\leq (\sup |\hat{f}|) \int_{\bR^6}  |\Phi(\vec{p}, \vec{Q})\Phi'(\vec{k},\vec{Q})| d^3pd^3k  \int_{\bR^4} \sp | \hat{f}(u)| d^4u  <+\infty\:.$$
The final formula for $S'_3(\Psi,\Psi')$ is
$$  n\int_\Delta \sp  \int_{\bR^{3(n-1)}} \int_{\bR^4} \int_{\bR^6} \sp \frac{\overline{\hat{f}(u-p)\Phi(\vec{p}, \vec{Q})} e^{i(\vec{k}-\vec{p})\cdot \vec{x}} \hat{f}(u-k)\Phi'(\vec{k},\vec{Q})}{(2\pi)^3} d^3pd^3k d^4u d^{3(n-1)}q d^3 x\:.$$
We can interpret this formula in the Hilbert space $L^2(\bR^3 \times \bR^4 \times \bR^{3(n-1)}, d^3p  d^4ud^{3(n-1)}q)$ as follows. 
If defining the function $F_\Phi := F_\Phi(\vec{p}, u, \vec{Q}) := \hat{f}(u^0 - E(p),\vec{u}-\vec{p})\Phi(\vec{p}, \vec{Q})$, we have that, for $s>0$
$$\int |F_\Phi|^s d^4ud^3pd^{3(n-1)}q= \int  |\hat{f}(u^0 - E(p),\vec{u}-\vec{p})|^s|\Phi(\vec{p}, \vec{Q})|^s d^3p d^4ud^{3(n-1)}q$$
$$= \int \left( \int  |\hat{f}(u^0 - E(p),\vec{u}-\vec{p})|^s d^4u\right) |\Phi(\vec{p}, \vec{Q})|^s d^3pd^{3(n-1)}q $$
\beq =  \left( \int  |\hat{f}(u^0 ,\vec{u})|^s d^4u\right) \int  |\Phi(\vec{p}, \vec{Q})|^s d^3pd^{3(n-1)}q < +\infty \:.\label{stim}\eeq
First of all, (\ref{stim}) implies that $F_\Phi, F_{\Phi'} \in L^s(\bR^3 \times \bR^4 \times \bR^{3(n-1)}, d^3p  d^4ud^{3(n-1)}q)$ if $s> 0$ and, in particular, for $s=2$, $||F_\Phi||^2 = ||\hat{f}||^2 ||\Phi||^2 $, where the norms are those of the relevant $L^2$-Hilbert spaces.
At this point, consider the PVM $$\cB(\bR^3)\ni \Delta \mapsto P(\Delta) \otimes I \otimes I \in \gB(L^2(\bR^3 \times \bR^4 \times \bR^{3(n-1)}, d^3p  d^4ud^{3(n-1)}q))$$
where $P(\Delta)$ is the joint PVM of the three standard position operators in $L^2(\bR^3, d^3x)$ in momentum representation (see, e.g., \cite{Moretti2}). Notice that this reasoning is valid even if $|\Delta| =+\infty$.
 Using the fact that the Fourier-Plancherel unitary operator in $L^2$ is the standard Fourier transform in $L^1$ and 
$\bR^3 \ni \vec{p} \mapsto F_\Phi(\vec{p}, u,\vec{Q})$ is Schwartz for every choice of $(u,\vec{Q})$,
$$S'_3(\Phi,\Phi')  =
 n \int_\Delta d^3x \sp  \int_{\bR^{3(n-1)}} \sp\sp\sp \sp  d^{3(n-1)}q \int_{\bR^4} d^4u \overline{\int_{\bR^3} \sp \frac{e^{i\vec{p}\cdot \vec{x}}}{(2\pi)^{3/2}}F_\Phi(\vec{p}, u,\vec{Q})d^3p}  
\int_{\bR^3} \sp \frac{e^{i\vec{k}\cdot \vec{x}}}{(2\pi)^{3/2}}F_{\Phi'}(\vec{k}, u,\vec{Q})d^3k \:.$$
$$ = n \langle F_\Phi| (P(\Delta) \otimes I\otimes I)  F_{\Phi'} \rangle\:.$$
 As a consequence,  
$$|S'_3(\Phi,\Phi')| \leq n || P(\Delta) \otimes I\otimes I||\, ||F_{\Phi}||\, ||F_{\Phi'}|| \leq n \left(\int_{\bM} f^2 d^4x\right) ||\Phi||\, ||\Phi'|| $$
where we also used the Plancherel theorem.
To conclude, we observe that if $\Psi,\Psi'\in \gS_0\cap \cH^{(n)} $,
$$T^\Delta_{\epsilon, 3}(\Psi,\Psi') = S_{\epsilon, 3}(\Phi,\Phi')= \sum_{a=0}^N S'_3(M_{\epsilon, a}\Phi, M'_{\epsilon, a}\Phi')$$
where, for some (of first order at most and $n$-independent) polynomials $P_a(p)$ in the components of $p\equiv (E(p), \vec{p})$, constructed by decomposing $u^\mu t_{\mu\nu}(p,k)u_\Sigma^\nu$,
$$(M_{\epsilon,a} \Phi)(\vec{p}, \vec{Q}) =\frac{P_a(p)}{\sqrt{ E^n_{\epsilon, u}(p,Q) E(p)}} \Phi(\vec{p}, \vec{Q})\:.$$
It holds
$$||M_{\epsilon,a}|| = \sup_{(\vec{p}, \vec{Q})\in \bR^{3n}} \left| \frac{P_a(p)}{\sqrt{ E^n_{\epsilon,u}(p,Q) E(p)}} \right| \leq n^{-1/2} K_a\:,$$
for some constant $ K_a\geq 0$ independent of $\epsilon$ and $n$ (but depending on $u, u_\Sigma$).
Indeed, change coordinates in order that the temporal axis coincides with $u$. Observe that the components of $p$ are however bounded by $p^0$ (referred to the new basis), and  $(u_\Sigma^0 - \frac{\vec{p}\cdot\vec{u}_\Sigma}{p^0})^2$ is bounded below when $p$ varies in $\sV_{m,+}$ and $u_\Sigma \in \sT_+$ is given\footnote{In fact, $(u_\Sigma^0 - \frac{\vec{p}\cdot\vec{u}_\Sigma}{p^0})^2= \left(\cosh\chi - \frac{\vec{p} \cdot \vec{e}}{\sqrt{\vec{p}^2 + m^2}} \sinh \chi \right)^2$ where $u_\Sigma \equiv (\cosh \chi, \sinh \chi \vec{e})$ for $\chi \geq 0$ given and some unit $\vec{e}$ $g$-normal to $u\equiv (1,\vec{0})$ and where $\vec{p}\in \bR^3$. Varying $\vec{p}$, the function is bounded below by $\cosh \chi -\sinh \chi >0$.}. Putting together these facts we have
$$ \left| \frac{P_a(p)}{\sqrt{ E^n_{\epsilon, u}(p,Q) E(p)}} \right| ^4 \leq  \frac{(p^0)^4 C_a}{(p^0 + q^0_2 +\cdots + q^0_n)^2}
\frac{1}{(u_\Sigma^0p^0 - \vec{p}\cdot\vec{u}_\Sigma)^2} $$ $$\leq 
 \frac{(p^0)^2 C_a}{(p^0 + q^0_2 +\cdots + q^0_n)^2} 
\frac{1}{(u_\Sigma^0 - \frac{\vec{p}\cdot\vec{u}_\Sigma}{p^0})^2} \leq \frac{(p^0)^2C_a}{(p^0 + q^0_2 +\cdots + q^0_n)^2}  \:.$$
The latter function is defined for $p^0, q^0_2, \ldots, q_n^0 \geq m$ and reaches its maximum value
$C_an^{-2}$
 on this boundary as follows by direct inspection.
Putting all together, if $\Psi,\Psi'\in \gS_0\cap \cH^{(n)} $, the $\epsilon, n$-uniform bound holds
$$|S_{\epsilon,3}(\Phi,\Phi')|\leq  \sum_a K_a \left(\int_{\bM} f^2 d^4x\right) ||\Phi|| ||\Phi'||\:.$$
We end up with the estimate, for some $C< +\infty$ independent of $\epsilon\geq 0$, $n\in \bN$, $f\in \cD_\bR(\bM)$, and $\Delta$ (but depending on $u,u_\Sigma$),
$$|T^\Delta_{\epsilon, 3}(\Psi,\Psi')|\leq   \left(\int_{\bM} f^2 d^4x\right)  C ||\Psi|| ||\Psi'|| \quad \mbox{if}\:\: \Psi,\Psi' \in \cH^{(n)}  \cap \gS_0\:.$$
In view of Lemma \ref{LEMMAR} we have that, if $n=0,1,\ldots$, there is a unique $\sA_{\epsilon, 3,n}(\Delta) \in \gB(\cH^{(n)})$ such that $\langle \Psi| \sA_{\epsilon, 3,n}(\Delta) \Psi' \rangle = T^\Delta_{\epsilon, 3}(\Psi,\Psi')$ for $\Psi, \Psi' \in \cH^{(n)} \cap \gS_0$.  
 $\sA_{\epsilon, 3,0}$ is the zero operator. These maps are defined for $\epsilon\geq 0$ and also if $|\Delta|=+\infty$.
We are now allowed to view the found operators as bounded linear maps
$\sA_{\epsilon, 3,n}(\Delta) : \cH^{(n)} \to \gF_s(\cH_m)$, so that $\cH^{(n)}$ is an invariant space, and this is consistent with the identity $\langle \Psi| \sA_{\epsilon,3,n}(\Delta) \Psi' \rangle = T^\Delta_{\epsilon,3}(\Psi,\Psi')$ for $\Psi'\in \cH^{(n)} \cap \gS_0$ and $\Psi\in \gS_0$, because the quadratic form 
only sees the $n$-th component.
Since the norms of the operators $\sA_{\epsilon,3,n}(\Delta)$ are $\epsilon,n$-uniformly bounded, the spaces $\cH^{(n)}$ orthogonally decompose $\gF_s(\cH_m)$, and the images of operators $\sA_{\epsilon,3,n}(\Delta)$ are mutually orthogonal, Lemma \ref{LEMMAS} entails that, for every $\epsilon\geq 0$ and $\Delta \in \cB(\Sigma)$, there is a unique $\sA_{\epsilon,3}(\Delta) \in \gB(\gF_s(\cH_m))$ which extends these operators. By construction 
\beq \langle \Psi| \sA_{\epsilon,3}(\Delta) \Psi' \rangle = T^\Delta_{\epsilon,3}(\Psi,\Psi') \quad  \mbox{$\Psi, \Psi' \in \gS_0$}, \label{TONE}\eeq
with $T^\Delta_{0,3}(\Psi,\Psi')$ defined as discussed under (\ref{expand}).
In particular, for $\Delta \in \cB(\Sigma)$ and $\epsilon\geq 0$,
\beq |\langle \Psi| \sA_{\epsilon,3}(\Delta) \Psi' \rangle |\leq   \left(\int_{\bM} f^2 d^4x\right)  C ||\Psi|| ||\Psi'|| \quad \mbox{if}\:\: \Psi,\Psi' \in \gF_s(\cH_m) \label{TONE3}\:.\eeq
Equivalently, there is $C<+\infty$ independent of $f$ such that
$$||\sA_{\epsilon,3}(\Delta) || \leq C  \left(\int_{\bM} f^2 d^4x\right) \quad \mbox{for every $\epsilon \geq 0$, $\Delta \in \cB(\Sigma)$.}$$
This inequality proves (a1), when also assuming $\epsilon>0$ and $|\Delta|<+\infty$ to have a well-defined $\sA^u_{f,\epsilon}(\Delta)$,
since $P_n \sA^u_{f,\epsilon}(\Delta)P_n= P_n \sA_{\epsilon,3}(\Delta)P_n$. The result trivially extends to the case $|\Sigma \setminus \Delta|<+\infty$. 

We consider the further term $ T^\Delta_{\epsilon,2}$ in the expansion (\ref{expand}), and the corresponding $\sA_{\epsilon,2}(\Delta)$ in (\ref{expand2}) where we explicitly assume that $\Delta$ has finite measure and $\epsilon>0$.\beq  T^\Delta_{\epsilon, 2}(\Psi,\Psi') =  \int_\Delta\int_{\bR^6}\sp\sp e^{i (\vec{p}+\vec{k})\cdot \vec{x}}\widehat{f^2}(-p-k)\frac{ \langle W^\epsilon_u\Psi |a_{p}a_{k}W^\epsilon_u\Psi' \rangle}{4\pi} u^\mu t_{\mu0}(p,-k)\frac{d^3pd^3k}{E(p)E(k)} d^3x \nonumber\:,\eeq
 where, as before, $W^\epsilon_u= (H^u +  \epsilon I)^{-1/2} \in \gB(\gF_s(\cH_m))$.  
 In this case the quadratic form vanishes unless $\Psi' \in \cH^{(n+2)} \cap \gS_0$ and $\Psi \in \cH^{(n)} \cap \gS_0$, therefore we will assume this henceforth. With the same procedure as above, taking advantage of the isomorphisms of Hilbert spaces $J_n$ defined in (\ref{UM}),
 $T^\Delta_{\epsilon,2}(\Psi,\Psi')= S_{\epsilon,2}(\Phi,\Phi')$ can be rearranged to
\beq C_n  \int_{\bR^6} \int_{\bR^{3n}}\sp
\frac{\overline{\Phi(\vec{Q})}}{\sqrt{E^n_{\epsilon,u}(Q)}}
 \frac{{ \left(\int_\Delta e^{i(\vec{p}+\vec{k})\cdot \vec{x}} d^3x\right) \widehat{f^2}(-p-k)u^\mu t_{\mu 0}(p,-k) \Phi'(\vec{p},\vec{k} , \vec{Q})}}{\sqrt{{E^{n+2}_{\epsilon,u}(p,k, Q)} E(p)E(k) }}
 d^3pd^3k d^{3n}q \:, \label{MO}\eeq
where $E^n_{\epsilon,u}(Q):= - \sum_{j=1}^n u\cdot q_j+ \epsilon$ and
$C_n :=(4\pi)^{-1} \sqrt{(n+1)(n+2)}$ and we have also interchanged the integral in $x$ with the others since $|\Delta|<+\infty$ and in the remaining variables the functions are of Schwartz type, including $(\vec{p}, \vec{k}) \mapsto  \hat{f^2}(-p-k)$ as already observed in the proof of Proposition \ref{PROP1}.
Defining the function $H_\Delta(\vec{p},\vec{k}):= \left(\int_\Delta e^{i(\vec{p}+\vec{k})\cdot \vec{x}} d^3x\right) \hat{f^2}(-p-k)u^\mu t_{\mu 0}(p,-k)  $, this (symmetric) element of $ L^2(\bR^6, d^3pd^3k)$ induces a bounded linear map
$$L_{n+2, H_{\Delta}} : J_{n+2}(\cH^{(n+2)}) \ni \Phi \mapsto \Phi_{H_{\Delta}} \in J_n (\cH^{(n)})\:,\quad  \Phi_{H_\Delta} := \int_{\bR^6} d^3pd^3k H_\Delta(\vec{p},\vec{k})\Phi(\vec{p},\vec{k}, \vec{Q})\:.$$
The said operator 
transforms Schwartz functions to Schwartz functions and 
can be written as $L_{n+2, H_{\Delta}}  = \langle H_\Delta |\otimes I : L^2(\bR^{6}, d^3pd^3k)  \otimes L^2(\bR^{3n}, d^{3n}q) \to L^2(\bR^{3n}, d^{3n}q)$,
so that in view of Riesz' lemma $$||L_{n+2, H_{\Delta}} || =||\langle H_\Delta |\otimes I || = || \langle H_\Delta |\:|| \: ||I||=  ||H_\Delta||_{L^2(\bR^6, d^3pd^3k)} <+\infty\:.$$
Looking at (\ref{MO}) we can already conclude that, if the Hilbert space isomorphisms $J_n$ are defined in (\ref{UM}), a bounded operator which implements the considered quadratic form is $ \sA_{\epsilon,2, n+2}(\Delta) : \cH^{(n+2)}\to \cH^{(n)}$ such that, if $n\geq 0$ and $\epsilon> 0$,
$$ \sA_{\epsilon, 2, n+2}(\Delta) = \frac{1}{4\pi}\sqrt{(n+1)(n+2)}J_n^{-1} S_{\epsilon,n} L_{n+2,H_\Delta} D_{\epsilon, n+2} J_{n+2}$$
where $S_{\epsilon,n}$ is the multiplicative operator defined by the function $S_{\epsilon,n}: \vec{Q}\mapsto \frac{1}{\sqrt{E^n_{\epsilon,u}(Q)}}$, bounded by $1/\sqrt{nm + \epsilon}$, and $D_{\epsilon, n+2}$ is the analogous operator defined by the function 
$D_{\epsilon, n+2} :(\vec{p},\vec{k}, \vec{Q})\mapsto  \frac{1}{\sqrt{{E^{n+2}_{\epsilon,u}(p,k, Q)} E(p)E(k) }}$ bounded by $1/(m\sqrt{(n+2)m+\epsilon})$. We therefore have the bound for $\Psi \in \cH^{(n+2)}$ and $\Psi' \in \cH^{(n)}$ with $n\geq 0$ and $\epsilon>0$
$$||\sA_{\epsilon,2, n+2}(\Delta)|| ||\Psi|| ||\Psi'|| \leq  \frac{ ||H_\Delta||}{4\pi m^2 }\frac{\sqrt{(n+1)(n+2)}}{\sqrt{(n+m^{-1}\epsilon)(n+2+m^{-1}\epsilon)}} ||\Psi|| ||\Psi'|| \:.$$
In summary, if $\epsilon>0$ and $n=0,1,\ldots$ we have an $n$-uniform bound
$$||\sA_{\epsilon,2, n+2}(\Delta)|| \leq   C_{\epsilon,\Delta}\:.$$
(Notice that however $C_{\epsilon, \Delta} \to +\infty$ for $\epsilon \to 0^+$.)
Let us interpret each operator $\sA_{\epsilon,2,n}(\Delta)$ as defined on the domain $\cH^{(n)}$ to the whole $\gF_s(\cH_m)$, and define 
$\sA_{\epsilon,2,0}(\Delta)= \sA_{\epsilon,2,1}(\Delta)=0$. Since $Ran(\sA_{\epsilon,2,r}) \perp Ran(\sA_{\epsilon,2,s})$ if $r\neq s$,
with the same procedure as for the case of $T^\Delta_{\epsilon,3}$ based on the two initially proved lemmata, we can conclude that there exists 
 a unique $\sA_{\epsilon,2}(\Delta) \in \gB(\gF_s(\cH_m))$ which extends the family of operators $\sA_{\epsilon, 2, n}(\Delta):  \cH^{(n)} \to \cH^{(n-2)} \subset \gF_s(\cH_m)$, $n=0,1,2, \ldots$ (where $\cH^{(n-2)}=\{0\}$ if $n<2$).  By construction, if $\epsilon>0$,
\beq T^\Delta_{\epsilon,2}(\Psi,\Psi') = \langle \Psi| \sA_{\epsilon,2}(\Delta) \Psi' \rangle \quad   \mbox{for $\Psi, \Psi' \in \gS_0$ and} \quad || \sA_{\epsilon,2}(\Delta) ||  \leq C_{\epsilon,\Delta}\:.   \label{TTWO}\eeq
Differently from the case of $\sA_{\epsilon,3}(\Delta)$, the norm of this operator diverges as $\epsilon \to 0^+$. This is only 
due to the restriction $\sA_{\epsilon,2}(\Delta)\spa\rest_{\cH^{(2)}}= \sA_{\epsilon,2,2}(\Delta):\cH^{(2)}\to \cH^{(0)}$.

We conclude with the analysis of the term $T^\Delta_{\epsilon,1}$ in the expansion (\ref{expand}), i.e. $\sA_{\epsilon,1}(\Delta)$ in (\ref{expand2}), where, again, we explicitly assume that $\Delta$ has finite measure.\beq  T^\Delta_{\epsilon,1}(\Psi,\Psi'):=  \int_\Delta\int_{\bR^6} \sp\sp e^{-i (\vec{p}+\vec{k})\cdot \vec{x}}\widehat{f^2}(p +k)\frac{ \langle W^\epsilon_u\Psi |a^\dagger_{p}a^\dagger_{k}W^\epsilon_u\Psi' \rangle}{4\pi} u^\mu t_{\mu0}(p,-k)\frac{d^3pd^3k}{E(p)E(k)} d^3x \nonumber\:,\eeq
 where, as before, $W^\epsilon_u= (H^u +  \epsilon I)^{-1/2} \in \gB(\gF_s(\cH_m))$.  An elementary procedure based on the observation that $\overline{\hat{f^2}(q)}= \hat{f^2}(-q)$ because $f$ is real, and taking the property 
$ \langle W^\epsilon_u\Psi |a^\dagger_{p}a^\dagger_{k}W^\epsilon_u\Psi' \rangle =
 \overline{\langle W^\epsilon_u\Psi' |a_{p}a_{k}W^\epsilon_u\Psi \rangle}$ into account, proves that the unique wanted operator $\sA_{\epsilon,1}(\Delta) \in \gB(\gF_s(\cH_m))$ such that 
\beq \langle \Psi| \sA_{\epsilon,1}(\Delta) \Psi' \rangle = T^\Delta_{\epsilon,1}(\Psi,\Psi') \quad  \mbox{$\Psi, \Psi' \in \gS_0$}. \label{TTHREE}\eeq
is exactly $\sA_{\epsilon,1}(\Delta) = \sA_{\epsilon,2}(\Delta)^\dagger$.

Item (e) can be finally proved as follows. The operators are bounded from below as an immediate consequence of their definition, the uniform bound from below established in item (e) of Proposition \ref{PROPTTR} and, obviously, the fact that $(H^u_\epsilon)^{-1}\leq \epsilon^{-1}I$. \\

\noindent {\bf Proof of Lemma \ref{LEMMAR}}.
 If $\psi_2 \in \cD_2$ extend by continuity $K(\cdot, \psi_2)$ to the whole $\cH_1$.
Use the Riesz lemma to define $A\psi_2\in \cH_1$, then show that it is $\psi_2$-linear and bounded by $||A\psi_2|| \leq ||\psi_2|| C$, and (uniquely) extend $\cD_2 \ni \psi_2 \mapsto A\psi_2 \in \cH_1$ to the whole $\cH_2$ by linearity and continuity.  \hfill $\Box$\\

\noindent {\bf Proof of Lemma \ref{LEMMAS}}. Define $\cD$ as the subspace of finite linear combinations of elements in the spaces $\cH_j$. $\cD$ is a dense subspace of $\cH$. As a consequence, every $A\in \gB(\cH)$ is fixed by its restriction to $\cD$, and also by the restrictions to the single spaces $\cH_j$, $j\in J$. This proves that, if an operator exists as in the hypothesis, it must be unique.
Now consider $\psi = \sum_j \psi_j  \in \cD$, where $\psi_j \in \cH_j$, and define $A': \cD \to \cH$ as $A'\psi := \sum_j A_j\psi_j$ (the sum being finite by construction). Since the vectors $A_j\psi_j$ are mutually orthogonal, it holds
$||A'\psi||^2 = \sum_j ||A_j\psi_j||^2 \leq \sum_j ||A_j||^2 ||\psi_j||^2 \leq (\sup_{j\in J} ||A_j||)^2 \sum_j ||\psi_j||^2 \leq  (\sup_{j\in J} ||A_j||)^2 ||\psi||^2$, where we also used $||\psi||^2 = \sum_j ||\psi_j||^2$. Therefore it must be $||A'|| \leq \sup_{j\in J} ||A_j|| < +\infty$ and, since $\cD$ is dense in $\cH$, it extends to a unique $A \in \gB(\cH)$ with $||A||= ||A'|| \leq \sup_{j\in J} ||A_j||$. $A$ is the wanted operator. If it were $||A|| =s  <  \sup_{j\in J} ||A_j||$, we would have 
$||A_j||  =  ||A\spa\rest_{\cH_j}|| \leq ||A|| = s$ and thus the contradiction $\sup_{j\in J} ||A_j|| \leq s < \sup_{j\in J} ||A_j||$. 
Now, according to the above orthogonal decompositions, define $A^{(N)}\psi := \sum_{j\leq N} A_j \psi_j$. We have $||A\psi -A^{(N)}\psi||^2= 
\sum_{j>N} ||A_j\psi_j||^2 \leq
\sum_{j>N} ||A||^2 ||\psi_j||^2 \to 0$ if $N\to +\infty$.
 \hfill $\Box$\\

\noindent{\bf Proof of Proposition \ref{COS}}.
We start with the existence and uniqueness proof and (a). Taking a Minkowskian coordinate system where $\Sigma$ is described by $x^0=0$, define $f_\Delta(x^0, \vec{x}):= \int_{\Delta} f^2(x^0, \vec{x}-\vec{y}) d^3y$ with $f\in \cD_\bR(\bR^4)$. Since the Borel set $\Delta\subset \bR^3$ is bounded, we have that $f_\Delta \in \cD_\bR(\bR^4)$. At this point define $\sH_f^{u}(\Delta)$ directly by (\ref{PART}) so that invariance of $\gS_0$ is automatic as well as the symmetry property, since $f_\Delta$ is real and (b) of Proposition \ref{PROP1} holds.
Now, (\ref{HD}) (namely (\ref{PART})) immediately follows from (\ref{TT}) by invoking Fubini's theorem.  
This concludes the existence part of the proof, the uniqueness being obvious since 
 $\Psi$ in (\ref{HD}) varies in the dense domain $\gS_0$. The covariance property in (a) immediately follows from the general covariance property of the stress-energy tensor operator. (f) Proposition \ref{PROP1}, taking (\ref{PART}) into account, proves essential selfadjointness.\\
(b) {If  $\sH_f^{u}(\Delta) =  :\spa \hat{T}_{\mu\nu} \spa: [f_\Delta] u^\mu   u^\nu_\Sigma \neq 0$, then  it  cannot be positive as  immediately follows from 
Proposition \ref{XYZ}.}\\
(c)  (\ref{HD2gen}) is a consequence of the divergence theorem and the conservation equation in (c) of Proposition \ref{PROPTTR}.\\
 (d)
First of all notice that the composition 
$\sH_f^{u}(\Delta) (H_\epsilon^{u})^{-1/2}$ is well defined on $\gS_0$ since $(H_\epsilon^{u})^{-1/2}$ leaves that space invariant.
The stated identity is immediate since the two operators in the thesis coincide on $\gS_0$ by construction.  $\Box$\\

\noindent {\bf Proof of Lemma \ref{LEMMALEMME}}. It is sufficient to apply Eq.(3) in \cite{F12} referring to the explicit computations presented in Sec.2.5 therein, specialized to the Minkowski spacetime and using as reference state the Poincar\'e-invariant state we indicated by $\Omega$. (The only necessary hypothesis is that the two-point function of the considered normalized state vector $\Psi$ is such that $\bM\times \bM \ni (x,y) \mapsto \langle \Psi| \hat{\phi}(x) \hat{\phi}(y)\Psi \rangle - \langle \Omega| \hat{\phi}(x) \hat{\phi}(y) \Omega\rangle  \in C^\infty(\bM\times \bM)$. For $\Psi\in \gS_0$ this fact is guaranteed by Proposition \ref{PROP2X}.)
In this case the partial differential operator $Q$ is the operator 
$u^\mu u'^\nu D_{\mu\nu}$ (\ref{Dmunu}) treated as in the proof of Proposition \ref{TEOB}:
 $u^\mu u'^\nu D_{\mu\nu}= c^2D'_{00} -s^2 D'_{11} $
$$= \frac{c^2-s^2}{2}(Q_0\otimes Q_0 + Q_1\otimes Q_1) + \frac{c^2+s^2}{2}
(Q_2\otimes Q_2 + Q_3\otimes Q_3 + Q_4 \otimes Q_4)$$
($c= \cosh \chi$, $s=\sinh\chi$ for some $\chi\geq 0$ and we have omitted the specifications $_{u,u'}$ and $^{u,u'}$ used in the proof of Proposition \ref{TEOB})
so that it is a positive linear combination of products of real first-order (at most) differential operators as requested. In our concrete case $u= \partial_0$.
With these choices, following the first part of Sec.2.5 in \cite{F12} we have that, where $k= (\omega, \vec{k})$ with $\omega=\sqrt{\vec{k}^2 +m^2}$, 
$$\int \langle \Psi | :\spa T_{\mu\nu}\spa: (x^0,\vec{x})\Psi \rangle u^\mu u'^\nu h'(x^0)^2  dx^0 \geq -\int_0^{+\infty} \frac{d\alpha}{\pi} \int_0^{+\infty} \frac{d^3k}{(2\pi)^3} \frac{2\omega(u'^0\omega - \vec{k}\cdot\vec{u}')}{4\omega} |\hat{h'}(\omega +\alpha)|^2\:.$$
Integration of $\vec{k}\cdot\vec{u}'/\omega$ in the angular variables of $d^3k= k^2 \sin \theta dk d \theta d\phi$ (choosing $z$ parallel to $\vec{u}'$) gives a vanishing contribution. The remaining integration, following the same route as in the first part of Sec.2.5 in \cite{F12}, furnishes
$$\int \langle \Psi | :\spa T_{\mu\nu}\spa: (x^0,\vec{x})\Psi \rangle  u^\mu u'^\nu h'(x^0)^2  dx^0 \geq -u'^0 \frac{1}{16 \pi^3} \int_{m}^{+\infty} 
|\hat{h'}(s)|^2 s^4 Q_3(s) ds\:,$$
that is our thesis since $u' = (u'^0, \vec{u}')$ with $u'^0= \sqrt{1+ \vec{u}'^2}$. $\Box$\\

\noindent{\bf Proof of Lemma \ref{F2}}.
(a) 
Since $A$ is essentially selfadjoint and non-negative, $\overline{A}$ is a non-negative selfadjoint operator. Hence its positive
square root $\sqrt{\overline{A}}$ is well defined, with domain
$D(\sqrt{\overline{A}})$.
We claim that
$
\overline{\operatorname{Ran}(\sqrt{\overline{A}})}
=
\overline{\sqrt{\overline{A}}(D(A))}.
$
First, since
$
D(A)\subset D(\overline{A})\subset D(\sqrt{\overline{A}}),
$
we immediately have
$
\sqrt{\overline{A}}(D(A))\subset \operatorname{Ran}(\sqrt{\overline{A}}),
$
and therefore
$
\overline{\sqrt{\overline{A}}(D(A))}
\subset
\overline{\operatorname{Ran}(\sqrt{\overline{A}})}.
$
Conversely, since $A$ is non-negative and essentially self-adjoint, the
closure $\overline{A}$ coincides with the Friedrichs extension of $A$.
Equivalently, $D(A)$ is a form core for $\sqrt{\overline{A}}$. Thus $D(A)$ is
dense in $D(\sqrt{\overline{A}})$ with respect to the form norm
$
\|u\|_{\sqrt{\overline{A}}}^2
:=
\|u\|^2+\|\sqrt{\overline{A}}u\|^2.
$
Therefore, for every $u\in D(\sqrt{\overline{A}})$, there exists a sequence
$(u_n)_{n\in\mathbb N}\subset D(A)$ such that
$
u_n\to u
$
and
$
\sqrt{\overline{A}}u_n\to \sqrt{\overline{A}}u
$
in the Hilbert space norm. Hence
$
\sqrt{\overline{A}}u\in \overline{\sqrt{\overline{A}}(D(A))}.
$
Since $u\in D(\sqrt{\overline{A}})$ was arbitrary, we obtain
$
\operatorname{Ran}(\sqrt{\overline{A}})
\subset
\overline{\sqrt{\overline{A}}(D(A))}.
$
Taking closures gives
$
\overline{\operatorname{Ran}(\sqrt{\overline{A}})}
\subset
\overline{\sqrt{\overline{A}}(D(A))}.
$
Combining the two inclusions, we conclude that
$
\overline{\operatorname{Ran}(\sqrt{\overline{A}})}
=
\overline{\sqrt{\overline{A}}(D(A))}
$.
At this point, using the fact that $\sqrt{\overline{A}}$ is selfadjoint, we have $\overline{\sqrt{\overline{A}}(D(A))} = \overline{Ran(\sqrt{\overline{A}})} = Ker(\sqrt{\overline{A}})^\perp = Ker(\overline{A})^\perp$, the last identity arising from $Ker(B) = Ker(\sqrt{B})$ which holds if $B^\dagger =B\geq 0$ by spectral calculus using $D(B)\subset D(\sqrt{B})$.\\
(b) Since $A-cI\geq 0$ on $D(A)$, its closure $\overline{A-cI}=\overline{A}-cI$ is non-negative. Hence $\sigma(\overline{A}) \geq c >0$ by spectral calculus.
In particular $Ker(\overline{A})=\{0\}$, so that $\overline{\sqrt{\overline{A}}(D(A))} = \overline{Ran(\sqrt{\overline{A}})} = Ker(\sqrt{\overline{A}})^\perp = Ker(\overline{A})^\perp= \cH$, but also $0\in \rho(\overline{A}^\alpha)$ for $\alpha \geq 0$ and the considered operators are closed since they are selfadjoint, so that
$Ran(\overline{A}^\alpha) = \overline{Ran(\overline{A}^\alpha)} = Ker(\overline{A}^\alpha)^\perp = \cH$, in particular, for $\alpha=1/2$ the identity in (b) follows. Again, since $0\in \rho(\overline{A}^\alpha)$
it holds $\overline{A}^{-\alpha} \in \gB(\cH)$.
$\Box$\\

\noindent {\bf Proof of Lemma \ref{LEMMAZd}}.  Observing that $D(A^0) = D(A)  \subset D(\overline{A^0})\cap D(\overline{A}) \subset D(\sqrt{\overline{A^0}})$,
the only item to be proved is the first inequality in (a). Next (b) is a straightforward consequence of (a), defining the operator $1/\sqrt{\overline{A^0}}$ via functional calculus. (c) follows from Lemma \ref{F2}; in particular the last sentence is a trivial consequence of the previous part.  Let us prove the first inequality in (a).   Take $c>0$ and define the operators $A^0+cI \geq 0$ and $A+cI \geq 0$, and consider the associated forms
 $a(x) := \langle x|(A+cI)x\rangle $ and $b(x) := \langle x|(A^0+cI)x\rangle$. By construction, every Cauchy sequence in the form norm $||\cdot||_b$ is Cauchy in the form norm $||\cdot||_a$. As a consequence, passing to the form completions, which are respectively the closed forms of the Friedrichs extensions of the considered operators  $(A+cI)_F = A_F+cI= \overline{A}+cI$ and $(A^0+cI)_F =(A^0)_F+cI = \overline{A^0}+cI$ where, for the first identities  we have used  known properties of Friedrichs extensions \cite{FE}, and for the second identities the fact that the operators are essentially selfadjoint, so their closures coincide with their Friedrichs extensions. Hence
 we have $D(\sqrt{\overline{A^0}+cI})\subset D(\sqrt{\overline{A}+cI})$. It immediately follows from spectral calculus that $x \in D(\sqrt{\overline{A^0}})$ implies 
$x \in D(\sqrt{\overline{A^0}+ cI})$.
If $x \in D(\sqrt{\overline{A^0}+ cI})$ there is a $||\cdot||_b$-Cauchy sequence $D(A)\ni x_n \to x$ such that  $\langle x_n | (A^0+cI) x_n\rangle \to \langle \sqrt{\overline{A^0}+ cI} x| \sqrt{\overline{A^0}+ cI} x\rangle$.
This sequence is also Cauchy for the other norm $||\cdot||_a$ and $\langle x_n | (A+cI) x_n\rangle \to \langle \sqrt{\overline{A}+ cI} x| \sqrt{\overline{A}+ cI} x\rangle$. Since $\langle x_n | (A^0+cI) x_n\rangle \geq 
\langle x_n | (A+cI) x_n\rangle$ we conclude that $\langle \sqrt{\overline{A^0}+ cI} x| \sqrt{\overline{A^0}+ cI} x\rangle \geq \langle \sqrt{\overline{A}+ cI} x| \sqrt{\overline{A}+ cI} x\rangle$.
Passing to the spectral representations $\int_{\bR^+} (\lambda +c) \langle x|P^0(d\lambda)x\rangle \geq \int_{\bR^+} (\lambda +c) \langle x|P(d\lambda)x\rangle$ is valid for every $c>0$, so that $\int_{\bR^+} \lambda \langle x|P^0(d\lambda)x\rangle \geq \int_{\bR^+} \lambda  \langle x|P(d\lambda)x\rangle$ which means 
$\langle \sqrt{\overline{A^0}}x |  \sqrt{\overline{A^0}}x \rangle \geq \langle \sqrt{\overline{A}}x |  \sqrt{\overline{A}}x \rangle$.  By hypothesis this is valid if $x \in D(\sqrt{\overline{A^0}})$ as wanted. $\Box$\\
 
    \end{document}